\documentclass[11pt,a4paper]{article}

\usepackage[utf8]{inputenc}
\usepackage[T1]{fontenc}
\usepackage[english]{babel}
\usepackage{amsmath,amsthm,amssymb,dsfont,paralist,todonotes}

\usepackage{graphicx}
\usepackage[lighttt]{lmodern} 
\usepackage{indentfirst}
\usepackage[normalem]{ulem}
\usepackage{enumitem}
\usepackage{tikz}

\usetikzlibrary{external,arrows,positioning,shapes,decorations.pathmorphing,
automata,backgrounds,petri,fit,calc,matrix}\usetikzlibrary{
decorations.pathreplacing }
\usepackage[format=plain,margin={15pt,15pt}, font=footnotesize,
labelfont=bf]{caption}
\DeclareCaptionLabelFormat{bf-parens}{\textbf{#1}}

\usepackage{listings}
{\theoremstyle{definition}
\newtheorem{df}{Definition}[section]
\newtheorem{ex}[df]{Examples}}
\newtheorem{lm}[df]{Lemma}

\newtheorem{theo}[df]{Theorem}

\newtheorem{newob}[df]{New Observation}

\newcommand{\Ain}{A_{\mathrm{in}}}
\newcommand{\eps}{\varepsilon}
\newcommand{\BEP}{\ensuremath{\mathrm{BE}(P)}}

\DeclareMathOperator{\acc}{acc}
\DeclareMathOperator{\Acc}{Acc}
\DeclareMathOperator{\CFext}{CF_{ext}}
\DeclareMathOperator{\DCFext}{DCF_{ext}}
\DeclareMathOperator{\llabel}{label}
\DeclareMathOperator{\son}{son}
\DeclareMathOperator{\up}{up}
\DeclareMathOperator{\down}{down}
\DeclareMathOperator{\rroot}{root}
\DeclareMathOperator{\sel}{sel}
\DeclareMathOperator{\expand}{expand}
\DeclareMathOperator{\test}{test}
\DeclareMathOperator{\bottom}{bottom}
\DeclareMathOperator{\DCF}{DCF}
\DeclareMathOperator{\rk}{rank}
\DeclareMathOperator{\Pd}{P}
\DeclareMathOperator{\p}{p}
\DeclareMathOperator{\TP}{TP}
\DeclareMathOperator{\SPACE}{SPACE}
\DeclareMathOperator{\DREG}{DREG}
\DeclareMathOperator{\DeREG}{D_eREG}
\DeclareMathOperator{\DfREG}{D_fREG}
\DeclareMathOperator{\CF}{CF}
\DeclareMathOperator{\REG}{REG}
\DeclareMathOperator{\RT}{RT}
\DeclareMathOperator{\DRT}{DRT}
\DeclareMathOperator{\DeRT}{D_eRT}
\DeclareMathOperator{\DfRT}{D_fRT}
\DeclareMathOperator{\dec}{dec}
\DeclareMathOperator{\dom}{dom}
\DeclareMathOperator{\yield}{yield}
\DeclareMathOperator{\treeD}{tree_\Delta}
\DeclareMathOperator{\tree}{tree}
\DeclareMathOperator{\false}{false}
\DeclareMathOperator{\Integer}{Integer}
\DeclareMathOperator{\id}{id}
\DeclareMathOperator{\en}{en}
\DeclareMathOperator{\N}{\mathbb{N}}
\DeclareMathOperator{\nnull}{null}
\DeclareMathOperator{\on}{\ on\ }
\DeclareMathOperator{\pop}{pop}
\DeclareMathOperator{\push}{push}
\DeclareMathOperator{\stay}{stay}
\DeclareMathOperator{\first}{first}
\DeclareMathOperator{\empt}{empty}
\DeclareMathOperator{\rread}{read}
\DeclareMathOperator{\ran}{ran}
\DeclareMathOperator{\ruleif}{\uline{\mathrm{if}}\ }
\DeclareMathOperator{\ruleelse}{\ \uline{\mathrm{else}}\ }
\DeclareMathOperator{\rulethen}{\ \uline{\mathrm{then}}\ }

\DeclareMathOperator{\true}{true}
\DeclareMathOperator{\ttop}{top}

\newcommand{\LA}{\mathrm{LA}}
\newcommand{\alphaCF}{\alpha$-$\CF}
\newcommand{\alphaCFext}{\alpha$-$\CFext}
\newcommand{\alphaDCF}{\alpha$-$\DCF}
\newcommand{\alphaREG}{\alpha$-$\REG}
\newcommand{\alphaDREG}{\alpha$-$\DREG}
\newcommand{\alphaRT}{\alpha$-$\RT}

\newcommand{\G}{G}
\newcommand{\lambdaCF}{\lambda$-$\CF}
\newcommand{\lambdaCFext}{\lambda$-$\CFext}
\newcommand{\lambdaDCF}{\lambda$-$\DCF}
\newcommand{\lambdaREG}{\lambda$-$\REG}
\newcommand{\lambdaDREG}{\lambda$-$\DREG}
\newcommand{\lambdaRT}{\lambda$-$\RT}

\newcommand{\lambdaDeREG}{\lambda$-$\DeREG}
\newcommand{\lambdaDfREG}{\lambda$-$\DfREG}
\newcommand{\lambdaDeRT}{\lambda$-$\DeRT}
\newcommand{\lambdaDfRT}{\lambda$-$\DfRT}
\newcommand{\tauCF}{\tau$-$\CF}
\newcommand{\tauCFext}{\tau$-$\CFext}
\newcommand{\tauDCF}{\tau$-$\DCF}
\newcommand{\tauDCFext}{\tau$-$\DCFext}
\newcommand{\tauREG}{\tau$-$\REG}
\newcommand{\tauDREG}{\tau$-$\DREG}
\newcommand{\tauRT}{\tau$-$\RT}
\newcommand{\tauDRT}{\tau$-$\DRT}

\newcommand*{\QEDB}{\hfill\ensuremath{\square}}%

\newcommand{\bool}[1]{\underline{#1}}

\newcommand{\cond}{\underline{cond} }

\newcommand{\ar}{\ensuremath{\quad\ \rightarrow \quad\ }}
\newcommand{\Ar}{\ensuremath{\quad\ \Rightarrow \quad\ }}

\newcounter{boxwidth}
\newcommand{\emp}[1]{\textit{#1}}

\newcommand{\sect}[1]{#1}

\begin{document}

\begin{center}

{\LARGE Context-Free Grammars with Storage}

\

\

\begin{tabular}{cc}
  {\large Joost Engelfriet}\\[3mm]
  {\small {\it Leiden Institute of Advanced Computer Science (LIACS)}}\\[1mm] 
  {\small {\it Leiden University, P.O. Box 9512}}\\[1mm]
  {\small {\it 2300 RA Leiden, The Netherlands}}\\[1mm] 
  {\small \tt engelfri@liacs.nl}
\end{tabular}

\

July 1986/2014 
\end{center}

\vspace{18mm}


\noindent\textbf{Abstract.}
Context-free $S$ grammars are introduced, for arbitrary
(storage) type $S$, as a uniform framework for recursion-based
grammars, automata, and transducers, viewed as programs. To each
occurrence of a nonterminal of a context-free $S$ grammar an
object of type $S$ is associated, that can be acted upon by tests
and operations, as indicated in the rules of the grammar. Taking
particular storage types gives particular formalisms, such as
indexed grammars, top-down tree transducers, attribute grammars,
etc. Context-free $S$ grammars are equivalent to pushdown $S$
automata. The context-free $S$ languages can be obtained from the
deterministic one-way $S$ automaton languages by way of the delta
operations on languages, introduced in this paper.


\vspace{1cm}

\section*{\sect{Foreword}}

This is a slightly revised version of a paper that appeared 
as Technical Report 86-11 of the Department of Computer Science
of the University of Leiden, in July 1986.
Small errors are corrected and large mistakes are repaired.
Here and there the wording of the text is improved.
The references to the literature are updated, 
and a few New Observations are added.
However, I have made no effort to bring the paper up-to-date.  
New references to the literature are indicated by a *.
The results of Section~8 were published in~\cite{*Eng11}. 
 
I thank Heiko Vogler for his idea to transform the paper into \LaTeX\ 
and put it on arXiv. I~am grateful to Luisa Herrmann for 
typing the text in \LaTeX\ and drawing the figures in Ti{\it k}Z. 


\newpage
\section*{\sect{Introduction}}

Context-free grammars, as formally defined by Chomsky,
are a very particular type of rewriting system. However, the
reason for their popularity is that they embody the idea of
recursion, in its simplest form. A context-free grammar is
really just a nondeterministic recursive program that generates
or recognizes strings. For instance, the context-free grammar
rules $A \rightarrow aBbCD$ and $A \rightarrow b$ can be understood
as the following program piece (in case the grammar is viewed as a string generator)

\begin{lstlisting}
procedure A;
   begin
     write(a); call B; write(b); call C; call D
   end
 or
   begin
     write(b)
   end
\end{lstlisting}
where ``{\ttfamily write(a)}'' means ``write {\ttfamily a} on the output tape''.
If the grammar is viewed as a (nondeterministic) string recognizer or
acceptor, as is usual in recursive descent parsing, ``{\ttfamily write}'' 
should be replaced by ``{\ttfamily read}'' (where ``{\ttfamily read(a)}'' means 
``read {\ttfamily a} from the input tape''). Thus, context-free grammars are recursive
programs; the nonterminals $A, B, \dots$  are procedures (without
parameters), and all the rules with left-hand side $A$ constitute
the body of procedure $A$; the main program consists of a call of
the main procedure, i.e., the initial nonterminal. Actually, it
is quite funny that in formal language theory some programs are
called grammars (as suggested here), and other programs are
called automata: e.g.,  a program with one variable, of type
pushdown, is called a pushdown automaton. Maybe the underlying
idea is: nonrecursive program = automaton, recursive program =
grammar.

Given that context-free grammars consist of recursive
procedures without parameters, what would happen if we
generalize the concept by allowing parameters? This is an idea
that has turned up in several places in the literature. Since
the resulting formalisms are usually still based on the idea of
recursive procedures, such context-free grammars with parameters
are easy to understand, construct, and prove correct, just as
ordinary context-free grammars. Here, we will fix this idea as
follows: each nonterminal of the context-free grammar will have
\emp{one} (input) \emp{parameter} of a given type. Although this is a very
simple case of the general idea, we will show that the resulting
formalism of ``context-free $S$ grammars'' (where $S$ is the type of
the parameter) has its links with several existing formalisms;
this is due partly to the fact that such a generalized
context-free grammar may be viewed as a grammar, an automaton, a
program, or a transducer, and partly to the freedom in the
choice of $S$. In this way we will see that the following
formalisms can be ``explained'', each formalism corresponding to
the context-free $S$ grammars for a specific type $S$: indexed grammars, 
top-down tree transducers, ETOL systems, attribute grammars, macro grammars, etc. 
Moreover, viewing the type $S$ as the storage type of an automaton (e.g., $S$ = pushdown,
$S$ = counter, etc.) context-free $S$ grammars can be used to model
all one-way $S$ automata, where $S$ is any type of storage, and all
alternating $S$ automata, where $S$ is any type of [storage plus
input]. In particular, it should be clear that right-linear $S$
grammars correspond to one-way $S$ automata (just as, classically,
right-linear grammars correspond to one-way finite automata). To
stress this link to automata, context-free $S$ grammars are also
called ``grammars with storage'' or even ``recursive automata'', and
the type $S$ is also called a storage type. Thus this programming
paradigm strengthens the similarities between automata and
grammars. Tree grammars and tree automata (such as the top-down
tree transducer) can be obtained either by defining an
appropriate type $S$ of trees, or by considering context-free $S$
grammars that generate trees, or both. Note that trees, in their
intuitive form of expressions, are particular strings.

Two main results are the following.

(1) Context-free $S$ grammars correspond to pushdown $S$
automata. A pushdown $S$ automaton, introduced in \cite{Gre}, is an
iterative program that manipulates a pushdown of which each
pushdown cell contains a symbol and an object of type $S$. This is
of course the obvious way to implement recursive procedures with
one parameter, and so the result is not surprising. What is nice
about it, however, is that it provides pushdown-like automata,
in one stroke, for all formalisms that can be explained as
context-free $S$ grammars (e.g., we will obtain pushdown$^2$ automata
for indexed grammars, tree-walking pushdown transducers for
top-down tree transducers, 
checking-stack/pushdown automata for ETOL systems,
etc.). To deal with determinism (of the
context-free $S$ grammar, viewed as a transducer) the notion of
look-ahead (on the storage $S$) is introduced, as a generalization
of both look-ahead on the input (in parsing, and for top-down
tree transducers) and the predicting machines of \cite{HopUll}.

(2) Apart from this automaton characterization of
context-free $S$ grammars we also give a characterization by means
of operations on languages (cf. AFL/AFA theory \cite{Gin}), but only
for rather specific $S$ (including iterated pushdowns). We define
a new class $\delta$ of ``delta''  operations on languages, such that the
languages generated by context-free $S$ grammars can be obtained
by applying the delta operations to the languages accepted by
deterministic(!) one-way $S$ automata. A delta operation is quite
different from the usual operations on languages; it takes a
(string) language, views the strings as paths through labeled
trees, constructs a tree language out of these paths, and then
produces a (string) language again by taking the yields of these
trees. As an example, the indexed languages are obtained by the
delta operations from the deterministic context-free languages.

Thus, the aim of this paper is to provide \emp{a uniform
framework for grammars and automata that are based on recursion},
including the usual one-way automata as special case. The
general theory that is built in this framework should give
transparent (and possibly easier) proofs of results for
particular formalisms, i.e., for particular $S$. In fact, this
paper may be viewed as an extension of abstract automata theory
to recursive automata, i.e., as a new branch of AFA/AFL theory
\cite{Gin, Gre}. The above two main results constitute a modest
beginning of such a general theory for context-free $S$ grammars;
more can be found in \mbox{\cite{EngVog2, EngVog3, EngVog4, DamGue,
Vog1, Vog2, Vog3, Eng9}}. Although these papers are based on
the first, very rough, version of this paper (\cite{Eng8}),
we will now feel free to mention results from them.

\subsection*{\sect{How to read this paper}}

Since rather many formalisms will be discussed in this
paper, the reader is not expected to know them all. However,
this paper is not a tutorial, and so, whenever a formalism is
discussed, it is assumed that the reader is more or less
familiar with it. If he is not, he is therefore advised to skip
that part of the paper, or read it with the above in his mind.\footnote{{\bf New Observation.}
The reader is asked to consider the word ``he'' to stand for ``he/she'', and
the word ``his'' for ``his/her''. 
My personal ideal is to remove all female forms of words from the language,
and to let the male forms refer to all human beings. 
}
Hopefully the paper is written in such a way that skipping is
easy.

The reader who is interested in the expressiveness of
the context-free $S$ grammar formalism only, should read
Sections \ref{sect1}, \ref{sect2}, \ref{sect3}, \ref{sect4}, and \ref{sect6} 
(after glancing at Section \ref{sect5}), or
parts of them. The reader who is interested in the theory of
context-free $S$ grammars only, can restrict himself to
Sections \ref{sect1.1}, \ref{sect5}, \ref{sect7}, and \ref{sect8}.

\subsection*{\sect{Organization of the paper}}

Context-free $S$ grammars are defined in Section \ref{sect1.1}.
They are compared with attribute grammars in Section \ref{sect1.2}, which
can be skipped without problems. In Section \ref{sect2} two particular
cases are defined: regular grammars and regular tree grammars,
both with storage $S$. It is argued that these correspond to
one-way automata and top-down tree automata, both with storage
$S$, respectively. The reader is advised to at least glance
through this section. Section~\ref{sect3} is divided into 8 parts;
in each part a specific storage type $S$ is defined (e.g.,
$S$ = Pushdown in the second part), and it is shown how the
resulting context-free $S$ grammars relate to existing formalisms.
Although these parts are not completely independent, it
should be easy to skip some of them. The relationship between
context-free $S$ grammars and alternating automata is contained in
Section \ref{sect4}, which can easily be skipped (it is needed in
Section \ref{sect6}(9) only). In Section~\ref{sect5} we start 
the theory and show the first main result mentioned above: the relation to pushdown
$S$ automata. Then, in Section \ref{sect6}, it is shown how this gives
pushdown-like automata for all the formalisms discussed in
Section \ref{sect3}. Thus, Section \ref{sect6} is divided into the same 8
parts as Section \ref{sect3}, according to the storage type, with
one additional part concerning alternating automata.
Section \ref{sect7} is a technical section devoted to determinism, as
needed for Section \ref{sect8}. Section \ref{sect8} contains the second
main result mentioned above: the characterization of context-free $S$ grammars
by means of the delta operations.

\subsection*{\sect{Notation}}

Before we start, we mention some elementary notation.
We assume the reader to be familiar with formal language theory,
see \cite{HopUll, Sal, Har, Ber}, and, to a much lesser extent, with
tree language theory, see \cite{GecSte, Eng1}. We denote by REG, CF,
DCF, Indexed, and RE, the classes of regular, context-free,
deterministic context-free, indexed, and recursively enumerable
languages, respectively. (RT denotes the class of regular tree
languages, also called recognizable tree languages.)

For a set $A$, $A^*$ is the set of strings over $A$. For $w \in A^*$, $|w|$
denotes the length of $w$. The empty string is denoted $\lambda$, and
$A^+ = A^*  - \{\lambda\}$. (In ranked alphabets, $\eps$ is a symbol
of rank 0 denoting $\lambda$, in the sense that $\yield(\eps) = \lambda$.)  

For a relation $\tau$, $\tau^*$ is its reflexive, transitive closure,
$\dom( \tau)$ is its domain, and $\ran(\tau)$ is its range. For a set $A$,
$\id(A)$ denotes the identity mapping $A \rightarrow A$. An ordered pair 
$(\varphi,\psi)$ of objects $\varphi$ and $\psi$ will also be denoted $\varphi(\psi)$, 
not to be confused with function application. For sets $\Phi$ and $\Psi$, 
both $\Phi \times \Psi$ and $\Phi(\Psi)$ will be used to denote their cartesian 
product $\{\varphi(\psi) \mid \varphi \in \Phi, \psi \in \Psi\}$.


\section{\sect{Context-free $S$ grammars}}\label{sect1}

\subsection{\sect{Examples and definitions}}\label{sect1.1}

To give the reader an idea of what context-free $S$
grammars are, let us first discuss three simple examples.

In the first grammar, $\G_1$, there is just one nonterminal
$A$, with one parameter of type integer (i.e., $S = \Integer$), and
there is one terminal symbol $a$. The two rules of the grammar are
\[
  \begin{array}{l}
    A(x) \ar \ruleif x \not= 0 \rulethen A(x-1) A(x-1)\\
    A(x) \ar \ruleif x=0 \rulethen a
  \end{array}
\]
where $x$ is a formal parameter. The meaning of the first rule is
that, for any integer $n$, $A(n)$ may be rewritten as $A(n-1) A(n-1)$,
provided $n \not= 0$; and similarly for the second rule: $A(0)$ may be
rewritten by $a$. Thus, for $n \ge 0$, $A(n)$ generates
$a^{2^n}$.
We may view $\G_1$

\begin{enumerate}
  \item[(1)] as a grammar generating the language 
    $L(\G_1) = \{a^{2^n} \mid n \ge 0\}$ (the input $n$ is chosen
    nondeterministically),
  \item[(2)] as a nondeterministic acceptor that recognizes $L(\G_1)$ (cf.
    the Introduction; $n$ is again chosen nondeterministically),
  \item[(3)] as a deterministic transducer that translates $n$ into $a^{2^n}$, and, finally,
  \item[(4)] as a deterministic acceptor of all integers $n \ge 0$ (the
    domain of the translation).
\end{enumerate}
These four points of view will be taken for all context-free
$S$ grammars. The grammar $\G_1$ is not such a good example for the 4th point of
view; a better example will be given in Section~\ref{sect4} (viz., $\G_6$).
From all four points of view, $\G_1$ (and any other context-free $S$
grammar) can be thought of as a program, similar to the one for
the context-free grammar in the Introduction. As a transducer,
$\G_1$ corresponds intuitively to the program

\begin{lstlisting}[mathescape]
procedure A(x: integer);
    begin if x $ \neq $ 0
             then call A(x-l); call A(x-1)
             else write(a)
          fi
    end;
{main program}
obtain n;
call A(n)
\end{lstlisting}
where ``{\ttfamily obtain n}'' means ``read {\ttfamily n} from an input device''. 
If $\G_1$ is a generator of $L(\G_1)$, ``{\ttfamily obtain~n}'' means 
``choose an integer {\ttfamily n}''. 
If $\G_1$ is an acceptor of $L(\G_1)$, then ``{\ttfamily obtain n}'' again means 
``choose an integer {\ttfamily n}'', and, in the program, ``{\ttfamily write(a)}'' should be
replaced by ``{\ttfamily read(a)}'', as observed for the context-free grammar
in the Introduction. If $\G_1$ is an acceptor of the nonnegative
integers, then ``{\ttfamily write(a)}'' can be replaced by ``skip'' (i.e.,
terminals do not matter).

The second grammar, $\G_2$, generates the language
$\{a^nb^nc^n  \mid n \ge 0\}$. It has nonterminals $\Ain, A,B,C$,  each having a
parameter of type pushdown (i.e., $S$ = Pushdown); the pushdown
symbols are $a$ and $\#$ (the bottom marker). The derivations of $\G_2$
start with $\Ain(\#)$: the initial call of the main procedure. The
rules of the grammar are the following, where the tests and
operations on the pushdown should be obvious (and $\lambda$ denotes the
empty string).

\begin{center}\parbox{0cm}{
\begin{tabbing}
    \=$\Ain(x) \ \ $\=$\ar $\=$A(x)$\\
    \>$A(x) $\>$\ar $\>$aA(\push\ a \on x)$\\
    \>$A(x) $\>$\ar $\>$B(x)C(x)$\\
    \>$B(x) $\>$\ar $\>$\ruleif \ttop(x) = a \rulethen bB(\pop\ x)$\\
    \>$B(x) $\>$\ar $\>$\ruleif \ttop(x) = \# \rulethen \lambda$\\
    \>$C(x) $\>$\ar $\>$\ruleif \ttop(x) = a \rulethen cC(\pop\ x)$\\
    \>$C(x) $\>$\ar $\>$\ruleif \ttop(x) = \# \rulethen \lambda$
\end{tabbing}}
\end{center}

\noindent The unconditional rules may always be used, the conditional
rules only if their tests are true. From $\Ain(\#)$, $\G_2$ generates
nondeterministically $a^nA(a^n\#)$, followed by $a^nB(a^n\#)C(a^n\#)$, where
the top of the pushdown is at the left. Then $B(a^n\#)$ and $C(a^n\#)$
generate deterministically $b^n$ and $c^n$, respectively. Thus, $\Ain(\#)$ 
generates all strings $a^nb^nc^n$. Note that dropping $C(x)$ from the
third rule gives a right-linear grammar (with pushdown
parameter), generating $a^nb^n$; clearly, as an acceptor of 
$\{a^nb^n \mid n\ge0\}$, this grammar is just an ordinary nondeterministic
pushdown automaton.

Another way to generate the language $\{a^nb^nc^n \mid n\ge 0\}$ is by 
a grammar $\G_3$ that is almost the same as $\G_2$. The first three rules 
of $\G_2$ should be replaced by

\begin{center}\parbox{0cm}{
  \begin{tabbing}
      \=$\Ain(x) $\=$\ar $ \=$A(x)B(x)C(x)$\\
      \>$A(x) $\>$\ar $\>$\ruleif \ttop(x) = a \rulethen aA(\pop x)$\\
      \>$A(x) $\>$\ar $\>$\ruleif \ttop(x) = \# \rulethen \lambda$
  \end{tabbing}}
\end{center}

\noindent But now the idea is that any integer $n \ge 0$ can be taken as input,
encoded as a pushdown $a^n\#$. The grammar $\G_3$ starts the derivation with
$\Ain(a^n\#)$, and deterministically generates $a^nb^nc^n$. Thus $\G_3$
translates $n$ into $a^nb^nc^n$, and generates (or accepts) the
language $\{a^nb^nc^n \mid  n \ge 0\}$. 

\vspace{1cm}
Let us now turn to the formal definitions. They are
inspired by Ginsburg and Greibach \cite{Gin}, who developed a general
theory of automata that is based on the separation of storage
and control. We will try to make our notation as readable as
that of Scott, who also started such a theory \cite{Sco}. We begin
with the definition of type, i.e., of the possible types of the
parameter. Because of the intuitive connection to automata we
will talk about storage type rather than type. A storage type is
a set of objects (called the storage configurations), together
with the allowed tests and operations on these objects. Since
each nonterminal has only one parameter, we may restrict
ourselves to unary tests and operations. But that is not all.
Since our context-free $S$ grammars will also be viewed as
transducers, it is necessary to specify a set of input elements,
together with the possibility to encode them as storage
configurations (e.g., in $\G_3$, integer $n$ is encoded as pushdown
$a^n\#$). But, in general, different transducers may use different
encodings (e.g., we should also have the freedom to encode $n$ as
$b^n\bowtie$). Thus, a set of possible encodings is specified.

\begin{df} 
    A \emp{storage type} $S$ is a tuple $S = (C,P,F,I,E,m)$,
    where $C$ is the set of configurations, $P$ is the set of predicate
    symbols, $F$ is the set of instruction symbols, $I$ is the set of
    input elements, $E$ is the set of encoding symbols, and $m$ is the
    meaning function that associates with every $p \in P$ a mapping
    $m(p): C \rightarrow \{\true,\false\}$, with every $f \in F$ a partial function
    $m(f): C \rightarrow  C$, and with every $e \in E$ a partial function
    $m(e): I \rightarrow C$. \QEDB
\end{df}

We let $\BEP$ denote the set of all boolean expressions over $P$, with the  
usual boolean operators \bool{and}, \bool{or}, \bool{not}, \bool{true}, and \bool{false}. 
For $b \in \BEP$, $m(b): C \rightarrow \{\true,\false\}$ is defined in the obvious
way. The elements of $\BEP$ are also called \emp{tests}.

We will also say ``predicate $p$'' instead of ``predicate symbol $p$'', 
with the intention to talk about $p$ and $m(p)$ at the
same time (when the distinction is not so important), and
similarly for ``instruction'' and ``encoding''.

Next we give our main definition: that of context-free
$S$ grammar, for any storage type $S$. However, to remain as general
as possible, we will call it a \emp{context-free $S$ transducer} (but
also grammar and acceptor, depending on the point of view).

First a remark on the notation of rules. Since all nonterminals, predicate symbols,
and instruction symbols always have one formal argument, we will drop ``$(x)$'' from
our formal notation. Thus, the rules of $\G_1$ can first be written as

\begin{center}\parbox{0cm}{
  \begin{tabbing}
      \=$A(x) $\=$ \ar \ruleif \text{\bool{not}} \nnull(x) \rulethen A(\dec(x))A(\dec(x))$\\
      \>$A(x) $\>$\ar \ruleif \nnull(x) \rulethen a $
  \end{tabbing}}
\end{center}

\noindent  where $\nnull(x)$ and $\dec(x)$ stand for $x = 0$ and $x-1$, respectively
(null is a predicate symbol, and dec is an instruction symbol of
the storage type $\Integer$). And then they can be written, formally, as

\begin{center}\parbox{0cm}{
  \begin{tabbing}
      \=$A $\=$ \ar \ruleif \text{\bool{not}}  \nnull \rulethen A(\dec)A(\dec)$\\
      \>$A $\>$\ar \ruleif \nnull \rulethen a $
  \end{tabbing}}
\end{center}

\noindent The definition now follows. Recall that, for objects $\varphi$ and $\psi$, 
$\varphi(\psi)$ is just another notation for the ordered pair $(\varphi,\psi)$;
similarly, $\Phi(\Psi)$ is another notation for $\Phi \times \Psi$. 
This is done to formalize $A(\dec)$ as an ordered pair $(A,\dec)$, but keep the old
notation.

\begin{df}
    Let $S = (C,P,F,I,E,m)$ be a storage type. A
    \emp{context-free $S$ transducer}, or $\CF(S)$ transducer, is a tuple $\G =
    (N,e,\Delta,\Ain,R)$, where $N$ is the nonterminal alphabet, $e \in E$ is the
    encoding symbol, $\Delta$ is the terminal alphabet (disjoint with $N$), 
    $\Ain \in N$ is the initial nonterminal, and $R$ is the finite set of
    rules; every rule is of the form
	\[
	A \ar \ruleif b \rulethen \xi
	\]
    with $A \in N$, $b \in \BEP$, and $\xi \in  (N(F) \cup \Delta)^*$. 

    The set of
    total configurations (or, instantaneous descriptions) is $(N(C) \cup  \Delta)^*$.
    The derivation relation of $\G$, denoted by $\Rightarrow_\G$ or just $\Rightarrow$,  
    is a binary relation on the set of total configurations, defined as
    follows: 

    \noindent if $A \rightarrow \ruleif b \rulethen \xi$ is in $R$, $m(b)(c)=\true$, 
    and $m(f)(c)$ is defined for all $f$ that occur in $\xi$, then 
    $\xi_1 A(c) \xi_2 \Rightarrow_\G \xi_1 \xi' \xi_2$ for all total configurations
    $\xi_1$ and $\xi_2$, where $\xi'$ is obtained
    from $\xi$  by replacing every $B(f)$ by $B(m(f)(c))$.

    The \emp{translation defined} by $\G$ is 
	\[
	T(\G) =
	\{(u,w) \in I \times \Delta^* \mid  \Ain(m(e)(u)) \Rightarrow_\G^* w\}.
	\]
    Note that $T(\G) \subseteq \dom(m(e)) \times \Delta^*$. 

    The \emp{language generated} (or, \emp{r-accepted}) by $\G$ is 
	\[
	L(\G) = \ran(T(\G)) = 
    \{w \in \Delta^* \mid \Ain(m(e)(u))\Rightarrow_\G^* w \text{ for some } u \in I \}.
	\] 

    The \emp{input set d-accepted} by $\G$ is 
	\[
	A(\G) = \dom(T(\G)) =
	\{u \in I \mid  \Ain(m(e)(u)) \Rightarrow_\G^* w \text{ for some } w \in\Delta^*\}.
	\]

    The $\CF(S)$ transducer $\G$ is \emp{(transducer) deterministic} if for
    every $c \in C$ and every two different rules $A \rightarrow \ruleif b_1 \rulethen \xi_1$
    and $A \rightarrow \ruleif b_2 \rulethen \xi_2$, $m(b_1 $ \bool{and} $ b_2)(c) = \false$. \QEDB
\end{df}

\sloppy The corresponding classes of translations, languages,
and input sets, are defined by 
\mbox{$\tauCF(S) = \{T(\G) \mid \G \text{ is a } \CF(S)\text{ transducer}\}$}, 
$\lambdaCF(S) = \{L(\G) \mid \G \text{ is a } \CF(S) \text{ transducer}\}$, 
and $\alphaCF(S) = \{A(\G) \mid \G \text{ is a } \CF(S) \text{ transducer}\}$.
Moreover,
$\tauDCF(S)= \{T(\G) \mid \G \text{ is a deterministic } \CF(S) \text{ transducer}\}$,
and similarly for $\lambdaDCF(S)$ and $\alphaDCF(S)$. Note that this $\lambda$ has
nothing to do with the empty string.

Let us discuss some notational conventions concerning
rules. A rule $A \rightarrow \ruleif $\bool{true}$ \rulethen \xi$ will be
abbreviated by $A \rightarrow \xi$ (these are the unconditional rules we
saw in $\G_2$ and $\G_3$). A~rule $A \rightarrow \ruleif$\bool{false}$\rulethen\xi$
may be omitted. Obviously, if $b$ and $b'$ are equivalent boolean expressions,
then they are interchangeable as tests of rules. Thus, if $S$ is a storage
type with $P = \emptyset$, we may assume that all rules are unconditional.
We write $A \rightarrow \ruleif b \rulethen \xi_1 \ruleelse \xi_2$ as an
abbreviation of the two rules
$A \rightarrow \ruleif b \rulethen \xi_1$ and 
$A \rightarrow \ruleif \text{\bool{not}} b  \rulethen \xi_2$ (cf. the program for
$\G_1$;
the two rules of $\G_1$ could be written 
$A \rightarrow \ruleif \nnull \rulethen a \ruleelse A(\dec) A(\dec)$).  
 A rule \mbox{$A \rightarrow \ruleif b_1 $ \bool{or} $ b_2 \rulethen \xi$} may be
replaced by the two rules $A \rightarrow \ruleif b_1 \rulethen \xi$ and
\mbox{$A \rightarrow \ruleif b_2 \rulethen \xi$}. 
Thus, using the disjunctive normal form of boolean expressions,
we may always assume that all tests in rules are conjunctions of
predicate symbols and negated predicate symbols. We will allow
rules $A \rightarrow \ruleif  b \rulethen \xi$ with
$\xi \in (N(F^+) \cup \Delta)^*$, where concatenation in $F^+$ is denoted
by a semicolon.  A rule
$A \rightarrow \ruleif b \rulethen \cdots B( f_1; f_2;\ldots;f_k) \cdots$
abbreviates the rules
$A \rightarrow \ruleif b \rulethen \cdots B_1(f_1)\cdots$, $B_1 \rightarrow B_2(f_2)$,
\ldots, $B_{k-1} \rightarrow B(f_k)$, where $B_1,\ldots, B_{k-1}$ are new nonterminals.

Let $\G$ be a $\CF(S)$ transducer. Whenever we view $\G$ in
particular as a generator of $L(\G)$, we will call $\G$ a 
\emp{context-free $S$ grammar}, or $\CF(S)$ grammar. Similarly, when viewing it as a
recognizer of $L(\G)$, as discussed before, we  call it a
\emp{context-free $S$ r-acceptor}, or $\CF(S)$ \mbox{r-acceptor} (where r
abbreviates range). But, of course, $\G$ can also be viewed as an
acceptor of $A(\G)$; in that case we call it a \emp{context-free $S$
d-acceptor}, or $\CF(S)$ d-acceptor (where d abbreviates domain).
Note that in this case the terminal alphabet $\Delta$ of $\G$ is
superfluous, i.e., we may assume that $\Delta = \emptyset$.

In the definition of $\CF(S)$ transducer we took the usual
terminology for grammars. From the point of view of recursive
automata, i.e., for transducers and acceptors, it would be more
appropriate to call the elements of $N$ \emp{states}, $\Ain$ the \emp{initial state},
$\Delta$ the \emp{output alphabet} (or the \emp{input alphabet}, for
r-acceptors), and to formalize $R$ as a \emp{transition function}. The
derivations of the transducer would then be called \emp{computations};
these computations start with $\Ain(c)$ where $c$ is an \emp{initial
configuration}, i.e., an element of $\ran(m(e))$. As said before,
from the point of view of programs, $N$ is the set of procedure
names, $\Ain$ is the main procedure, and $R$ the program (consisting
of the procedure declarations).

The notion of (transducer) determinism, as defined
above, is what one would expect for transducers and d-acceptors
(and, perhaps, grammars). Obviously, for a deterministic $\CF(S)$
transducer $\G$, $T(\G)$ is a partial function from $I$ to $\Delta^*$. For
r-acceptors, this notion is too strong, because the terminals
should also be involved; r-acceptor determinism will be
considered in Section \ref{sect7}.

As an example we now give a complete formal definition
of $\G_1$. First we define the storage type $\Integer = (C,P,F,I,E,m)$,
where $C$ is the set of integers, $P = \{\nnull\}$, $F = \{\dec\}$, $I = C$,
$E = \{\mathrm{en}\}$; for every $c \in C$, $m(\nnull)(c) = (c = 0)$ and
$m(\dec)(c) = c-1$; and $m(\mathrm{en}) = \id(C)$, the identity on $C$. 
Second we define 
the $\CF(\Integer)$ transducer $\G_1=(N,e,\Delta,\Ain,R)$ where
$N = \{A\}$, $e = \mathrm{en}$,   $\Delta = \{a\}$, $\Ain = A$, and 
$$R = \{A \rightarrow \ruleif \nnull \rulethen a, \;A \rightarrow \ruleif
\text{\bool{not}} \nnull \rulethen A(\dec)A(\dec)\}.$$
A derivation of $\G_1$ is 

\begin{center}\parbox{0cm}{
  \begin{tabbing}
   \=$A(2)$ \= $\Ar$ \= $A(1)A(1)$\\
   \> \> $\Ar$ \> $A(0)A(0)A(1)$\\
   \> \> $\Ar$ \> $aA(0)A(1)$\\
   \> \> $\Ar$ \> $aA(0)A(0)A(0)$\\
   \> \> $\quad\ \Rightarrow^*$ \> $ aaaa$
  \end{tabbing}}
\end{center}
Hence, since $m(\mathrm{en})(2) = 2$, we have $(2,aaaa) \in T(\G_1)$,
$aaaa \in  L(\G_1)$, and $2 \in A(\G_1)$. Clearly, 
$T(\G_1) = \{(n,a^{2^n}) \mid n \ge 0\}$,
$L(\G_1)=\{a^{2^n} \mid n \ge 0\}$, 
and $A(\G_1)=\{n\mid n\geq 0\}$.

\subsection{\sect{Comparison with attribute grammars}}\label{sect1.2}

As remarked in the Introduction, $\CF(S)$ grammars are a
very special case of the general idea of adding parameters to
context-free grammars. A much more powerful realization of this
idea is the notion of attribute grammar \cite{Knu}, in particular in
its formulation as affix grammar \cite{Kos, Wat1}. In fact, a $\CF(S)$
grammar may be viewed as an attribute grammar with one,
inherited, attribute. To explain this, let in particular
$S = (C,P,F,I,E,m)$ be a storage type such that $I$ is a singleton
($I = \{u_0\}$), and let $\G = (N,e,\Delta,\Ain,R)$ be a $\CF(S)$ grammar. Let us
call the inherited attribute $i$. It is an attribute of all
nonterminals in $N$, and it has type $S$, i.e., $C$ is the set of
attribute values for $i$, and $P$ and $F$ contain the possible tests
and operations on these attribute values ($S$ is also called the
semantic domain, cf. \cite{EngFil}). Every rule
$A \rightarrow \ruleif b \rulethen w_0 B_1(f_1) w_1  B_2(f_2) w_2 \cdots B_n(f_n) w_n$
of $R$ (with $A,B_j \in  N$ and $w_j \in \Delta^*$) determines
a rule of the underlying context-free grammar $\overline{\G} =
(N,\Delta,\Ain,\overline{R})$
of the attribute grammar, together with the semantic rules and semantic
conditions for its attributes,
as follows (for the notion of semantic condition, see \cite{Wat2}).
The rule $A \rightarrow w_0 B_1 w_1  B_2  w_2 \ldots B_n w_n$
is in $\overline{R}$, the semantic rules to 
compute the attributes (of the sons) are $i(B_j) = f_j(i(A))$, for
$1 \le  j \le n$, and the semantic condition (on the father) is
$b(i(A)) = \true$.\footnote{{\bf New Observation.} In this paper,
the direct descendants of a node of a tree are called its ``sons''
and the node itself is then called the ``father''; 
moreover, two such sons are ``brothers'' of each other. 
To avoid this patriarchate, 
many authors now use ``children'', ``parent'' and ``siblings''.
That terminology is misleading, because every child has two parents.
}
Usually the initial nonterminal $\Ain$ is not
allowed to have an inherited attribute; we allow this but fix
its value $i(\Ain)$ to be $m(e)(u_0)$. It should now be clear that
$L(\G)$ is the set of all strings in $L(\overline{\G})$ that have derivation
trees of which the values of the attributes satisfy all semantic
conditions. This is the usual way in which attribute grammars
define the context-sensitive syntax of languages. Note that to
the (left-most) derivations of the $\CF(S)$ grammar $\G$ correspond
derivation trees, in an obvious way; the nodes of these trees
are labeled by pairs $A(c)$ from $N(C)$. These derivation trees
correspond to the semantic trees of the attribute grammar $\G$,
i.e., derivation trees of $\overline{\G}$ together with the values of their
attributes. Thus, in this case, attribute evaluation can be
defined by way of the derivations of the $\CF(S)$ grammar; in fact
this holds in general, and it is precisely the way in which
attribute evaluation is defined formally in affix grammars (see
\cite{Kos,Wat1}).

As an example, a variation $\G_1'$ of $\G_1$ might have the rules

\begin{center}\parbox{0cm}{
  \begin{tabbing}
   \=$\Ain(x)$ \= $\ar$ \= $B(x)$\\
   \>$B(x)$ \> $\ar$ \> $B(x+1)$\\
   \>$B(x)$ \> $\ar$ \> $A(x)$\\
   \>$A(x)$ \> $\ar$ \> $\ruleif x \not= 0 \rulethen A(x-1)A(x-1)$\\
   \>$A(x)$ \> $\ar$ \> $\ruleif x = 0 \rulethen a$
  \end{tabbing}}
\end{center}

\noindent with $m(e)(u_0) = 0$. This corresponds to an attribute grammar with
$i(\Ain) = 0$ and rules $\Ain \rightarrow B$, $B \rightarrow B$, 
$B \rightarrow A$, $A \rightarrow AA$, and $A \rightarrow a$, in $\overline{R}$.
The attribute values are defined as follows. Note that, as
usual, subscripts denote different occurrences of the same
nonterminal; we use \cond to indicate a semantic condition.

\begin{minipage}{0.5\textwidth}
  \begin{center}\parbox{0cm}{
    \begin{tabbing} 
      syntactic rules\\[3mm]
      \=$\Ain$ \= $\ar$ \= $B$\\
      \>$B_0$ \> $\ar$ \> $B_1$\\
      \>$B$ \> $\ar$ \> $A$\\
      \>$A_0$ \> $\ar$ \> $A_1 A_2$\\
      \\
      \\
      \>$A$ \> $\ar$ \> $a$
    \end{tabbing}}
  \end{center}
\end{minipage}
\begin{minipage}{0.4\textwidth}
  \begin{center}\parbox{0cm}{
    \begin{tabbing}
      semantic rules\\[3mm]
      \=$i(B)$ \= $\ =\ $ \= $i(\Ain)$\\
      \>$i(B_1)$ \> $\ =\ $ \> $i(B_0)+1$\\
      \>$i(A)$ \> $\ =\ $ \> $i(B)$\\
      \>\cond $i(A_0) \ \neq \ 0$\\
      \>$i(A_1)$ \> $\ =\ $ \> $i(A_0)-1$\\
      \>$i(A_2)$ \> $\ =\ $ \> $i(A_0)-1$\\
      \>\cond $i(A) \ =\  0$      
    \end{tabbing}}
  \end{center}
\end{minipage}

\noindent The underlying context-free grammar generates all strings in $a^+$, 
but the attribute grammar $\G_1'$ generates $\{a^{2^n}\mid n\geq 0\}$.

Although $\CF(S)$ transducers are a very particular case of attribute
grammars, we will see in Section \ref{sect3}(7) how they can be used 
to model arbitrary attribute grammars!


\section{\sect{Regular grammars}}\label{sect2}

Two particular subcases of the context-free grammars
are the regular (= right-linear) grammars and the regular tree
grammars. Adding storage $S$, these can be used to model known
classes of automata: regular $S$ grammars for one-way $S$ automata,
and regular tree $S$ grammars for top-down tree automata with
storage $S$. We now discuss these two subcases one by one.

\begin{df} 
    A \emp{regular $S$ transducer}, or $\REG(S)$ transducer, is
    a context-free $S$ transducer $\G = (N,e,\Delta,\Ain,R)$ of which all rules
    in $R$ have one of the forms $A \rightarrow \ruleif b \rulethen wB(f)$ or
    $A \rightarrow \ruleif b \rulethen w$, where 
    $A,B \in N$, $b \in \BEP$, $w \in \Delta^*$, and $f \in F$. \QEDB
\end{df}

The corresponding classes of translations, languages,
and input sets are denoted by \mbox{$\tauREG(S)$}, $\lambdaREG(S)$,
and $\alphaREG(S)$,
respectively, and similarly for the deterministic case by 
$\tauDREG(S)$, $\lambdaDREG(S)$, and $\alphaDREG(S)$.

From the programming point of view a regular $S$
transducer consists of recursive procedures that only call each
other at the end of their bodies (tail recursion). It is well
known that such recursion can easily be removed, replacing calls
by goto's, and keeping the actual parameter in the global state.
Thus, we may view $\REG(S)$ transducers as ordinary flowcharts with
elements of \BEP \ in their diamonds and elements of
$F \cup \{\mathrm{write}(a) \mid a \in \Delta\}$ in their boxes. 
These flowcharts operate,
in the usual way, on a global state consisting of an object of
type $S$ (i.e., an $S$-configuration) and an output tape. If we
consider the $\REG(S)$ r-acceptors (with a one-way input tape
instead of an output tape: replace ``write'' by ``read''), it should
be clear that these are precisely the usual nondeterministic
\emp{one-way $S$ automata}. Thus:
\begin{quote}
``$\REG(S)$ r-acceptor = one-way $S$ automaton'',
\end{quote}
\noindent and $\lambdaREG(S)$ is the class of languages
accepted by one-way $S$
automata (but, as observed before, determinism does not carry
over; see Section \ref{sect7}).

An informal example of a $\REG(S)$ r-acceptor was given in
the discussion of $\G_2$ in Section \ref{sect1.1}: a one-way pushdown
automaton accepting $\{a^n b^n \mid n \geq 0\}$. It should be clear that the
type pushdown can be formalized as a storage type Pushdown, in
such a way that the $\REG$(Pushdown) r-acceptor corresponds to the
usual one-way pushdown automaton, see Section \ref{sect3}(2). This can be
done for all usual one-way automata, as shown successfully in
AFA theory \cite{Gin}. Thus AFA theory is the theory of $\REG(S)$
r-acceptors.

Let us look more closely at our ``notation'' for one-way
$S$ automata. As noted before in Section \ref{sect1.1}, for a $\REG(S)$
r-acceptor $\G = (N,e,\Delta,\Ain,R)$, the elements of $N$ are its states,
$\Ain$ is the initial state, $\Delta$ is the input alphabet, $\ran(m(e))$ is
the set of initial $S$-configurations with which $\G$ may start its
computations, and $R$ represents the transition function. A rule
$A \rightarrow \ruleif b \rulethen wB(f)$ should be interpreted as: 
``if the current state of the automaton is $A$, $b$ holds for its current storage
configuration, and $w$ is a prefix of the (rest of the) input,
then the automaton may read $w$, go into state $B$, and apply $f$ to
its storage configuration.'' A~rule $A \rightarrow \ruleif b \rulethen w$
should be interpreted as: ``if $\ldots\langle$as before$\rangle\ldots$, 
then the automaton may read~$w$, and halt.'' Thus, a string is accepted 
if the automaton reads it to its end, and then halts according to a 
rule of the second form. Note that in the total configurations of the $\REG(S)$
r-acceptor the already processed part of the input appears,
rather than the rest of the input. This is unusual, but may be
viewed as a notational matter. Actually, $\REG(S)$ transducers are
as close to one-way $S$ automata as right-linear grammars are to
finite automata (which is very close!).

However, intuitively, $\REG(S)$ r-acceptors only
correspond to one-way $S$ automata in case $S$ has an identity: in
general a one-way $S$ automaton is not forced to transform its
storage at each move.

\begin{df}
    A storage type $S = (C,P,F,I,E,m)$ \emp{has an identity},
    if there is an instruction symbol $\id \in F$ such that
    $m(\id) = \id(C)$. \QEDB
\end{df}

Grammar $\G_2$ in Section \ref{sect1.1} uses an identity in the first
and third rules.

Thus, when modelling particular well-known types of
one-way automata, such as pushdown automata, we should see to it
that the corresponding storage type has an identity. The reason
that we also consider storage types without identity is that
there exist devices, such as the top-down tree transducer, that
\emp{have} to transform their storage configuration at each step of
their computation. Instead of formalizing this in the control
(i.e., the form of the rules) of the device, it turns out to be
useful to formalize it in the storage type, as we do now. In
case we wish to consider both types of transducers, we define $S$
without an identity, and then add one, as follows.

\begin{df}\label{sid}
    For a storage type $S = (C,P,F,I,E,m)$, \emp{$S$ with
    identity} is the storage type $S_{\id} = (C,P,F\cup\{\id\},I,E,m')$,
    where $\id$ is a ``new'' instruction symbol, $m'$ is the same as
    $m$ on $P \cup F \cup E$, and $m'(\id) = \id(C)$. \QEDB
\end{df}

This is useful in particular when (part of) the storage
is viewed as input: then the identity constitutes a ``$\lambda$-move'' on
this input.

Note that, in the above definition, $S_{\id}$ is also defined
in case $S$ already has an identity; this simplifies some
technical definitions.

As an illustration of the use of an identity, for a
$\REG(S)$ r-acceptor, we note the following. Some people may not
like that the automaton can read a whole string from the input
in one stroke. Let us say that a $\REG(S)$ r-acceptor is in \emp{normal
form} if $w \in \Delta \cup \{\lambda\}$ in all its rules. Now let us assume that $S$
has an identity id. Then it is quite easy to see that every
$\REG(S)$ r-acceptor can be put into normal form: replace a rule of
the form $A \rightarrow \ruleif b \rulethen a_1 a_2 \cdots a_n B(f)$,
with $n \geq 2$, by the $n$ rules 
$A \rightarrow \ruleif b \rulethen a_1 B_1 (\id)$, 
$B_1 \rightarrow a_2 B_2 (\id),\ldots, B_{n-1} \rightarrow a_n B(f)$, where
$B_1,\ldots,B_{n-1}$ are new states, and similarly for a rule of the form
$A \rightarrow \ruleif b \rulethen a_1a_2\cdots a_n$.

\bigskip
In the remaining part of this section we consider
regular tree grammars and generalize them to regular tree $S$
grammars, just as we did for context-free grammars. First we
need some well-known terminology on trees (see, e.g., \cite{GecSte,
Eng1}).

A ranked set $\Delta$ is a set together with a mapping $\rk:
\Delta \rightarrow \{0,1,2,\ldots\}$. If $\Delta$ is finite, 
it is called a ranked alphabet.
For $k \geq 0$, $\Delta_k = \{\sigma \in\Delta\mid \rk(\sigma) = k\}$.
The set of trees over $\Delta$,
denoted $T_\Delta$, is the smallest subset of $\Delta^*$ such that (1) for every
$\sigma\in\Delta_0$, $\sigma$ is in $T_\Delta$, and (2) for every 
$\sigma \in\Delta_k$ with $k \geq 1$, and
every $t_1,t_2,\ldots,t_k \in T_\Delta$, $\sigma t_1 t_2\cdots t_k$ is in $T_\Delta$. 
For a set $Y$ disjoint with $\Delta$, $T_\Delta [Y]$
denotes $T_{\Delta\cup Y}$, where the elements of $Y$ are given rank 0. Note 
that we write trees in prefix notation, without parentheses or
commas; however, for the sake of clearness, we will sometimes
write $\sigma(t_1,t_2,\ldots,t_k)$ instead of $\sigma t_1 t_2 \cdots t_k$. A
language $L \subseteq \Delta^*$
is called a tree language if $L \subseteq T_\Delta$.

\begin{df}
    A \emp{regular tree $S$ transducer}, or $\RT(S)$ transducer,
    is a context-free $S$ transducer $\G = (N,e,\Delta,\Ain,R)$, 
    such that $\Delta$ is
    a ranked alphabet, and, for every rule 
    $A \rightarrow \ruleif b \rulethen \xi$ of $R$, $\xi$
    is in $T_\Delta [N(F)]$. \QEDB
\end{df}

As usual, the corresponding classes of translations,
languages, and input sets are denoted $\tauRT(S)$, 
$\lambdaRT(S)$, and $\alphaRT(S)$, respectively.

It is easy to see that, for an $\RT(S)$ transducer $\G$,
$L(\G) \subseteq T_\Delta$. Thus $L(\G)$ is a tree language, 
and $T(\G)$ translates input elements into trees.

As an example we consider the $\RT(S)$ transducer $\G'_2$, a
variation of $\G_2$; as for $\G_2$, \mbox{$S$ = Pushdown}.
The (informal) rules of $\G'_2$ are the same as those of $\G_2$, 
except that $\lambda$ has to be replaced by $\tau$
(of rank 0), and the rule $A(x) \rightarrow B(x)C(x)$
should be replaced by the rule $A(x) \rightarrow \sigma B(x)C(x)$,
where $\sigma$ has rank 2. Symbols $a$, $b$, and
$c$ have rank 1. Thus $\Delta = \{a,b,c,\sigma,\tau\}$ with 
$\Delta_0 = \{\tau\}$, $\Delta_1 = \{a,b,c\}$, and $\Delta_2 = \{\sigma\}$. 
The transducer $\G_2'$ generates all trees of the form $a^n\sigma(b^n\tau,c^n\tau)$,
i.e., a chain of $a$'s that forks into a chain of $b$'s and a chain
of $c$'s, cf. Fig.~\ref{Fig.1} for $n = 3$. This ends the example.

\begin{figure}
  \centering
\tikzset{sibling distance=4.5em, level distance=2.2em}
\begin{tikzpicture}

\node {$a$}
	child {node {$a$}
		child {node {$a$}
			child {node {$\sigma$}
				child {node {$b$}
					child {node {$b$}
						child {node {$b$}
							child {node {$\tau$}}
						}
					}
				}
				child {node {$c$}
					child {node {$c$}
						child {node {$c$}
							child {node {$\tau$}}
						}
					}
				}
			}
		}
	};

\end{tikzpicture}
  \caption[1]{A tree generated by the $\RT$(Pushdown) transducer $\G'_2$.}
  \label{Fig.1}
\end{figure}
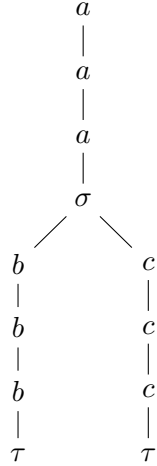

\smallskip
We will now discuss the fact that the regular tree $S$
grammar may also be viewed as a generalization of the
context-free $S$ grammar, just as in the case without $S$. A tree is
an expression built up from the ``operators'' of the ranked
alphabet $\Delta$. When an interpretation of these operators is given,
as operations on some set $D$ (a so-called $\Delta$-algebra $D$), then the
expressions of $T_\Delta$ denote elements of $D$, in the usual way (see
\cite{GogThaWagWri}). Thus an $\RT(S)$ transducer $\G$ together with a
$\Delta$-algebra $D$ define a translation from the input set $I$ to the
output set $D$. Note that from the programming point of view it is
this time best to view an $\RT(S)$ transducer as a set of recursive
function procedures, with arguments of type $S$ and results of
type $T_\Delta$ (or $D$). In this way the $\RT(S)$ transducer (together with
a $\Delta$-algebra) generalizes the $\CF(S)$ transducer. In fact, as is
well known, by taking the $\Delta$-algebra $\Delta_0^*$ with every 
$\sigma \in \Delta$ of \mbox{rank $\geq 1$} interpreted as concatenation, 
the $\RT(S)$ transducer turns, as a special case, into the $\CF(S)$ transducer, 
with terminal alphabet $\Delta_0$. In this case every tree denotes its yield,
defined as follows.

Let $\eps$ be a special symbol of rank 0. 
Then (1) for $\sigma\in\Delta_0$ with $\sigma\neq\eps$, $\yield(\sigma) = \sigma$,
and $\yield(\eps) = \lambda$, and 
(2) for $\sigma \in \Delta_k$, $k \geq 1$, 
$\yield(\sigma t_1 t_2\cdots t_k) = \yield(t_1) \yield(t_2)\cdots\yield(t_k)$. 
For a tree language $L \subseteq T_\Delta$,
$\yield(L) = \{\yield(t) \mid t \in L\}$. 
For a relation $R \subseteq I \times T_\Delta$, we
define $\yield(R) = \{(u,\yield(t)) \mid (u,t) \in R\}$. And for a class $K$
of tree languages or relations, $\yield(K) = \{\yield(B) \mid B \in K\}$.

It is easy to see (and will be proved in
Theorem \ref{theo8.1}(1)) that, for every $S$,
\begin{center}
$\lambdaCF(S) = \yield(\lambdaRT(S))$
\end{center}
and, in fact, $\tauCF(S) = \yield(\tauRT(S))$, and so $\alphaRT(S) = \alphaCF(S)$. 
Of course also $\tauRT(S) \subseteq \tauCF(S)$, because an $\RT(S)$
transducer is defined as a special type of $\CF(S)$ transducer.

Just as we viewed $\REG(S)$ r-acceptors as one-way $S$
automata, we can view $\RT(S)$ \mbox{r-acceptors} as \emp{top-down $S$ tree
automata}, i.e., ordinary top-down tree automata such that to
each occurrence of its states an $S$-configuration is associated.
Such an automaton receives a tree from $T_\Delta$ as input, processes
the tree from the root to its leaves (splitting at each node
into as many copies as the node has sons), and accepts the tree
if all parallel computations are successful. Let us say that an
$\RT(S)$ r-acceptor $\G = (N,e,\Delta,\Ain,R)$ is in \emp{normal form}
if all its rules are of one of the forms 
$A \rightarrow \ruleif b \rulethen B(f)$ or
$A \rightarrow \ruleif b \rulethen \sigma B_1(f_1)\cdots B_k(f_k)$
with $\sigma\in\Delta_k$, $k \geq 0$. As for the
regular case it is easy to show that if $S$ has an identity, then
every $\RT(S)$ r-acceptor can be transformed into normal form. A
rule $A \rightarrow \ruleif b \rulethen \sigma B_1(f_1)\cdots B_k(f_k)$, 
with $k \geq 1$, should be
interpreted as: ``if, at the current node, the state of the tree
automaton is $A$, $b$ holds for its storage configuration, and the
label of the node is $\sigma$, then the tree automaton splits into $k$
copies, one for each son of the node; at the $j$-th son, the
automaton goes into state $B_j$ and applies $f_j$ to its storage
configuration.'' A rule $A \rightarrow \ruleif b \rulethen \sigma$, 
with $\sigma\in\Delta_0$, should be
interpreted as: ``if $\ldots\langle$as before$\rangle\ldots$, 
then the automaton halts
at this node.'' Finally, a rule $A \rightarrow \ruleif b  \rulethen B(f)$ 
should be interpreted as: ``if $\ldots\langle$as before$\rangle\ldots$, 
then the automaton stays
at this node, goes into state $B$, and applies $f$ to its storage
configuration.'' Grammar $\G'_2$ discussed above, is a top-down
pushdown tree automaton, in this way.

Thus, if $S$ has an identity, then $\lambdaRT(S)$ is the class
of tree languages accepted by top-down $S$ tree automata.

An example of an $\RT(\Integer)$ transducer is
$\widetilde \G_1 = (N,e,\Delta,\Ain,R)$, where $N = \{A\}$, $e = \mathrm{en}$,
$\Delta = \{+,1\}$ with $\rk(+) = 2$ and $\rk(1) = 0$, $\Ain = A$,
and $R$ contains the rules (written informally)

\begin{center}\parbox{0cm}{
  \begin{tabbing}
      \=$A(x) $\=$\ar $ \=$\ruleif x\neq 0 \rulethen +\;A(x-1) \; A(x-1)$\\
      \>$A(x) $\>$\ar $\>$\ruleif x = 0 \rulethen 1$
  \end{tabbing}}
\end{center}

\noindent Note that $\widetilde \G_1$ is in normal form. It translates a nonnegative
integer $n$ into an expression over $\{+,1\}$ that, when interpreted
over the integers, with + as ordinary addition and 1 as the
integer 1, denotes $2^n$. As a tree, this expression is the full
binary tree of depth $n$. Thus, since

\begin{center}\parbox{0cm}{
  \begin{tabbing}
   \=$A(2)$ \= $\Ar$ \= ${+}A(1)A(1)$\\
   \> \> $\Ar$ \> ${++}A(0)A(0)A(1)$\\
   \> \> $\Ar$ \> ${++}1A(0)A(1)$\\
   \> \> $\Ar$ \> ${++}1A(0){+}A(0)A(0)$\\
   \> \> $\quad\ \Rightarrow^*$ \> $ {++}11{+}11$,
  \end{tabbing}}
\end{center}

\noindent 2 is translated into the tree +(+(1,1),+(1,1)) that denotes 4
when interpreted over the \mbox{$\Delta$-algebra} of integers. Note that, when
interpreted over the $\Delta$-algebra $a^*$, with + interpreted as
concatenation and 1 as $a$, this tree denotes $aaaa$ (and $\widetilde \G_1$ is
really $\G_1$).
Viewed as a program, with $\{+,1\}$ interpreted over the
integers, $\widetilde \G_1$ looks as follows (where + is written infix, as usual):

\begin{lstlisting}[mathescape]
function A(x: integer): integer;
    begin if x $ \neq $ 0
             then return(A(x-l) + A(x-1))
             else return(1)
          fi
    end;
{main program}
obtain n;
deliver A(n).
\end{lstlisting}


\newpage
\section{\sect{Specific storage types}}\label{sect3}

In this section we discuss several cases in which,
taking $S$ to be a specific storage type, the $\CF(S)$ transducer
turns into a well-known device. In each of these cases we claim
that the $\CF(S)$ transducer is just a ``definitional variation'' of
the known device. This means that their definitions are very
close (differing in some technical details only), and that the
equivalence of the two formalisms is easy to prove (sometimes
using some nontrivial known property of the device). Moreover,
we hope that the formulation of the device as a $\CF(S)$ transducer
gives more insight into ``what it really is'', in other words,
that the $\CF(S)$ transducer captures the essence of the device.
Thus, whenever the reader is not familiar with a certain device,
he may safely consider the corresponding $\CF(S)$ transducer as its
definition (but he should be careful with determinism). In what
follows we will usually say that the $\CF(S)$ transducer ``is'' the
device, in order to avoid the repeated use of phrases like ``can
be viewed as'', ``corresponds to'', ``is a definitional variation
of'', etc.

\vspace{3.5em}

(1) The first case is trivial: after adding a storage
type to context-free grammars, we now drop it again.

\begin{df}
    The \emp{trivial storage type} $S_0$ is defined by 
    $S_0 = (C,P,F,I,E,m)$, where
    \begin{itemize}
     \setlength{\itemsep}{0pt}
     \item $C = \{c_0\}$ for some arbitrary, but
    fixed, object $c_0$,
    \item $P = \emptyset$,
    \item $F = \{\id\}$,
    \item $I = C$, 
    \item $E = \{\mathrm{en}\}$, and 
    \item $m(\id) = m(\mathrm{en}) = \id(C)$.  \QEDB
    \end{itemize}
  
\end{df}

Obviously, the $\CF(S_0)$ grammar is the \emp{context-free
grammar}. As argued in the Introduction, it may also be viewed as
a recursive r-acceptor. Similarly, the $\REG(S_0)$ grammar is the
\emp{regular} (or right-linear) \emp{grammar}, and the $\RT(S_0)$ grammar is the
\emp{regular tree grammar}. Moreover, the $\REG(S_0)$ r-acceptor is the
\emp{finite automaton}, in particular when it is in normal form (note
that $S_0$ has an identity). The $\RT(S_0)$ r-acceptor is the \emp{top-down
finite tree automaton}, in particular, again, when it is in
normal form. Note, however, that usually finite (tree) automata
do not have $\lambda$-moves, i.e., rules $A \rightarrow B(\id)$; 
it is easy to see that these can be removed.

Thus $\lambdaCF(S_0) = \CF$, $\lambdaREG(S_0) = \REG$,
and $\lambdaRT(S_0) = \RT$ (the class of regular tree languages).

\vspace{3.5em}

(2) Our first nontrivial case is to take $S$ to be the
storage type pushdown. It is funny to attach pushdowns to the
nonterminals of a context-free grammar, but let us see what
happens.

\begin{df}
    The storage type \emp{Pushdown}, abbreviated P (not to
    be confused with the set of predicate symbols!), is defined by
    Pushdown = $(C,P,F,I,E,m)$, where
    
    \begin{itemize}
     \setlength{\itemsep}{0pt}
    \item $C = \Gamma^+$ for some fixed infinite set $\Gamma$ of pushdown
    symbols,
    \item $P = \{\ttop{=}\,\gamma \mid \gamma\in\Gamma\} \cup \{\bottom\}$,
    \item $F = \{\push(\gamma) \mid \gamma\in\Gamma\} \cup \{\pop\} 
	   \cup\{\stay(\gamma) \mid \gamma\in\Gamma\} \cup \{\stay\}$,
    \item $I = \{u_0\}$ for a fixed object $u_0$,
    \item $E = \Gamma$, with $m(\gamma)(u_0) = \gamma$ for every $\gamma\in E$, 
    \end{itemize}
    
    \noindent and for every $c = \delta\beta$ with $\delta\in\Gamma$ and $\beta\in\Gamma^*$
    (intuitively, $\delta$ is the
    top of the pushdown $\delta\beta$),
   
   \begin{itemize}
     \setlength{\itemsep}{0pt}
    \item $m(\ttop{=}\,\gamma)(c) = \true $ iff $ \delta=\gamma$,\footnote{Throughout
this paper we use ``iff'' as an abbreviation of ``if and only if''.
}
    \item $m(\bottom)(c) = \true $ iff $ \beta=\lambda$,
    \item $m(\push(\gamma))(c) = \gamma\delta\beta$,
    \item $m(\pop)(c) = \beta $ if $ \beta \neq \lambda$ and undefined
    otherwise,
    \item $m(\stay(\gamma))(c) = \gamma\beta$, and
    \item $m(\stay)(c) = c$. \QEDB
    \end{itemize}
    
\end{df}

It should be clear that this Pushdown corresponds to
the usual storage type of pushdowns. Note however that there is
no empty pushdown. In fact, ordinary pushdown automata halt in
case the pushdown becomes empty, but (if Pushdown would have an
empty pushdown) a $\CF$(Pushdown) transducer could always continue
on an empty pushdown with unconditional rules; this would cause
some technical inconveniences.

Note that Pushdown has an identity, viz. $\stay$. Note
that, due to our use of a set of encodings, each $\CF$(Pushdown)
transducer $\G = (N,\gamma_0,\Delta,\Ain,R)$ has its own initial bottom
pushdown symbol $\gamma_0$. Note that the $\stay(\gamma)$ instructions are
superfluous: if the pushdown does not consist of one ($\bottom$)
cell, then $\stay(\gamma)$ can be simulated by $(\pop;\push(\gamma))$; the
pushdown symbol of the bottom cell can be kept in the finite
control. The bottom predicate is also superfluous: one can
always mark the pushdown symbol of the bottom cell (and keep it
marked). Note finally that we may assume that each test in a
rule of a $\CF$(P) transducer consists of a single predicate symbol
of the form $\ttop{=}\,\gamma$. We state this as a lemma (see Lemma~3.30 of
\cite{EngVog2}).

\begin{lm}\label{cfp_nf}
    Every $\CF(\mathrm{P})$ transducer is equivalent to one in which
    all rules are of the form $A \rightarrow \ruleif \ttop{=}\,\gamma \rulethen \xi$.
\end{lm}

\begin{proof}
    Let $\G = (N,\gamma_0,\Delta,\Ain,R)$ be a $\CF$(P) transducer. We may
    assume that $\G$ does not use the bottom predicate (see above). Let
    $\Gamma_\G$ be the set of all pushdown symbols that occur in $R$, together
    with $\gamma_0$ (these are all pushdown symbols $\G$ ever uses). First
    transform $\G$ so that all tests in rules are conjunctions of
    negated and nonnegated predicate symbols (through their
    disjunctive normal form); we may assume that every predicate
    symbol $\ttop{=}\,\gamma$, with $\gamma\in\Gamma_\G$, occurs exactly once in
    such a conjunction. 
    Now consider a rule $A \rightarrow \ruleif b \rulethen \xi$.
    If $b$ contains
    only negated predicate symbols, i.e., $\text{\bool{not}}\ttop{=}\,\gamma$,
    then throw the rule away. Do the same if $b$ contains two
    nonnegated predicate symbols (they are mutually exclusive). In
    the remaining rules, erase all negated predicate symbols
    (because $\ttop{=}\,\gamma_1 \text{ \bool{and} } \text{\bool{not}}\ttop{=}\,\gamma_2$ 
    is equivalent to $\ttop{=}\,\gamma_1$). Now all
    tests are of the form $\ttop{=}\,\gamma$. Note that the construction
    preserves several special properties of $\CF$(P) transducers (such
    as determinism and regularity).
\end{proof}

Hence the $\REG$(P) r-acceptor is the ordinary one-way
\emp{pushdown automaton}. And the $\RT$(P) r-acceptor, in particular when
in normal form, is the top-down \emp{pushdown tree automaton},
recently defined by Guessarian (\cite{Gue2}; in fact she calls it the
restricted pushdown tree automaton, see (8) of this section).
Actually this is the only example in the literature of a
top-down $S$ tree automaton.

Now we turn to a less predictable connection. The $\CF$(P)
grammar is the \emp{indexed grammar} of \cite{Aho1} (see \cite{HopUll,Sal})
and so $\lambda$-$\CF$(P) = Indexed. Viewing flags as pushdown symbols, the
sequence of flags attached to each nonterminal in a sentential
form of an indexed grammar behaves just as a pushdown. Flag
production corresponds to pushing, and flag consumption to
popping. A flag producing rule (called a production in \cite{Aho1})
corresponds to an unconditional rule
$A \rightarrow w_0 B_1(\push(\gamma_1))w_1\cdots B_n(\push(\gamma_n))w_n$,
where some of the
$\push(\gamma_j)$ may also be stay (actually, in an indexed grammar, more
than one symbol can be pushed; this can be done here by using
elements of $N(F^+))$. A flag consuming rule (called an index
production in \cite{Aho1}) corresponds to a rule
$A \rightarrow\ruleif \ttop{=}\,\gamma \rulethen w_0 B_1(\pop)w_1\cdots B_n(\pop)w_n$ 
(consumption of $\gamma$).
Thus, in a $\CF$(P) grammar both kinds of rules are present, and
flag production and consumption can even be mixed in one rule
(it is easy to see that this can be simulated by an indexed
grammar). Consequently, a $\CF$(P) grammar has a more uniform
notation than an indexed grammar (see \cite{ParDusSpe2} for a
definition of indexed grammar closer to the $\CF$(P) grammar).
Also, personally, I must confess that I only understood indexed
grammars when I~found out they were just $\CF$(P) grammars in
disguise!

An example of a $\CF$(P) grammar was given as $\G_2$ of
Section \ref{sect1.1}. The reader should now recognize it as an indexed
grammar. Let us write down $\G_2$ formally: $\G_2 = (N,e,\Delta,\Ain,R)$,
where $N = \{\Ain,A,B,C\}$, $e = \#$, $\Delta = \{a,b,c\}$, and $R$ consists of the
rules

\begin{center}\parbox{0cm}{
\begin{tabbing}
    \=$\Ain \ \ $\=$\ar $\=$A(\stay)$\\
    \>$A $\>$\ar $\>$aA(\push (a))$\\
    \>$A $\>$\ar $\>$B(\stay)C(\stay)$\\
    \>$B $\>$\ar $\>$\ruleif \ttop {=}\, a \rulethen bB(\pop)$\\
    \>$B $\>$\ar $\>$\ruleif \ttop {=}\, \# \rulethen \lambda$\\
    \>$C $\>$\ar $\>$\ruleif \ttop {=}\, a \rulethen cC(\pop)$\\
    \>$C $\>$\ar $\>$\ruleif \ttop {=}\, \# \rulethen \lambda$
\end{tabbing}}
\end{center}	
Note that the pushdown alphabet is not explicitly mentioned in
$\G_2$; it can be obtained from $R$ and $\#$. As noted before, when
dropping $C(\stay)$ from the third rule, $\G_2$ turns into a $\REG$(P)
r-acceptor, i.e., a pushdown automaton. Also, changing $\lambda$ into $\tau$,
and changing the third rule into $A \rightarrow \sigma B(\stay)C(\stay)$, we obtain
the $\RT$(P) r-acceptor $\G_2'$ of Section \ref{sect2}, i.e., a top-down
pushdown tree automaton.

Thus, $\CF(S)$ grammars may also be viewed as a
generalization of indexed grammars: the nonterminals are indexed
by $S$-configurations rather than sequences of flags.

Since $\CF(S)$ grammars are attribute grammars with one
inherited attribute (Section \ref{sect1.2}), it follows that the indexed
languages are generated by attribute grammars with one,
inherited, attribute of type Pushdown; this was shown in \cite{Dob}.

The $\RT$(P) grammar might be called the \emp{indexed tree
grammar}. It corresponds to the regular tree grammar in exactly
the same way as the indexed grammar to the context-free grammar.
Thus $\lambdaRT$(P) is both the class of tree languages accepted by
pushdown tree automata and the class of tree languages generated
by indexed tree grammars (in fact, as shown in \cite{Gue2}, these are
the context-free tree languages \cite{Rou}). Since, as noted in the
previous section for arbitrary $S$, 
$\yield(\lambdaRT$(P)$) = \lambdaCF$(P), the
yields of these languages are the indexed languages 
(cf.~Section~5 of \cite{Gue2}).

One-turn pushdown automata are pushdown automata that
can never push anymore after doing a pop move. This property can
easily be incorporated in the storage, and this gives rise to
the storage type \emp{One-turn Pushdown}, abbreviated P$_{1\text{t}}$, 
defined and studied in \cite{Gin, Vog2, Vog3}. 
The $\REG($P$_{1\text{t}})$ r-acceptor is the one-way one-turn 
pushdown automaton \cite{GinSpa}, and the $\CF($P$_{1\text{t}})$ grammar is
the \emp{restricted indexed grammar} \cite{Aho1} (defined in such a way
that after flag consumption there may be no more flag
production). The grammar $\G_2$ of Section \ref{sect1.1} can easily be turned
into such a $\CF($P$_{1\text{t}})$ grammar, because its pushdowns make one turn
only.

\vspace{3.5em}

(3) The storage type Counter (of the usual one-way
counter automaton) can be obtained by restricting Pushdown to
have one pushdown symbol only: therefore we will call this
storage type also Pure-pushdown.

\begin{df}
    The storage type \emp{Counter} or \emp{Pure-pushdown} is
    defined in exactly the same way as Pushdown, except that
    $\Gamma=\{\gamma_0\}$ for some fixed symbol $\gamma_0$. \QEDB
\end{df}

Thus, the $\REG$(Counter) r-acceptor is the one-way
\emp{counter automaton}. Note that a pushdown of $k$ cells represents
the fact that the counter, say $x$, contains the number $k-1$;
$\push(\gamma_0)$ corresponds to $x := x+1$, $\pop$ to $x := x-1$,
and the bottom predicate to the predicate $x = 0$.

The $\CF$(Counter) grammar is a special type of indexed
grammar: the \emp{1-block-indexed grammar} (see \cite{Ern}). Note that the
grammar $\G_2$ of Section \ref{sect1.1} can easily be turned into a
$\CF$(Counter) grammar by using the predicate bottom rather than
$\ttop{=}\,\#$ (and, dropping $C(\stay)$, it is a one-way counter automaton).

\vspace{3.5em}

(4) The (new) storage type Count-down is the same as
Integer, defined at the end of Section~\ref{sect1.1}, except that we drop
the negative integers from the set of configurations.

\begin{df}
    The storage type \emp{Count-down} is $(C,P,F,I,E,m)$,
    where
    \begin{itemize}
     \setlength{\itemsep}{0pt}
     \item $C$ is the set of nonnegative integers,
     \item $P = \{\nnull\}$,
     \item $F = \{\dec\}$,
     \item $I = C$,
     \item $E = \{\en\}$, with $m(\en)=\id(C)$, 
    \end{itemize}
    \noindent and for every $c \in C$, 
    \begin{itemize}
     \setlength{\itemsep}{0pt}
     \item $m(\nnull)(c) = \true $ iff $ c =0$,
     \item $m(\dec)(c+1) = c$, and
     \item $m(\dec)(0)$ is undefined.\QEDB
    \end{itemize}
     
\end{df}

The $\CF$(Count-down) grammar is the \emp{EOL system} (see
\cite{RozSal}), and so $\lambdaCF$(Count-down) is the class of EOL
languages. Rather than applying rules in parallel (as in a
``real'' EOL system), the $\CF$(Count-down) grammar applies rules in
the ordinary context-free way; but it chooses an integer $n$ at
the start of the derivation, and then sees to it, by counting
down to zero, that the paths through the derivation tree do not
become longer than $n$. Note that Count-down has no identity.

Clearly, grammar $\G_1$ of Section \ref{sect1.1} is a $\CF$(Count-down)
grammar. It corresponds to the EOL system with the rule $a \rightarrow aa$.
In general, an EOL system with alphabet $\Sigma$ and terminal alphabet
$\Delta\subseteq \Sigma$ corresponds to a $\CF$(Count-down) grammar 
with \mbox{terminals} $\Delta$ and nonterminals $\{\bar{a}\mid a\in\Sigma\}$;
an EOL rule $a\rightarrow a_1 a_2\cdots a_n$
corresponds to a rule \mbox{$\bar{a} \rightarrow \ruleif \text{\bool{not}}\nnull 
\rulethen \bar{a}_1(\dec)\bar{a}_2(\dec)\cdots\bar{a}_n(\dec)$;}
and moreover the grammar has rules
\mbox{$\bar{a} \rightarrow \ruleif \nnull 
\rulethen a$} for all $ a\in\Delta$. Vice versa, it is not very
difficult to show that the $\CF$(Count-down) grammar is not more
powerful than the EOL system (cf. \cite{Eng2}); the proof involves
closure of EOL under homomorphisms.

Note that $\lambdaREG$(Count-down) is just the class of
regular languages.

\vspace{3.5em}

(5) Strings can be read from left to right, symbol by
symbol. This is defined in the next storage type that we call
One-way, because it corresponds to the input tape of one-way
automata.

\begin{df}
    The storage type \emp{One-way} is $(C,P,F,I,E,m)$, where
    \begin{itemize}
     \setlength{\itemsep}{0pt}
     \item $C = \Omega^*$ for some fixed infinite set $\Omega$ of input symbols,
     \item $P = \{\first{=}\,a \mid a \in \Omega\} \cup \{\empt\}$,
     \item $F = \{\rread\}$,
     \item $I = \Omega^*$, 
     \item $E = \{\Sigma\mid\Sigma$ is a finite subset of $ \Omega\}$,
     with $m(\Sigma) = \id(\Sigma^*)$ for every $\Sigma\in E$,
    \end{itemize}
    \noindent and for every $c = bw \in \Omega^*$ (with $b \in \Omega$, $w \in \Omega^*$),
    \begin{itemize}
     \setlength{\itemsep}{0pt}
     \item $m(\first{=}\,a)(c) = (b = a)$,
     \item $m(\first{=}\,a)(\lambda) = \false$,
     \item $m(\empt)(c) = \false$, 
     \item $m(\empt)(\lambda)=\true$,
     \item $m(\rread)(c) = w$, and 
     \item $m(\rread)(\lambda)$ is undefined. \QEDB
    \end{itemize}
\end{df}

Note that One-way has no identity. Note that every
alphabet in $E$ is viewed as one encoding symbol. Note that a
$\CF$(One-way) transducer $(N,\Sigma,\Delta,\Ain,R)$ translates strings of
$\Sigma^*$ into strings of $\Delta^*$. Thus the encoding of the transducer
determines its input alphabet (another example of the usefulness
of a set of encodings). Note finally that it may be assumed that
all rules of a $\CF$(One-way) transducer are of the form
$A \rightarrow \ruleif \first{=}\,a \rulethen \xi$ or of the form 
$A\rightarrow \ruleif \empt \rulethen \xi$ (cf.
Lemma~\ref{cfp_nf} for Pushdown). 

As an example, the following $\CF$(One-way) grammar $\G_4$
translates every string over the alphabet $\Sigma$ into the sequence of
its suffixes. We define $\G_4 = (N,\Sigma,\Sigma\cup\{\#\},\Ain,R)$, where $N = \{A,C\}$,
$\Ain = A$, and $R$ consists of the rules

\begin{center}\parbox{0cm}{
  \begin{tabbing}
      \=$A $\=$\ar $ \=$\ruleif \first{=}\,a \rulethen aC(\rread) \#A(\rread)$
	  \qquad \qquad \=for all $a\in\Sigma$\\
      \>$C $\>$\ar $ \>$\ruleif \first{=}\,a \rulethen aC(\rread)$
	  \qquad \qquad \>for all $a\in\Sigma$\\
      \>$A $\>$\ar $ \=$\ruleif \empt \rulethen \lambda$\\
      \>$C $\>$\ar $ \=$\ruleif \empt \rulethen \lambda$
  \end{tabbing}}
\end{center}

\noindent Clearly, $\G_4$ is a deterministic $\CF$(One-way) transducer that translates
$a_1 a_2 \cdots a_n$ (with $a_j \in \Sigma$) into
$a_1 a_2\cdots a_n\#a_2\cdots a_n\#a_3\cdots a_n\# \cdots\#a_n\#$. Note, by the
way, that, viewing $\G_4$ as a system of recursive (function) procedures, the
correctness of $\G_4$ is immediate.

The $\CF$(One-way) grammar is the \emp{ETOL system}, and the
deterministic $\CF$(One-way) grammar is the deterministic ETOL
system, or \emp{EDTOL system} \cite{RozSal}. Thus $\lambda$-(D)$\CF$(One-way)
is the class of E(D)TOL languages. The sequence of tables applied
during the derivation of the ETOL system corresponds to the
input string of the $\CF$(One-way) grammar (chosen
nondeterministically). Otherwise the correspondence is exactly
the same as for EOL systems and $\CF$(Count-down) grammars. In
fact, Count-down is the same as the restriction of One-way to
one symbol $a$, i.e., $\Omega = \{a\}$: $\nnull$ corresponds to empty (and also
to $\text{\bool{not}}\first{=}\,a$), and $\dec$ corresponds to $\rread$. See
\cite{EngRozSlu} for
a definition of ETOL system that is close to the $\CF$(One-way)
grammar.

Grammar $\G_4$ corresponds to the EDTOL system with a table
$a$ for each $a \in \Sigma$, containing the rules $A \rightarrow aC\#A$ 
and $C \rightarrow aC$ (and
identity rules $b \rightarrow b$ for all other symbols $b$), plus an
additional table that contains the rules $A \rightarrow \lambda$ 
and $C \rightarrow \lambda$ (and identity rules).

We note here that it is quite easy to show that
$\lambdaCF$(One-way) is included in $\lambdaCF$(P) (and even in
$\lambdaCF$(P$_{1\text{t}}$)).
Given a $\CF$(One-way) grammar $\G$ = $(N,\Sigma,\Delta,\Ain,R)$, construct the
$\CF($P$)$ grammar $\G'=(N\cup\{Z\},\#,\Delta,Z,R')$ , where $R'$ consists of
(1) all rules that are needed to build up an arbitrary string in the
pushdown: $Z \rightarrow Z(\push(a))$ and $Z \rightarrow \Ain$, for all $a \in \Sigma$, and 
(2)~all rules that simulate $\G$ with this input string, i.e., all rules of
$\G$ in which $\first{=}\,a$ is replaced by $\ttop{=}\,a$, $\empt$ by $\bottom$, and
$\rread$ by $\pop$. It should be clear that $L(\G') = L(\G)$. Thus the ETOL
languages are contained in the (restricted) indexed languages,
as originally shown, in this way, in \cite{Cul}. This result is
generalized in \cite{Vog1}, showing that ``the ETOL hierarchy'' is
included in ``the OI-hierarchy''. As a particular case it can be
shown that $\lambdaCF$(Count-down) $\subseteq \lambdaCF$(Counter),
i.e., the EOL languages are contained in the 1-block-indexed languages.

We also note that the controlled ETOL systems \cite{Asv,EngRozSlu} 
can be modelled by an appropriate generalization of
One-way, as shown in \cite{Vog1}; as a special case the controlled
linear context-free grammars are obtained, see \cite{Vog2}. In fact,
the other way around, context-free $S$ grammars may be viewed as
``storage controlled'' context-free grammars, i.e., context-free
grammars of which the derivations are controlled by the storage
configurations of $S$. This point of view is explained in the
first chapter of \cite{Vog4}.

The $\REG$(One-way$_{\id}$) transducer is
the finite-state transducer or \emp{a-transducer} (see Definition~\ref{sid}
for the definition of $S_{\id}$ for a storage type $S$). Thus,  
$\tauREG$(One-way$_{\id}$) is the class of
\mbox{a-transductions}. A small difference is that, using $\empt$, the
$\REG$(One-way$_{\id}$) transducer can detect the end of the input string (and
so the deterministic transducers define slightly different
classes of translations). Thus it would be better to say that
it is the a-transducer with endmarker. 
Note also that the $\REG$(One-way$_{\id}$) transducer has a
look-ahead of one input symbol.

The $\REG$(One-way) transducer is slightly more powerful than the 
\emp{generalized sequential machine} (gsm). 
A gsm can only translate the empty input string $\lambda$ 
into itself, whereas a $\REG$(One-way) transducer can translate it 
into any finite number of output strings, using rules of the form 
$\Ain \rightarrow \ruleif \empt \rulethen w$.
Disregarding the empty input string, $\tauREG$(One-way)
is the class of gsm mappings (with endmarker).

As an example, if we drop $C(\rread)$ from the first set of
rules of $\G_4$, we obtain a gsm that translates every $a_1 a_2\cdots a_n$
into $a_1\#a_2\#\cdots\#a_n\#$.

Note the asymmetric way in which input and output
strings are modeled in $\REG$(One-way) transducers: the input by a
storage type, the output by a terminal alphabet. This asymmetry
is of course inherent to the formalism of $\CF(S)$ transducers.

Note finally that the $\REG$(One-way) d-acceptor is the
finite automaton again.

\vspace{2.5em}

(6) Next we generalize input strings to input trees.
They can be read from top to bottom, node by node. See Section \ref{sect2}
for notation concerning trees.

\begin{df}
 The storage type \emp{Tree} is $(C,P,F,I,E,m)$, where
 \begin{itemize}
 \setlength{\itemsep}{0pt}
  \item $C = T_\Omega$ for some fixed ranked set $\Omega$, such that
	  $\Omega_k$ is infinite for every $k \geq 0$,
  \item $P = \{\rroot{=}\,\sigma\mid\sigma\in\Omega\}$, 
  \item $F = \{\sel_i \mid i \in \{1,2,3,\ldots\}\}$,
   \item $I = T_\Omega$,
  \item $E = \{\Sigma \mid \Sigma$ is a finite subset of $\Omega\}$,
     with $m(\Sigma) = \id(T_\Sigma)$ for every $\Sigma \in E$,
 \end{itemize}

 \noindent and for every $c = \tau t_1\cdots t_k \in T_\Omega$ 
  (with $\tau \in \Omega_k$, $k \geq 0$, $t_1,\ldots,t_k \in T_\Omega$),
 
 \begin{itemize}
  \setlength{\itemsep}{0pt}
  \item $m(\rroot{=}\,\sigma)(c) = (\tau = \sigma)$, 
  \item $m(\sel_i)(c) = t_i$ if $1 \leq i \leq k$, and
  \item $m(\sel_i)(c)$ is undefined if $i>k$. \QEDB
 \end{itemize}
\end{df}

Thus $m(\sel_i)$ selects the $i$-th subtree of the given
tree. Clearly, $\REG$(Tree) transducers do not have much sense:
they can only look at one path through the tree.

The $\RT$(Tree) transducer is the \emp{top-down tree
transducer}, and similarly for determinism \cite{Rou, Tha, EngRozSlu,
GecSte}. Thus $\tau$-(D)$\RT$(Tree) is the class of (deterministic)
top-down tree transductions. Hence, the $\CF$(Tree) transducer is
the \emp{top-down tree-to-string transducer} (see, e.g., \cite{EngRozSlu}),
or the \emp{generalized syntax-directed translation scheme} (GSDT,
\cite{AhoUll}). A rule of a top-down tree-to-string transducer has
the form $q(\sigma(x_1,\ldots, x_k)) \rightarrow 
w_0q_1(x_{j_1})w_1\cdots q_n(x_{j_n})w_n$ with $n \geq 0$;
this corresponds to the $\CF($Tree$)$ transducer rule
$q \rightarrow \if \ruleif \rroot{=}\,\sigma \rulethen 
w_0q_1(\sel_{j_1})w_1\cdots q_n(\sel_{j_n})w_n$.

As an example, let $\G = (N,\Delta,\Ain,R)$ be an ordinary context-free grammar,
say, in Chomsky normal form. The following
$\CF$(Tree) transducer $\G'$ generates the language $L(\G')= \{w\#w\#w \mid
w \in L(\G)\}$. Let $\G' = (N\cup\{Z\},\Sigma,\Delta\cup\{\#\},Z,R')$, where
$\Sigma = R \cup \{\$\}$, $\Sigma_0 = \{r \in R \mid r \text{ has the form }A
\rightarrow a\}, \Sigma_1 = \{\$\}$, $\Sigma_2 = \{r \in R \mid r$ has the
form $A \rightarrow BC\}$, and $R'$ is defined as follows.
\begin{itemize}
  \item The rule $Z \rightarrow \ruleif \rroot{=}\,\$ \rulethen
  \Ain(\sel_1)\#\Ain(\sel_1)\#\Ain(\sel_1)$ is in $R'$.
  \item If $r: A \rightarrow BC$ is in $R$, then $A \rightarrow \ruleif \rroot{=}\,r
  \rulethen B(\sel_1)C(\sel_2)$ is in $R'$.
  \item If $r: A \rightarrow a$ is in $R$, then $A \rightarrow \ruleif \rroot{=}\,r
  \rulethen a$ is in $R'$.
\end{itemize}

\noindent Clearly $\G'$ translates a (rule labeled) derivation tree of $\G$
(with extra root $\$$) into $w\#w\#w$ where $w$ is the ``yield'' of the
tree.

As shown in \cite{Eng2, EngRozSlu}, the ETOL system is
really the ``monadic case'' of the top-down tree transducer (i.e.,
all input symbols have rank 1 or 0). In fact, if all elements of
$\Omega$ have rank 1 or 0, then Tree becomes (almost) the same as
One-way.

The class $\lambdaRT$(Tree) is called the class of surface tree
languages \cite{Rou, Eng2}; $\alphaRT$(Tree) is the class $\RT$ of regular
tree languages (as shown in \cite{Rou}). Note finally that the
$\RT$(Tree$_{\id}$) transducer is a top-down tree transducer for which
``$\lambda$-moves'' are allowed (considered in [Eng6,*MalVog]).

\vspace{1.5em}

(7) Input trees can of course also be read by walking
from node to node along the edges of the tree. This is a
generalization of the previous storage type, in which one can
only walk downwards.

Informally, the storage configurations of Tree-walk
are nodes of trees. We assume the reader to be familiar with the
usual informal terminology concerning trees.

\begin{df}
 The storage type \emp{Tree-walk} is $(C,P,F,I,E,m)$, where 
 \begin{itemize}
  \setlength{\itemsep}{0pt}
  \item $C = \{(t,n) \mid t \in T_\Omega, n\text{ is a node of }t\}$, 
        with $\Omega$ as in Tree,
  \item $P = \{\llabel{=}\,\sigma \mid \sigma \in \Omega\} \cup \{\rroot\} \cup
      \{\son{=}\,i \mid i \in \{1,2,3,\ldots\}\}$,
  \item $F = \{\down_i \mid i \in \{1,2,3,\ldots\}\} \cup \{\up,\stay\}$,
  \item $I = T_\Omega$,
  \item $E = \{\Sigma \mid \Sigma\text{ is a finite subset of }\Omega\}$,
 \end{itemize}
  \noindent and $m$ is defined as follows. For every $(t,n)\in C$, 
  \begin{itemize}
  \setlength{\itemsep}{0pt}
   \item $m(\llabel{=}\,\sigma)(t,n) = \true$ iff $n$ is labeled $\sigma$ in $t$,
   \item $m(\rroot)(t,n) = \true$ iff $n$ is the root of $t$,
   \item $m(\son{=}\,i)(t,n) = \true$ iff $n$ is the $i$-th son of its father,
   \item $m(\down_i)(t,n) = (t,i$-th son of $n)$, and
   \item $m(\up)(t,n) = (t,$ father of $n)$.
  \end{itemize}
  \noindent Note that some of these may be undefined. Furthermore, 
   \begin{itemize}
  \setlength{\itemsep}{0pt}
  \item $m(\stay) = \id(C)$,
   \end{itemize}
   \noindent and for every $\Sigma \in E$,
   \begin{itemize}
  \setlength{\itemsep}{0pt}
     \item $m(\Sigma)(t) = (t,n)$ for every $t \in T_\Sigma$,
	where $n$ is the root of $t$, and
     \item $m(\Sigma)(t)$ is undefined for $t\notin T_\Sigma$. \QEDB
 \end{itemize}
\end{df}

The $\son{=}\,i$ predicates are needed because otherwise a
$\REG$(Tree-walk) transducer is not even able to do a depth-first
left-to-right search of the input tree (as shown in \cite{KamSlu}):
when it returns to the father, it does not know which son to
visit next. By the way, it is open whether the $\REG$(Tree-walk)
d-acceptor can accept all regular tree languages.

\vspace{1em}

\begin{newob}\rm\label{bojcol}
This is not open anymore. The problem was solved by Boja\'nczyk and Colcombet 
in~\cite{*BojCol2}: the $\REG$(Tree-walk)
d-acceptor \emp{cannot} accept all regular tree languages, i.e., 
$\alphaREG$(Tree-walk) is properly included in $\RT$. 
In~\cite{*BojCol1} they proved that $\alphaDREG$(Tree-walk) 
is properly included in $\alphaREG$(Tree-walk). 
\end{newob}

\vspace{1em}

The $\REG$(Tree-walk) transducer is the \emp{tree walking automaton} of
\cite{AhoUll}, which is called \emp{checking tree transducer} (CT
transducer) in \cite{EngRozSlu}. It translates trees into strings,
and is equivalent to a subcase of the top-down tree-to-string
transducer (cf. (7) of Section \ref{sect6}).

We now wish to convince the reader that the
deterministic $\RT$(Tree-walk) transducer can be viewed as the
\emp{attribute grammar} of \cite{Knu}. Actually, it corresponds to the
attribute grammar viewed as tree transducer, cf. \cite{EngFil}. But
if we also take into account the possibility of interpreting $\Delta$
(into a $\Delta$-algebra, see Section \ref{sect2}), it really corresponds to the
attribute grammar. However, in this case we do not claim that
$\tauDRT$(Tree-walk) is equal to the tree transductions realized by
attribute grammars; due to several (uninteresting) technical
details this may not be true. Still, the formalisms are ``very
close''. To see this, let us consider an arbitrary deterministic
$\RT$(Tree-walk) transducer $\G = (N,\Sigma,\Delta,\Ain,R)$. Then $\Sigma$
determines,
in some sense, the underlying context-free grammar of the
attribute grammar; the input trees from $T_\Sigma$ are the derivation
trees of the context-free grammar (this holds in particular in
the ``many-sorted case'', see \cite{GogThaWagWri}). The elements of $N$
are the attributes. They are not partitioned into inherited and
synthesized attributes (for the fact that this is not necessary,
see \cite{Tie}). The elements of $T_\Delta$ (or of $D$, in case a
$\Delta$-algebra $D$
is also given) are the attribute values. The attribute $\Ain$ is a designated
attribute of the root of the input tree, the value of which is
the translation of the tree. Finally, $R$ contains the semantic
rules for computing the attribute values. A rule of the form,
say, $A \rightarrow \ruleif \llabel{=}\,\sigma \rulethen
\cdots B(\up)\cdots C(\down_2)\cdots D(\stay)\cdots $ with
$\rk(\sigma) \geq 2$, expresses the $A$-attribute of any node labeled $\sigma$
in terms of the $B$-attribute of its father, the $C$-attribute of its
second son, and its own $D$-attribute. Thus, such a rule combines
the inherited and synthesized features of the attributes; it
just expresses the attribute of a node in terms of those of its
neighbors (and itself). Note that semantic conditions cannot be
used explicitly, but should be simulated by semantic rules. The
$\son{=}\,i$ predicates are needed to express attributes of a node in
terms of those of its brothers (see the next example).

This point of view on attribute grammars is actually
one of the most obvious ones. Recall from Section \ref{sect2} that an
$\RT(S)$ transducer is a set of recursive function procedures with
arguments of type $S$ (here: node of a tree) and results of type
$T_\Delta$ (or $D$, where $D$ is a \mbox{$\Delta$-algebra}). Now, it should be clear that
the semantic rules of an attribute grammar are just a way of
recursively programming functions of that type, one function for
each attribute. The derivations of the $\RT$(Tree-walk) transducer
are precisely the computations of the recursive function
procedures, and thus form a (inefficient) way of attribute
evaluation. In \cite{Ful, Kam} these derivations are used to formally
define the translation realized by an attribute grammar, and in
\cite{Jal, Jou} they form the basis of more efficient attribute
evaluation (cf. Section 4.1 of \cite{Eng10}).

Note that if we drop ``up'' and ``stay'' from Tree-walk, we
more or less reobtain Tree. This shows the well-known fact that
attribute grammars with synthesized attributes only are closely
related to DRT(Tree) transducers, i.e., deterministic top-down
tree transducers (see, e.g., \cite{CouFra}).

Finally we note that attribute grammars with strings as
values, and concatenation as only operation (see \cite{DusParSedSpe,
Kam, EngFil, Eng7}), are modeled by deterministic $\CF$(Tree-walk)
transducers.

To illustrate the use of the $\son{=}\,i$ predicates we give
an example of a deterministic $\CF$(Tree-walk) transducer $\G_5$. It
has two nonterminals, $L$ and $U$; for every node $(t,n), L(t,n)$
generates that part of the yield of $t$ that is to the left of $n$,
and $U(t,n)$ generates that part of the yield of $t$ that is to the
left of $n$ or below $n$, see Fig.~2. Intuitively, $L$ is an inherited
attribute, and $U$ is synthesized. So, $\G_5 = (\{L,U\},\Sigma,\Sigma_0,U,R)$,
where $R$ contains the following rules (recall the convention of
using elements of $N(F^+)$, see Section \ref{sect1.1}).

\begin{center}\parbox{0cm}{
\begin{tabbing}
    \=$L  \ $\=$\ar $\=$\ruleif \rroot \rulethen \lambda$\\
    \>$L $\>$\ar $\>$\ruleif \son{=}\,1 \rulethen L(\up)$\\
    \>$L $\>$\ar $\>$\ruleif \son{=}\,i \rulethen U(\up; \down_{i-1})$\\[1mm]
    \>\>\>(for $2\leq i\leq m$, where $m$ is the maximal rank in
$\Sigma$)\\[1mm]
    \>$U $\>$\ar $\>$\ruleif \llabel{=}\,\sigma \rulethen L(\stay)\sigma \qquad
    \qquad$ \= for all $\sigma\in\Sigma_0$\\
    \>$U $\>$\ar $\>$\ruleif \llabel{=}\,\sigma \rulethen U(\down_k)$ \> for all
$\sigma\in\Sigma_k, k\geq 1$
\end{tabbing}}
\end{center}

\noindent Note that the purpose of $\G_5$ is not to define $T(\G_5)$, because
that could be done in a much easier way; instead, its purpose is to
be able to compute $L(t,n)$ and $U(t,n)$ for arbitrary $(t,n)$.

\vspace{3.5em}

\begin{figure}[h]
  \centering
\begin{tikzpicture}[decoration=brace]

\draw (0cm,0cm)  -- (4cm,0cm)  -- +(120:4cm) -- cycle;

\draw (2cm,1.732cm)  -- (1cm,0cm)  -- (0:3cm) -- cycle;

\coordinate[label=left:$n$] (N) at (2cm,1.732cm); 
\node (t) [yshift=8em, xshift=3em]{$t$};

\fill (N) circle (0.05);

\draw[decorate, yshift=-1ex] (1,0) -- node[below=0.4ex] {$l$} (0,0);

\draw[decorate, yshift=-5ex] (3,0) -- node[below=0.4ex] {$u$} (0,0);

\end{tikzpicture}
  \caption[1]{$l=L(t,n)$ and $u=U(t,n)$.}
  \label{Fig.2}
\end{figure}
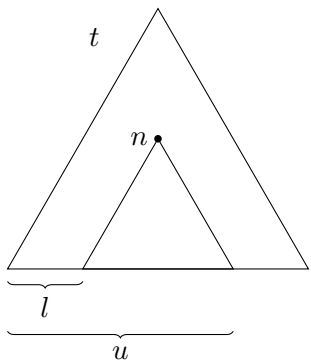

\vspace{3.5em}

(8) The last storage type we consider is the
generalization of the pushdown to trees, introduced in \cite{Gue2}
and formalized in \cite{DamGue}. The top of the tree-pushdown is the
root of the tree. This root may be replaced by any piece of
tree. Let $Y = \{y_1,y_2,...\}$ be an infinite set of ``variables''.

\begin{df}
 The storage type \emp{Tree-pushdown} is $(C,P,F,I,E,m)$, where 
 \begin{itemize}
 \setlength{\itemsep}{0pt}
  \item $C = T_\Omega$ (with $\Omega$ as in Tree),
  \item $P = \{\rroot{=}\,\sigma \mid \sigma \in \Omega\}$,
  \item $F = \{\expand(\zeta) \mid \zeta \in T_\Omega[Y]\}$,
  \item $I = \{u_0\}$ for a fixed object $u_0$,
  \item $E = \Omega_0$, with $m(\sigma)(u_0) = \sigma$ for every $\sigma \in E$,
 \end{itemize}
 \noindent and for every $c = \tau t_1\ldots t_k \in T_\Omega$ (with $\tau \in
 \Omega_k$, $k \geq 0$, $t_1,\ldots ,t_k \in T_\Omega$),
 \begin{itemize}
 \setlength{\itemsep}{0pt}
  \item $m(\rroot{=}\,\sigma)(c) = (\tau = \sigma)$, and 
  \item $m(\expand(\zeta))(c) =$ the result of substituting $t_i$ for $y_i$ in
        $\zeta$, for all $1 \leq i \leq k$ (undefined if this is not in
	$T_\Omega$). \QEDB
 \end{itemize}

\end{df}

Note that an $\expand(y_i)$ corresponds to a pop, whereas
all other expands correspond to pushes.

The $\RT$(Tree-pushdown) r-acceptor is the top-down
\emp{tree-pushdown tree automaton} (called pushdown tree automaton in
\cite{Gue2}). Although Tree-pushdown has no identity, it can easily
be simulated.

As noted in \cite {DamGue}, the $\RT$(Tree-pushdown) grammar is
the \emp{creative dendrogrammar} of \cite{Rou}. The creative dendrogrammar
with one state (= nonterminal) only, is the (outside-in)
\emp{context-free tree grammar}. Moreover, every creative
dendrogrammar is equivalent to one with one state only (\cite{Rou},
the proof is nontrivial). Thus $\lambdaRT$(Tree-pushdown) is the class
of context-free tree languages (see also \cite{Gue2, EngVog2}). In
the same way the $\CF$(Tree-pushdown) grammar is related to the
(outside-in) \emp{macro grammar} of \cite{Fis}. Recall that the
$\CF$(Pushdown) grammar is the indexed grammar; thus the
equivalence of the macro grammar and the indexed grammar \cite{Fis}
can (partly) be explained by the equivalence of the storage
types Tree-pushdown and Pushdown (see \cite{EngVog2}), and similarly
for the context-free tree grammar and the indexed tree grammar
(cf. $\RT$(P) of point (2)).

Inside-out context-free tree grammars and macro
grammars do not seem to fit into the $\CF(S)$ formalism.


\newpage
\section{\sect{Alternating automata}}\label{sect4}

Since alternation is close to recursion, context-free
grammars may be viewed as the prototype of alternating automata.
In fact, the $\CF(S)$ d-acceptor is the \emp{alternating $S$ automaton}
(see \cite{ChaKozSto, LadLipSto, Ruz} for alternating automata; and
see \cite{May} for the relationship between AND/OR programming and
context-free grammars). Here, $S$ is assumed to contain both the
input and the internal storage of the automaton (as opposed to
the treatment of one-way $S$ automata in Section \ref{sect2}, using
r-acceptors). In particular when $S$ is a storage type with $I = \Omega^*$
for some infinite alphabet $\Omega$, the input set $A(\G)$ d-accepted by a
$\CF(S)$ d-acceptor $\G$ is an ordinary language (because $A(\G) \subseteq I$),
and so $\alphaCF(S)$ is the class of languages accepted by alternating
$S$ automata. Note that, as observed before, we may assume that
the terminal alphabet is empty.

To explain the correspondence between $\CF(S)$ d-acceptors
and alternation, recall that an alternating automaton has two
kinds of states: existential and universal states. In an
existential state, \textit{some} possible next move has to lead to
success, whereas in a universal state \textit{all} possible next moves
should lead to acceptance. In a $\CF(S)$ d-acceptor, there is no
such difference between nonterminals, but existentiality is
modeled by the choice between two possible rules with the same
left-hand nonterminal, whereas universality is modeled by having
several nonterminals in the right-hand side (as opposed to a
$\REG(S)$ d-acceptor, where there is at most one). For instance, if
all rules for $A$ and $b$ are $A \rightarrow \ruleif b \rulethen B(f_1)C(f_2)$
and $A \rightarrow \ruleif b \rulethen D(f_3)$, then the alternating automaton,
in state $A$ and with $b$ true for its storage configuration, either goes into
state $D$ (applying $f_3$), \textit{or} splits itself in two and goes into
state $B$ (applying $f_1$) \textit{and} into state $C$ (applying $f_2$). It is
easy to see that, assuming that $S$ has an identity, every $\CF(S)$
d-acceptor can be transformed into one that has existential
states and universal states: for an existential state $A$, all
rules with left-hand side $A$ have the form of those of a $\REG(S)$
d-acceptor, and for a universal state $A$, all rules with
left-hand side $A$ have the form of those of a deterministic $\CF(S)$
d-acceptor.

From this it should also be clear that the
deterministic $\CF(S)$ d-acceptor is the \emp{universal $S$ automaton},
i.e., the alternating $S$ automaton with universal states only.

As a first example, the $\CF$(One-way) d-acceptor is the
\emp{alternating one-way finite automaton}. Thus, $\alphaCF$(One-way) is the
class of languages accepted by alternating one-way finite
automata. Recall that the $\CF$(One-way) grammar is the ETOL
system; this relationship between the parallelism of L systems
and that of alternating automata was pointed out in \cite{Eng6}. Note
that the $\REG$(One-way) d-acceptor is the finite automaton. In
general, the $\REG(S)$ d-acceptor is the nondeterministic $S$
automaton (where $S$ contains both the input and the internal
storage of the automaton).

To be able to consider arbitrary alternating one-way $S$
automata (i.e., alternating automata with a one-way input tape
and internal storage $S$), we need the notion of product of
storage types.

\begin{df}\label{prod_type}
 Let $S_i = (C_i,P_i,F_i,I_i,E_i,m_i)$, for $i = 1,2$, be two
storage types, with \mbox{$P_1 \cap P_2 = \emptyset$.} Their \emp{product} $S_1
\times S_2$ is the storage type \[(C_1 \times C_2, P_1 \cup P_2,F_1 \times
F_2, I_1 \times I_2, E_1 \times E_2, m),\] where

\bigskip
\begin{itemize}
 \setlength{\itemsep}{0pt}
 \item $m(p)(c_1,c_2) = m_i(p)(c_i)$ for $p \in P_i$,
 \item $m(f_1,f_2)(c_1,c_2) = (m_1(f_1)(c_1),m_2(f_2)(c_2))$, and
 \item $m(e_1,e_2)(u_1,u_2) = (m_1(e_1)(u_1),m_2(e_2)(u_2))$. \QEDB
\end{itemize}

\end{df}

In other words, everything is defined element-wise
except the predicates, which can be combined in boolean
expressions anyway. We can now say that the $\CF$(One-way$_{\id} \times S$)
\mbox{d-acceptor} is the \emp{alternating one-way $S$ automaton}. As an
example, the following \mbox{$\CF$(One-way$_{\id} \times $P)} \mbox{d-acceptor} $\G_6$, i.e.,
alternating one-way pushdown automaton, accepts the language
$A(\G_6) = \{a^n b^n c^n \mid n \geq 1\}$. Since for P the input set $I$ is a
singleton, we identify the input set of One-way$_{\id} \times $P with that
of One-way$_{\id}$, i.e., $\Omega^*$.

Let $\G_6 = (N,e,\emptyset,\Ain,R)$, where $N = \{\Ain,A,B,C\}$,
$e = (\{a,b,c\},\#)$, i.e., the input alphabet is $\{a,b,c\}$, the bottom
pushdown symbol is $\#$, and $R$ contains the following rules

\begin{center}\parbox{0cm}{
\begin{tabbing}
    \=$\Ain  $\=$\ar $\=$A(\id,\stay)$\\
    \>$A $\>$\ar $\>$\ruleif \first{=}\,a \rulethen A(\rread,\push (a))$\\
    \>$A $\>$\ar $\>$\ruleif \first{=}\,b \text{ \bool{and} } \ttop{=}\,a
      \rulethen B(\rread,\pop)C(\rread,\stay)$\\
    \>$B $\>$\ar $\>$\ruleif \first{=}\,b \text{ \bool{and} } \ttop{=}\,a \rulethen
      B(\rread,\pop)$\\
    \>$B $\>$\ar $\>$\ruleif \first{=}\,c \text{ \bool{and} } \ttop{=}\,\# \rulethen
      \lambda$\\
    \>$C $\>$\ar $\>$\ruleif \first{=}\,b \rulethen C(\rread,\stay)$\\
    \>$C $\>$\ar $\>$\ruleif \first{=}\,c \text{ \bool{and} } \ttop{=}\,a
      \rulethen C(\rread,\pop)$\\
    \>$C $\>$\ar $\>$\ruleif \empt \text{ \bool{and} } \ttop{=}\,\# \rulethen \lambda$.
\end{tabbing}}
\end{center}

\noindent Note that $\G_6$ is deterministic, i.e., universal. Note also that,
informally, we may view the nonterminals of $\G_6$ to have \textit{two}
parameters, say, $x$ and $y$. The fourth rule of $\G_6$ may be written
informally as $B(x,y) \rightarrow \ruleif \first(x)=b$ \bool{and} $\ttop(y)=a
\rulethen B(\rread(x),\pop(y))$, and similarly for the other rules.

Replacing One-way by Tree or Tree-walk gives
\emp{alternating tree automata} \cite{Slu}. Thus, the $\CF$(Tree) d-acceptor
is the alternating top-down finite tree automaton; note that it
is just the domain of a top-down tree transducer. In general,
the $\CF$(Tree$_{\id} \times S$) d-acceptor may be called the alternating
top-down $S$ tree automaton. Similarly, the $\CF$(Tree-walk)
d-acceptor is the alternating tree walking automaton; its
deterministic (i.e., universal) version is just the domain of an
attribute grammar.

Instead of One-way, it is possible to consider other
ways of handling the input, such as allowing \emp{auxiliary Turing
machine space}. For every space-constructable function $f$ on the
nonnegative integers, the storage type SPACE($f$) can be defined
as $(C,P,F,I,E,m)$ where $C$ consists of 
all 4-tuples of strings $(\#w_1,w_2\$,\#v_1,v_2\$)$
with $|v_1|+|v_2| = f(|w_1|+|w_2|)$. Intuitively, $w_1$ is the content of the input
tape to the left of the reading head, and $w_2$ is the remainder of the input tape.
Similarly, $v_1v_2$ is the content of the auxiliary space-restricted
Turing machine tape, with the reading head on the first symbol of $v_2$. 
The sets $P$ and $F$ can be defined so as to model the
usual tests and operations on a two-way input tape and a Turing
machine tape. Finally, $I = \Omega^*$, and $E$ consists of all functions $e$
that encode a string $w$ over some alphabet $\Sigma\subseteq \Omega$ as
$m(e)(w) = (\#,w\$,\#,v\$)$, where $v$ is the blank tape of length $f(|w|)$.
Then the REG(SPACE($f$)) d-acceptor is the nondeterministic
SPACE($f$) Turing machine, and the CF(SPACE($f$)) d-acceptor is the
\emp{alternating} SPACE($f$) \emp{Turing machine}. Also, e.g., the
CF(SPACE($f)\times$ P) d-acceptor is the alternating SPACE($f$) auxiliary
pushdown automaton \cite{LadLipSto}. Similarly TIME($f$) can be defined
as a storage type, such that $C$ consists of pairs $(u,t)$,
where $u$ codes an input tape and a sequence of Turing machine tapes, and
$t\in\N$. Initially $t = f(|w|)$ where $w$ is the content of the input tape, 
and each instruction applied to $(u,t)$ decreases $t$ by 1 (cf. Count-down). 
In this way, the CF(TIME($f$)) d-acceptor is the alternating TIME($f$) Turing machine.

We conclude this section with two related, rather
technical, observations.

First observation. In the regular case, both
$\alphaREG$(One-way$_{\id}\times S$) and $\lambdaREG(S)$ are the class
of languages accepted by nondeterministic one-way $S$ automata. In the
context-free case there does not seem to be any relationship
between the alternating one-way $S$ automaton and the context-free
$S$ grammar, i.e., between $\alphaCF$(One-way$_{\id}\times S)$ and
$\lambdaCF(S)$. Since the alternating finite automata accept the regular
languages \cite{ChaKozSto}, taking $S = S_0$ gives
\begin{quote}
$\alphaCF$(One-way$_{\id}\times S_0) =
\alphaCF$(One-way$_{\id}) = \REG \subsetneq \CF = \lambdaCF(S_0)$.
\end{quote}
But, since the
alternating one-way pushdown automata accept $\cup\{$DTIME$(c^n) \mid
c > 0\}$, see \cite{ChaKozSto}, taking $S$ = Pushdown gives 
\begin{quote}
$\lambdaCF($P) = Indexed $\subsetneq\cup\{$DTIME$(c^n) \mid c > 0\}$ =
$\alphaCF$(One-way$_{\id}\times$ P), 
\end{quote}
because the indexed languages are context-sensitive and closed under homomorphisms. 
Thus there is no generally valid inclusion between these classes.

Second observation. However, there does exist a
relationship between $\alphaRT$(One-way$_{\id}\times S$) and
$\lambdaRT(S)$: it
is stated in Theorem~6 of \cite{Eng6} (see \cite{DamGue} for a detailed proof) that
\begin{quote}
$\alphaCF$(One-way$_{\id}\times S$) =
$\tauRT$(One-way$_{\id})^{-1}(\lambdaRT(S)$)
\end{quote}
where it is assumed that $S$ has an identity and that $I$ is a
singleton (cf. Section \ref{sect7}). This establishes a formal link
between the alternating one-way $S$ automata and the top-down $S$
tree automata (by way of monadic top-down tree transducers with
$\lambda$-moves). In \cite{DamGue}, the same result is also proved with
One-way$_{\id}$ replaced by Tree$_{\id}$, establishing an analogous link
between the alternating and the nondeterministic top-down $S$ tree
automata (by way of top-down tree transducers with $\lambda$-moves).
Note that the $\RT$(Tree$\times S$) transducers studied in \cite{DamGue} may
be called top-down $S$ tree transducers.


\newpage
\section{\sect{Pushdown $S$ automata}}\label{sect5}

The recursion of the context-free grammar can be
simulated by the (nonrecursive) pushdown automaton, and, vice
versa, every pushdown automaton can be simulated by a
context-free grammar. Of course, recursion can always be
implemented on a pushdown. Thus, for the context-free $S$ grammar
we now look for a pushdown-like automaton equivalent to it.
Clearly, recursive procedures with one parameter (of type $S$) can
be implemented on a pushdown of which each cell contains both a
pushdown symbol (the name of a procedure) and an object of type
$S$ (its actual parameter), see Fig.~3. Let us define this as a
storage type (introduced in \cite{Gre}).

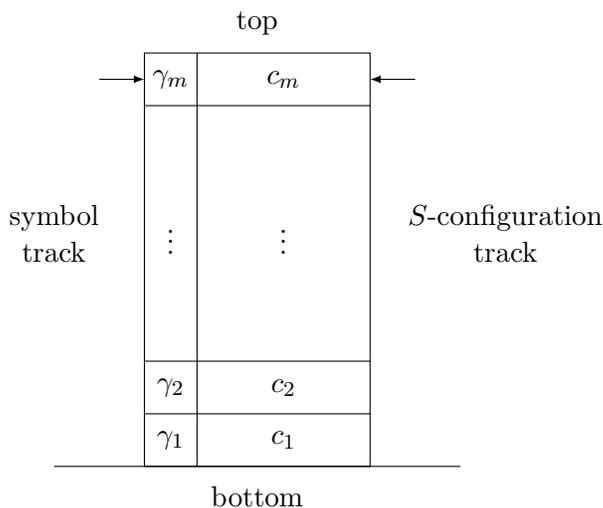
\begin{figure}[h]
 \centering
\begin{tikzpicture}

\draw (0,0) -- (3,0) -- (3,5.5) -- (0,5.5) --(0,0);
\draw (0,0.7) -- (3,0.7);
\draw (0,1.4) -- (3,1.4);
\draw (0,4.8) -- (3,4.8);
\draw (0.7,0) -- (0.7,5.5);
\draw (-1.2,0) -- (4.2,0);

\draw[-latex] (-0.6,5.15) -- (0,5.15);
\draw[-latex] (3.6,5.15) -- (3,5.15);

\node at (1.5,-0.4) {bottom};
\node at (1.5,5.9) {top};
\node[align=center] at (-1.2,3.1) {symbol\\ track};
\node[align=center] at (4.8,3.1) {$S$-configuration\\ track};

\node at (0.35,0.35) {$\gamma_1$};
\node at (0.35,1.05) {$\gamma_2$};
\node at (0.35,5.15) {$\gamma_m$};

\node at (1.85,0.35) {$c_1$};
\node at (1.85,1.05) {$c_2$};
\node at (1.85,5.15) {$c_m$};
\node at (1.85,3.1) {$\vdots$};
\node at (0.35,3.1) {$\vdots$};

\end{tikzpicture}
 \caption{Configuration $(\gamma_m,c_m)\cdots(\gamma_2,c_2)(\gamma_1,c_1)$ of
Pushdown of $S$. Only $\gamma_m$ and $c_m$ are accessible.}
\end{figure}

\begin{df}

Let $S = (C,P,F,I,E,m)$ be a storage type. The
storage type \emp{Pushdown~of~$S$}, abbreviated by $\Pd(S)$, is
$(C',P',F',I',E',m')$, where
\begin{itemize}
 \setlength{\itemsep}{0pt}
 \item $C'=(\Gamma\times C)^+$ for the fixed infinite set $\Gamma$ of pushdown
  symbols,
 \item $P' = \{\ttop{=}\,\gamma \mid \gamma\in\Gamma\} \cup \{\bottom\} \cup
  \{\test(p) \mid p \in P\}$,
 \item $F' = \{\push(\gamma,f) \mid \gamma \in \Gamma, f \in F\} \cup \{\pop\}
  \cup \{\stay(\gamma) \mid \gamma \in \Gamma\} \cup \{\stay\}$,
 \item $I' = I$,
 \item $E' = \Gamma \times E,$ with $m(\gamma,e)(u) = (\gamma,m(e)(u))$ for
  every $\gamma \in \Gamma$, $e \in E,$ and $u \in I$,
\end{itemize}

\noindent and for every $c' = (\delta,c)\beta$ with $\delta \in \Gamma$, 
$c \in C$, and $\beta \in (\Gamma \times C)^*$ (intuitively, $(\delta,c)$ is the top of
the pushdown $(\delta,c)\beta$),

\begin{itemize}
 \setlength{\itemsep}{0pt}
 \item $m'(\ttop{=}\,\gamma)(c') = \true$ iff $\delta=\gamma$,
 \item $m'(\bottom)(c') = \true$ iff $\beta = \lambda$,
 \item $m'(\test(p))(c') = m(p)(c)$,
 \item $m'(\push(\gamma,f))(c') = (\gamma,m(f)(c))c'$ if $m(f)$ is
  defined on $c$, and undefined otherwise,
 \item $m'(\pop)(c') = \beta$ if $\beta \neq \lambda$, and undefined otherwise,
 \item $m'(\stay(\gamma))(c') = (\gamma,c)\beta$, and
 \item $m'(\stay)(c') = c'$. \QEDB
\end{itemize}
\end{df}

For $b \in \BEP$, we will use $\test(b)$ to denote the
element of BE($P'$) that is obtained from $b$ by replacing every
$p \in P$ by $\test(p)$.

\begin{df}
The storage type \emp{Pure-pushdown of $S$}, abbreviated
by $\Pd_{\p}(S)$, is defined in exactly the same way as $\Pd(S)$, except
that $\Gamma=\{\gamma_0\}$ for some fixed symbol $\gamma_0$. \QEDB
\end{df}

Remarks: (1) The storage type Pushdown of $S_0$ is almost the
same as Pushdown. The only difference is that it additionally
has $c_0$ in all its pushdown cells. Therefore we may identify the
two. Thus $\Pd(S_0) = \Pd$ (not to be confused with the set $P$ of
predicates). Similarly, $\Pd_{\p}(S_0)$ = Counter.

(2) As for $\Pd$, it can be shown that the $\stay(\gamma)$ instructions
and the $\bottom$ predicate are superfluous, see Lemma 3.31 of
\cite{EngVog2}.

(3) $\Pd(S)$, or rather $\Pd_{\p}(S)$, was introduced, in a slightly
different form, in \cite{Gre}, where it was called a \emp{nested AFA}. For
the reader familiar with \cite{Gre}, we now discuss the
correspondence in more detail. Let us add instructions $\stay(\gamma,f)$
to P($S$), with meaning $m'(\stay(\gamma,f))(c') = (\gamma,m(f)(c))\beta$,
where $c'$ is as in the definition of $\Pd(S)$ (see Lemma 3.31 of \cite{EngVog2}
for the fact that this does not strengthen $\Pd(S)$). Now, if $S$ contains
an identity $\id$, then $\stay(\gamma)$ is the same as $\stay(\gamma,\id)$.
Moreover, it is easy to see that the $\push(\gamma,f)$ instructions can
then be restricted to the ``duplicate'' instructions $\push(\gamma,\id)$.
Thus the instructions now are: $\push(\gamma,\id)$, $\pop$, $\stay(\gamma,f)$,
and $\stay$. In this way $\Pd_{\p}(S)$ is very close to the nested AFA.
Formal equivalence between $\Pd(S)$ and the nested AFA (in the sense that
the corresponding classes of one-way automata are equivalent)
holds under a few weak restrictions on $S$.

(4) In Section \ref{sect1.2} we have seen that the $\CF(S)$ grammar is an
attribute grammar with one inherited attribute. In \cite{LewRosSte}
the \emp{attributed pushdown machine} was defined, and shown to be
equivalent to a particular type of attribute grammar (the L
attribute grammar, including the case of one inherited
attribute). Clearly, a pushdown cell $(\gamma,c)$ of $\Pd(S)$ may be viewed
as an attributed pushdown symbol. Thus, in the case of one
inherited attribute, the attributed pushdown machine has storage
type $\Pd(S)$. To handle synthesized attributes, the attributed
pushdown machine also has instructions of the form, say, $\pop(f)$
that send attribute values downwards in the pushdown (as opposed
to $\push(\gamma,f)$ that send them upwards). \QEDB

\vspace{1em}

Our aim is now to show that \emp{the context-free $S$ grammar
is equivalent to the one-way $\Pd(S)$ automaton}, or, more general,
that the $\CF(S)$ transducer is equivalent to the $\REG(\Pd(S))$
transducer. It turns out that this is true in case $S$ has an
identity. However, in general the $\CF(S)$ transducer can be
simulated by the $\REG(\Pd(S))$ transducer but not vice versa. To
obtain an equivalence anyway, there are two solutions: either
restrict the $\REG(\Pd(S))$ transducer or extend the $\CF(S)$
transducer. Here we consider the second solution (for the first
\mbox{see \cite{EngVog2}}, where the so-called bounded excursion pushdown is
defined). See Definition~\ref{sid} for the definition of $S_{\id}$.

\begin{df}
 Let $S = (C,P,F,I,E,m)$ be a storage type. An
 \emp{extended context-free $S$ transducer}, or $\CFext(S)$ transducer,
 is a $\CF(S_{\id})$ transducer $(N,e,\Delta,\Ain,R)$ such that, for every
 nonterminal $A$, $A(\id)$ may only appear at the end of a rule. In
 other words: for every rule $A \rightarrow \ruleif b \rulethen \xi$ of $R$,
 either $\xi \in (N(F) \cup \Delta)^*$ or $\xi \in (N(F) \cup
\Delta)^*N(\{\id\})$. \QEDB 
\end{df}

The corresponding classes of translations, languages, and input sets are denoted 
as usual, with subscript ext. 

It is easy to see that there is a \emp{normal form} for
$\CFext(S)$ transducers (similar to Chomsky normal form): we may
assume that all rules are of one of the forms

\begin{center}\parbox{0cm}{
\begin{tabbing}
    \=$A $\=$\ar $\=$\ruleif b \rulethen wB(\id)$\\
    \>$A $\>$\ar $\>$\ruleif b \rulethen C(f)B(\id)$\\
    \>$A $\>$\ar $\>$\ruleif b \rulethen w$
\end{tabbing}}
\end{center}

\noindent where $A,B,C \in N$, $b \in \BEP$, $f \in F$, and $w \in \Delta^*$.

We now state the correspondence between $\CF(S)$ and
$\REG(\Pd(S))$. Note that if two classes of translations are equal,
then so are the classes of their ranges and of their domains.
Thus, ``if something holds for $\tau$-, then it also holds for $\lambda$- and
$\alpha$-.''

\begin{theo}\label{theo5.1}
Let $S$ be a storage type.
\begin{enumerate}[label=\rm(\arabic*)]
 \item $\tauCFext(S) = \tauREG(\Pd(S))$.
 \item If $S$ has an identity, then $\tauCF(S)$ = $\tauREG(\Pd(S))$.
 \item The $\subseteq$-inclusions in (1) and (2) also hold for the
corresponding deterministic transducers.  
\end{enumerate}
\end{theo}

\begin{proof}
 If $S$ has an identity, then that can be used instead of
$\id$; and so $\tauCFext(S) = \tauCF(S)$. Hence (2) follows from
(1). We now show (1) and (3).

``$\tauCFext(S) \subseteq \tauREG(\Pd(S))$.'' 
Let $\G = (N,e,\Delta,\Ain,R)$ be a
$\CFext(S)$ transducer in normal form. The $\REG(\Pd(S))$ transducer $\G'$
to be constructed simulates, as usual, the left-most derivations
of $\G$. To this end it uses the elements of $N$ as pushdown symbols.
The $S$-configuration associated to a nonterminal $A$ in a
derivation of $\G$, is stored in the same pushdown cell as $A$. Thus,
$\G' = (N',e',\Delta,\Ain',R')$, where $N' = \{\$\}$, $e' = (\Ain,e), \Ain' =
\$$,
and $R'$ is defined as follows (recall the convention concerning
the use of $N(F^+)$ in rules, see Section~\ref{sect1.1}).

\begin{itemize}
 \item If $A \rightarrow \ruleif b \rulethen wB(\id)$ is in $R$, 
       then $R'$ contains the rule \\
       $\$ \rightarrow \ruleif \test(b)\text{ \bool{and} }\ttop{=}\,A 
       \rulethen w\$(\stay(B))$.
 \item If $A \rightarrow \ruleif b \rulethen C(f)B(\id)$ is in $R$, 
       then $R'$ contains the rule \\ 
       $\$ \rightarrow \ruleif \test(b)\text{ \bool{and} }\ttop{=}\,A
	   \rulethen \$(\stay(B);\push(C,f)).$
 \item If $A \rightarrow \ruleif b \rulethen w$ is in $R$, 
       then $R'$ contains the rules \\
       $\$ \rightarrow \ruleif \test(b)\text{ \bool{and} }\ttop{=}\,A
       \text{ \bool{and} \bool{not}}\bottom \rulethen w\$(\pop)$, and \\
       $\$ \rightarrow \ruleif \test(b)\text{ \bool{and} }\ttop{=}\,A
       \text{ \bool{and} }\bottom \rulethen w$.
\end{itemize}

This concludes the construction of $\G'$. It should be
clear that $T(\G')$ = $T(\G)$. Note that if $\G$ is deterministic, then
so is $\G'$ (the transformation into normal form also preserves
determinism).

``$\tauREG(\Pd(S))$ $\subseteq \tauCFext(S)$.'' 
Let $\G = (N,e,\Delta,\Ain,R)$ be a
$\REG(\Pd(S))$ transducer, with $e = (\gamma_0,e')$. An equivalent
$\CFext(S)$
transducer will be obtained by the usual triple construction. To
facilitate this construction we make three assumptions
concerning $\G$. First, we assume that ``final'' rules
$A \rightarrow \ruleif b \rulethen w$ of $\G$ (with $w \in \Delta^*$) are only
applied when the
pushdown contains exactly one cell (clearly, using the bottom
predicate, one can reduce the pushdown to one cell just before
such a rule is applied). Second, we assume that the $\bottom$
predicate is not used by $\G$. Third, we assume that all rules of $\G$
are of the form $A \rightarrow \ruleif \test(b)\text{ \bool{and} }\ttop{=}\,\gamma
\rulethen \xi$, with $b \in \BEP$ and $\gamma \in \Gamma$. 
This can be achieved in a straightforward way (see
Lemma~\ref{cfp_nf}, and see Lemma 3.30 of \cite{EngVog2}).

Let $\Gamma_\G$ be the set of all pushdown symbols that occur in
$R$, together with $\gamma_0$. We now construct the $\CFext(S)$
transducer
$\G' = (N',e',\Delta,\Ain',R')$, where $N'$ consists of all triples
$\langle A,\gamma,B\rangle$
with $A \in N$, $\gamma \in \Gamma_\G$, and $B \in N$ or $B = \omega$; $\omega$
is a new symbol,
indicating the end of a derivation of $\G$; $\Ain' = \langle
\Ain,\gamma_0,\omega\rangle$.
Intuitively, $\G'$ has a derivation $\langle A,\gamma,B\rangle(c) \Rightarrow^*
w$ (with $B \in N$)
iff $\G$, starting in state $A$ and with a cell containing $(\gamma,c)$ at
the top of its pushdown, can derive $w$ and pop that cell from the
pushdown, ending in state $B$ (i.e., $\G$ has a derivation
$A((\gamma,c)B) \Rightarrow^* wB(\beta)$ that does not test $\beta$). Similarly,
$\langle A,\gamma,\omega\rangle(c) \Rightarrow^* w$ indicates that $\G'$ has a
derivation
$A((\gamma,c)) \Rightarrow^* w$. From this, the following construction of $R'$
should be clear.

\begin{itemize}
 \setlength{\itemsep}{0pt}
 \item If $R$ contains $A \rightarrow \ruleif \test(b)$ \bool{and} $\ttop{=}\,\gamma
      \rulethen w$, then $R'$ contains\\ $\langle A,\gamma,\omega\rangle
      \rightarrow \ruleif b \rulethen w$.
 \item If $R$ contains $A \rightarrow \ruleif \test(b)$ \bool{and} $\ttop{=}\,\gamma
      \rulethen wB(\pop)$, then $R'$ contains\\ $\langle A,\gamma,B\rangle
      \rightarrow \ruleif b \rulethen w$.
 \item If $R$ contains $A \rightarrow \ruleif \test(b)$ \bool{and} $\ttop{=}\,\gamma \rulethen
      wB(\push(\delta,f))$, then $R'$ contains the rules \\
      $\langle A,\gamma,C\rangle \rightarrow \ruleif b \rulethen
      w\langle B,\delta,E\rangle(f)\langle E,\gamma,C\rangle(\id)$ for all $C
      \in N \cup \{\omega\}$ and $E \in N$.
 \item If $R$ contains $A \rightarrow \ruleif \test(b)$ \bool{and} $\ttop{=}\,\gamma
      \rulethen wB(\stay(\delta))$, then $R'$ contains the rules \\ 
      $\langle A,\gamma,C\rangle \rightarrow \ruleif b \rulethen
      w\langle B,\delta,C\rangle(\id)$ for all $C \in N \cup \{\omega\}$.
 \item If $R$ contains $A \rightarrow \ruleif \test(b)$ \bool{and} $\ttop{=}\,\gamma
      \rulethen wB(\stay)$, then $R'$ contains the rules \\
      $\langle A,\gamma,C\rangle \rightarrow \ruleif b \rulethen
      w\langle B,\gamma,C\rangle(\id)$ for all $C \in N \cup \{\omega\}$.
\end{itemize}

This concludes the construction of $\G'$. It is left to
the reader to prove formally that $T(\G')=T(\G)$. Note that even
if $\G$ is deterministic, $\G'$ is not deterministic, due to the
choice of the ``return nonterminal'' $E$ in the rules of $\G'$
corresponding to the push-rules of $\G$.
\end{proof}

It is quite easy to show, for all $S$, that if $S$ has an
identity, $\lambdaREG(S)$ is a full trio, i.e., it is closed under
a-transductions (see \cite{Gin}, or \cite{EngVog4}). Thus, using this for
P$(S)$ rather than $S$, it follows from Theorem \ref{theo5.1}(2) that if $S$ has
an identity, then $\lambdaCF(S)$ is a full trio. Probably $\lambdaCF(S)$
is even a (full) super-AFL (see \cite{Gre}), under suitable conditions
on $S$. It seems that AFL theory for $\lambdaRT(S)$, 
i.e., for top-down $S$ tree automata, does not yet exist.

In the remaining part of this section we will try to
say more about the deterministic case, in the direction from
$\REG(\Pd(S))$ to $\CFext(S)$. By the previous theorem we know of course
that $\tauDREG(\Pd(S)) \subseteq \tauCFext(S)$, but
actually the extension is not needed. Since this is only based on the fact that
the translation defined by a deterministic transducer is a partial
function, we state this as follows.

\begin{lm}\label{cfext_pf}
 Let $S$ be a storage type, and let $\mathrm{PF}$ be the class of
 partial functions. Then $\tauCFext(S)\cap \mathrm{PF} \subseteq \tauCF(S)$.
\end{lm}

\begin{proof}
Let $\G = (N,e,\Delta,\Ain,R)$ be a $\CFext(S)$ transducer in normal form such
that $T(\G)$ is a partial function. Construct a $\CF(S)$ transducer $\G'$
by repeatedly replacing the $A(\id)$ in the right-hand side of a
rule by the right-hand side of one of $A$'s rules, avoiding
repetition of the same $A$. Formally, $\G' = (N,e,\Delta,\Ain,R')$, where
$R'$ is defined as follows.

\begin{tabbing}
If $R$ contains the rules \=$A_1 \rightarrow \ruleif b_1 \rulethen \xi_1 A_2(\id)$,\\
\>$A_2 \rightarrow \ruleif b_2 \rulethen \xi_2 A_3(\id)$,\\
\>$\ \ldots$\\
\>$A_{n-1} \rightarrow \ruleif b_{n-1} \rulethen \xi_{n-1}A_n(\id)$, and\\
\>$A_n \rightarrow \ruleif b_n \rulethen \xi_n$\\
where $A_1,\ldots,A_n$ are \emp{different} nonterminals of $N$
$(n\geq 1)$, $b_i \in \BEP$, and 
$\xi_i \in (N(F) \cup \Delta)^*$, \\
then the rule
$A_1 \rightarrow \ruleif b_1\text{ \bool{and} }\cdots\text{ \bool{and} }b_n \rulethen
\xi_1\cdots \xi_n$ is in $R'$.
\end{tabbing}

If we would drop the condition that the $A_1,\ldots,A_n$ are
different, it would be clear that $T(\G') = T(\G)$ (if you are not
afraid of infinitely many rules). A repetition of, say, $A$ in
$A_1,\ldots,A_n$ allows for derivations in $\G$ of the form
$A(c) \Rightarrow^* \xi A(c)$. But, derivations in $\G$ of the form
$\Ain(m(e)(u)) \Rightarrow^* \alpha A(c)\beta \Rightarrow^* \alpha\xi A(c)\beta
\Rightarrow^* wxyz$ (with $\alpha \Rightarrow^* w \in \Delta^*$, etc.) 
are superfluous; in fact, for the derivation
$\Ain(m(e)(u)) \Rightarrow^* \alpha A(c)\beta \Rightarrow^* wyz$, we must have
$wyz = wxyz$ because $T(\G)$ is a partial function. This shows that we can restrict
ourselves to different $A_1,\ldots,A_n,$ and so $T(\G') = T(\G)$.
\end{proof}

Remarks: (1) Since the construction in the previous proof
preserves determinism of the transducer, we conclude that
$\tauDCFext(S) = \tauDCF(S)$.

(2) Since every $\CFext(S)$ transducer with empty terminal alphabet 
defines a partial function, we conclude that $\alphaCFext(S) = \alphaCF(S)$.

(3) An argument similar to the proof of this lemma would
show, for $S$ = Tree, that the $\CFext$(Tree) transducer is
equivalent to the regularly extended top-down tree-to-string
transducer of \cite{EngRozSlu}. \QEDB

\vspace{1em}

We now know that $\tauDREG(\Pd(S))$ is in $\tauCF(S)$,
and we would like to show that it is even in $\tauDCF(S$). But in
general this is not true. Take for instance $S$ = Tree. Let $\sigma$ be
of rank 2, and $a,b$ of rank 0, and consider the tree-to-string
translation $T = \{(\sigma aa,a),(\sigma bb,b)\}$. It is easy to see (and well
known) that $T$ cannot be defined by a deterministic top-down tree
transducer (i.e., a $\DCF$(Tree) transducer): then also $\sigma ab$ and $\sigma ba$ 
would be in its domain. On the other hand, it is also easy to
see that $T$ can be defined by a $\DREG(\Pd$(Tree)) transducer: use
$\push(\gamma,\sel_1)$ and $\push(\gamma,\sel_2)$ to inspect the subtrees of the
input tree. Thus $\tauDCF$(Tree) is properly included in
$\tauDREG(\Pd$(Tree)). To solve such problems, the top-down tree
transducer was equipped with a \emp{look-ahead} facility
\cite{Eng4, Eng5,EngRozSlu}. To define $T$, the top-down tree transducer
could look ahead at the subtrees of the input tree, to see
whether they have the same label.

Let $\G$ be a deterministic $\REG(\Pd(S))$ transducer, and
consider the equivalent $\CFext(S)$ transducer $\G'$ as defined
in the second part of the proof of Theorem~\ref{theo5.1}.
As observed at the end of that proof, the only determinism of $\G'$
is due to the choice of the ``return nonterminal'' $E$ in rules of the form 
$$\langle A,\gamma,C\rangle \rightarrow \ruleif b \rulethen
      w\langle B,\delta,E\rangle(f)\langle E,\gamma,C\rangle(\id)$$
that correspond to the push-rules of $\G$. 
For a configuration $c$, $\langle B,\delta,E\rangle(c)$ generates 
a terminal string iff $\G$ has a derivation $B((\delta,c)\beta) \Rightarrow^* wE(\beta)$
that does not test the nonempty pushdown~$\beta$. 
Clearly, due to the determinism of $G$, the ``return state'' $E$ is uniquely determined by 
the ``state''~$B$, the pushdown symbol $\delta$, and the configuration $c$. 
In other words, for given $B$, $\delta$ and $c$, there is at most one $E$
such that $\langle B,\delta,E\rangle(c)$ has a successful derivation in $\G'$. 
If $\G'$ could test the latter property (i.e., have some
knowledge about its own future behavior), then it could pick,
deterministically, the unique successful rule for $\langle A,\gamma,C\rangle$ (if it exists). 
Such tests will be called \emp{look-ahead} tests (also because in case $S$ = Tree it
corresponds to the above notion of look-ahead at subtrees).
Formally, we define them as an extension $S_\LA$ of a given storage
type~$S$. Thus, a transducer with look-ahead will not only be able
to test its own future behavior, but also that of others.

\begin{df}
Let $S = (C,P,F,I,E,m)$ be a storage type, and let
$\G = (N,e,\Delta,\Ain,R)$ be a $\CF(S)$ transducer. 
The \emp{set of configurations accepted} by $\G$ is 
$\Acc(\G) = \{c \in C \mid \Ain(c) \Rightarrow^* w$ 
for some $w \in \Delta^*\}$. \QEDB
\end{df}

\begin{df}\label{sla}
For a storage type $S = (C,P,F,I,E,m)$,  \emp{$S$ with look-ahead} 
is the storage type $S_\LA = (C,P \cup P',F,I,E,m')$, where
$P' = \{\acc(\G) \mid \G$ is a $\CF(S)$ transducer$\}$, $m'$ restricted to
$P \cup F \cup E$ is equal to $m$, and, for every $c \in C$,
$m'(\acc(\G))(c) = \true$ iff $c \in \Acc(\G)$.  \QEDB
\end{df}

Note that $\Acc(\G) = A(\G')$, where $\G' = (N,\mathrm{en},\Delta,\Ain,R)$ is
a $\CF(S')$ transducer for $S' = (C,P,F,C,\{\mathrm{en}\},m)$, with
$m(\mathrm{en}) = \id(C)$. Thus the encoding $e$ of $\G$ is irrelevant
and its terminal alphabet can be taken empty.
This also shows, by Lemma~\ref{cfext_pf}, that in Definition~\ref{sla} we can
equivalently put $P' = \{\acc(\G) \mid \G$ is a $\CFext(S)$ transducer$\}$,
see the second remark after Lemma~\ref{cfext_pf}.

Since the domain of a top-down tree(-to-string)
transducer is a regular tree language \cite{Rou}, it follows that
$\Acc(\G)$, restricted to some $T_\Sigma$, is a regular tree language for
every $\CF$(Tree) transducer $\G$. Hence the $\RT$(Tree$_\LA$) transducer is
the \emp{top-down tree transducer with regular look-ahead} of \cite{Eng4}
(and the notions of determinism are the same).

We now show that a deterministic $\REG(\Pd(S))$ transducer can be simulated by a
deterministic $\CF(S_\LA)$ transducer.

\begin{theo}\label{theo5.3}
 For every storage type $S$,
$\tauDREG(\Pd(S)) \subseteq \tauDCF(S_\LA)$. 
\end{theo}

\begin{proof}
Let $\G = (N,e,\Delta,\Ain,R)$ be a deterministic
$\REG(\Pd(S))$ transducer,
and let $\G' = (N',e',\Delta,\Ain',R')$ be the $\CFext(S)$ transducer
constructed in the second part of the proof of Theorem~\ref{theo5.1},
where $N'$ consists of all triples
$\langle B,\delta,E\rangle$
with $B \in N$, $\delta \in \Gamma_\G$, and $E \in N \cup \{\omega\}$. 
Note that the assumptions on $\G$, made in that proof, preserve determinism. 
For every triple $\langle B,\delta,E\rangle$ with $E\in N$, 
and every instruction symbol $f$, define the $\CFext(S)$ transducer
$\G'(\langle B,\delta,E\rangle(f))=(N'\cup\{\underline{A}\},e',\Delta,\underline{A},\underline{R})$
where $\underline{A}$ is a new nonterminal and $\underline{R}$ consists of all rules of $R'$
plus the rule $\underline{A}\rightarrow \langle B,\delta,E\rangle(f)$.
Note that $\G'(\langle B,\delta,E\rangle(f))$ can be used as a look-ahead test,
as observed after Definition~\ref{sla}.
We now construct the $\CFext(S_\LA)$ transducer $\G''$ from $\G'$ by changing every rule 
$$\langle A,\gamma,C\rangle \rightarrow \ruleif b \rulethen
      w\langle B,\delta,E\rangle(f)\langle E,\gamma,C\rangle(\id)$$
that corresponds to a push-rule
$$A \rightarrow \ruleif \test(b) \text{ \bool{and} } \ttop{=}\,\gamma \rulethen wB(\push(\delta,f))$$
of $\G$, into the rule 
$$\langle A,\gamma,C\rangle \rightarrow 
      \ruleif b \text{ \bool{and} } \acc(\G'(\langle B,\delta,E\rangle(f))) \rulethen
      w\langle B,\delta,E\rangle(f)\langle E,\gamma,C\rangle(\id).$$
Since $\G$ is deterministic, 
and hence the tests $\acc(\G'(\langle B,\delta,E\rangle(f)))$ are mutually disjoint
(for fixed $B$, $\delta$, and $f$), it should now be clear that 
$\G''$ is deterministic and equivalent to $\G$. 

This proves that 
$\tauDREG(\Pd(S)) \subseteq \tauDCFext(S_\LA)=\tauDCF(S_\LA)$,
see the first remark after Lemma~\ref{cfext_pf}. 
\end{proof}

We now have $\tauDCF(S) \subseteq \tauDREG(\Pd(S)) \subseteq
\tauDCF(S_\LA) \subseteq \tauDREG(\Pd(S_\LA))$, and in general not more
can be said. For $S$ = Tree it can be shown that $\Pd$(Tree$_\LA$) is ``equivalent'' to
$\Pd$(Tree), see \cite{EngVog2}, and hence $\tauDREG(\Pd$(Tree)) =
$\tauDCF$(Tree$_\LA$), as shown in Theorem~4.7 of \cite{EngRozSlu} (cf.
next section). Also, there exist storage types $S$ that are ``closed
under look-ahead'', i.e., for which $S_\LA$ is ``equivalent'' to $S$ (see
\cite{EngVog2,Eng9} for this notion of equivalence). For such storage
types the nice equality $\tauDCF(S) = \tauDREG(\Pd(S))$ holds. As proved
in \cite{EngVog2}, an example of such a storage type is $\Pd$(Tree); this
was used in \cite{EngVog2} to show that deterministic \emp{macro
tree-to-string transducers} \cite{EngVog1,CouFra} are equivalent to
$\tauDCF(\Pd$(Tree)) = $\tauDREG(\Pd^2$(Tree)), i.e., to deterministic
$\REG(\Pd(\Pd$(Tree))) transducers. Another storage type closed under
look-ahead is Pushdown, i.e., $\Pd_\LA$ is ``equivalent'' to~P. In fact,
pushdown automata with look-ahead, i.e., $\REG(\Pd_\LA)$ r-acceptors,
are similar to the \emp{predicting machines} of (Section 10.3 of)
\cite{HopUll}. Note how funny the development of the notion of
look-ahead is: pushdown automata with look-ahead on the input
string (in parsing theory), top-down tree transducers with
look-ahead on the input tree, $\CF(S)$ transducers with look-ahead
on the storage (note that, very often, the storage is also the
input), pushdown automata with look-ahead on the pushdown (also
useful in parsing theory, as shown in \cite{EngVog4}). Anyway,
``look-ahead'' seems to be a useful tool in several parts of
formal language theory.

We have shown that $\tauDREG(\Pd(S))$ is included in $\tauCF(S)\cap \mathrm{PF}$, 
and that it is included in $\tauDCF(S_\LA)$. 
This suggests that maybe even $\tauCF(S)\cap \mathrm{PF} \subseteq \tauDCF(S_\LA)$.
We show that this holds provided $S$ is \emp{noetherian}, which means that 
there do not exists $c_i\in C$ and $f_i\in F$ for $i\in\N$ such that 
$m(f_i)(c_i)=c_{i+1}$ for every $i\in\N$. This property of $S$ ensures that 
a $\CF(S)$ transducer has no infinite derivations. The storage types One-way and Tree
are noetherian, but the storage types Pushdown and Tree-walk are not. 

For a noetherian storage type $S$, consider a $\CF(S)$ transducer $\G$ such that $T(\G)$ is a
partial function. At each moment of time during a derivation of
$\G$, several rules may be applicable. Since $T(\G)$ is a partial
function, it really does not matter which of these rules is
taken, as long as application of the rule leads to a successful
derivation. This can be detected by a look-ahead test, which allows to pick
one of the successful rules, deterministically. 
Thus, look-ahead can be used to turn semantic determinism into syntactic determinism.

\begin{theo}\label{theo5.2}
 Let $S$ be a noetherian storage type, and $\mathrm{PF}$ the class of all
partial functions. Then $\tauCF(S) \cap \mathrm{PF} \subseteq \tauDCF(S_\LA)$.
\end{theo}

\begin{proof}
 The proof is entirely analogous to the one for the
case $S$ = Tree in \cite{Eng5}. Let $\G = (N,e,\Delta,\Ain,R)$ be a $\CF(S)$
transducer such that $T(\G)$ is a partial function. For every rule
$r: A \rightarrow \ruleif b \rulethen \xi$ of $R$, define the $\CF(S)$
transducer $\G(r) = (N\cup\{\underline{A}\},e,\Delta,\underline{A},\underline{R})$, 
where $\underline{R}$ consists of all rules of $R$
plus the rule $\underline{A} \rightarrow \ruleif b \rulethen \xi$. 
The acceptor $\G(r)$ will be used
as a look-ahead test, to see whether $r$ leads to a successful derivation:
$$\Acc(\G(r)) = \{c \in C \mid A(c) \Rightarrow^r_\G \xi' \Rightarrow^*_\G w \text{ for some } 
w \in \Delta^* \text{ and } \xi' \in (N(F)\cup\Delta)^*\},$$ 
where $\Rightarrow^r_\G$ means that the derivation step
consists of the application of $r$. We now construct the
deterministic $\CF(S_\LA)$ transducer $\G' = (N,e,\Delta,\Ain,R')$, where $R'$
is defined as follows. For a given nonterminal $A$, let
$r_1: A \rightarrow \ruleif b_1 \rulethen \xi_1,\ldots, r_k: A \rightarrow
\ruleif b_k \rulethen \xi_k$ be all rules
in $R$ with left-hand side $A$, numbered arbitrarily from 1 to $k$.
Let $p_i = \acc(\G(r_i))$, a predicate symbol of $S_\LA$. For every test
$d = q_1\text{ \bool{and} }q_2\text{ \bool{and} }\cdots\text{ \bool{and} }q_k$ 
with $q_i \in \{p_i,\text{ \bool{not}}\,p_i\}$ for all $i$, the
rule $A \rightarrow \ruleif d \rulethen \xi_{i(d)}$ is in $R'$, where $i(d)$ is
the first
integer $i$, $1 \leq i \leq k$, such that $q_i = p_i$ (if there is such an $i$).
For this $\G'$, $T(\G') = T(\G)$. 
\end{proof}

\vspace{1em}

\begin{newob}\rm
One of the most stupid mistakes that I have made in my entire life 
(I~mean, mathematical mistakes),
is to think that the proof of Theorem~\ref{theo5.2} works for arbitrary storage types:
in the original version of this report the noetherian condition was not present 
in Theorem~\ref{theo5.2}, and 
Theorem~\ref{theo5.3} was derived from Lemma~\ref{cfext_pf} and Theorem~\ref{theo5.2}.
This stupid mistake was repeated for two-way pebble transducers 
in the proof of Theorem~3 of~\cite{*EngMan}, which is therefore wrong.
Fortunately, Hendrik Jan Hoogeboom discovered the mistake in 2006. 
I keep wondering about all my mistakes that were not discovered yet $\dots$. 
\end{newob}


\newpage
\section{\sect{Specific pushdown machines}}\label{sect6}

Applying the theorems of the previous section to the
specific storage types discussed in Sections \ref{sect3} and \ref{sect4},
immediately gives us several known pushdown machine
characterizations of the corresponding formalisms. The idea is
to convince the reader that these characterizations ``trivially''
follow from the fact that the formalisms are (or: can be viewed
as) context-free $S$ grammars. Let us discuss them one by one (as
numbered in Section \ref{sect3}).

\vspace{1.5em}

(1) $S = S_0$. Since $S_0$ has an identity,
$\lambdaCF(S_0) = \lambdaREG(\Pd(S_0))$. This shows that the
context-free grammar and the one-way pushdown automaton are equivalent
(surprise!).

\vspace{1.5em}

(2) $S$ = Pushdown. Since Pushdown has an identity, it follows from 
Theorem \ref{theo5.1}(2) that $\lambdaCF(\Pd) = \lambdaREG(\Pd(\Pd))$. In other words:
the indexed grammar is equivalent to the one-way $\Pd(\Pd)$ automaton.
Since a storage configuration of this automaton is a pushdown of
pushdowns (see Fig.~4), we also call this the \emp{one-way pushdown$^2$
automaton} (or $\Pd^2$ automaton). This result was first mentioned in
\cite{Gre}, then its proof was sketched in \cite{Mas1,Mas2}, and finally it
was proved in \cite{ParDusSpe1}, where the $\Pd^2$ automaton is defined as
\emp{indexed pushdown automaton} (which means that with each pushdown
symbol a sequence of flags is associated). As mentioned in
\cite{Mas1,Mas2}, and shown in \cite{Eng9,EngVog2}, the P$^2$ automaton is
equivalent to (and rather close to) the \emp{nested stack automaton}
of \cite{Aho2} (the nested stack may be viewed as a space
optimization of the pushdown of pushdowns). Thus, this fits with
the equivalence of the indexed grammar and the nested stack
automaton, established in \cite{Aho2}.

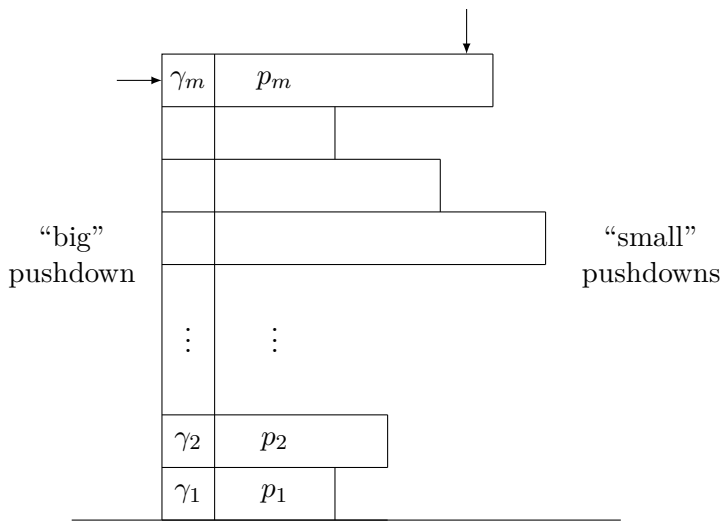
\begin{figure}[h]
\centering
\begin{tikzpicture}

\draw (0,0) -- (3,0);
\draw (4.4,6.2) -- (0,6.2) --(0,0);
\draw (0,0.7) -- (3,0.7);
\draw (0,1.4) -- (3,1.4);
\draw (0,5.5) -- (4.4,5.5);
\draw (0,4.8) -- (3.7,4.8);
\draw (0,4.1) -- (5.1,4.1);
\draw (0,3.4) -- (5.1,3.4);
\draw (0.7,0) -- (0.7,6.2);
\draw (-1.2,0) -- (6.1,0);
\draw (2.3,0) -- (2.3,0.7);
\draw (3,0.7) -- (3,1.4);
\draw (4.4,6.2) -- (4.4,5.5);
\draw (2.3,5.5) -- (2.3,4.8);
\draw (3.7,4.8) -- (3.7,4.1);
\draw (5.1,4.1) -- (5.1,3.4);

\draw[-latex] (-0.6,5.85) -- (0,5.85);
\draw[-latex] (4.05,6.8) -- (4.05,6.2);

\node[align=center] at (-1.2,3.5) {``big''\\ pushdown};
\node[align=center] at (6.5,3.5) {``small''\\ pushdowns};

\node at (0.35,0.35) {$\gamma_1$};
\node at (0.35,1.05) {$\gamma_2$};
\node at (0.35,5.85) {$\gamma_m$};

\node at (1.5,0.35) {$p_1$};
\node at (1.5,1.05) {$p_2$};
\node at (1.5,5.85) {$p_m$};
\node at (1.5,2.5) {$\vdots$};
\node at (0.35,2.5) {$\vdots$};

\end{tikzpicture}
 \caption{A configuration of $\Pd(\Pd)$: a pushdown of pushdowns. Each
``small'' pushdown $p_i$ is inside a cell of the ``big'' pushdown. The
``small'' pushdowns are drawn with their top to the right. Only $\gamma_m$ and
the top-cell of $p_m$ are accessible. The division of cells in the ``small''
pushdowns is not shown.}
\end{figure}

This storage type $\Pd^2$ might give you the idea to
consider $\CF(\Pd^2)$ grammars. By Theorem \ref{theo5.1}(2) again, these are
equivalent with the one-way $\Pd^3$ automata, etc. Define in general
$\Pd^{n+1}(S) = \Pd(\Pd^n(S))$ and $\Pd^0(S) = S$. Abbreviate $\Pd^n(S_0)$ by
$\Pd^n$. The $\CF(\Pd^n)$ grammar is rather close to the \emp{$n$-level indexed grammar} of
\cite{Mas1,Mas2}, restricted to left-most derivations. The one-way $\Pd^n$
automaton is the \emp{$n$-level pushdown automaton} of \cite{Mas1,Mas2,DamGoe},
called \emp{$n$-iterated pushdown automaton} in
\cite{Eng9,EngVog2,EngVog3,EngVog4}.
Theorem \ref{theo5.1}(2) implies that $\lambdaCF(\Pd^n) = \lambdaREG(\Pd^{n+1})$, 
as shown in \cite{Mas1,Mas2}. Thus, the $n$-level indexed grammars are equivalent to
the $(n+1)$-iterated pushdown automata.

\vspace{1.5em}

(3) $S$ = Counter. We get $\lambdaCF$(Counter) = $\lambdaREG(\Pd$(Counter)), 
a special case of the pushdown$^2$ automaton.

\vspace{1.5em}

(4) $S$ = Count-down. The one-way $\Pd$(Count-down) automaton
is the \emp{preset pushdown automaton} of \cite{vLe1}. It is just like an
ordinary pushdown automaton, except that at the beginning of
each computation the length of the pushdown is preset to some
number of cells. The automaton ``knows'' when it reaches this
``ceiling'', and can react to it. In fact, the number of cells is
chosen (nondeterministically) by the encoding, and is decreased
by one by each $\push(\gamma,\dec)$ instruction. The ``ceiling'' can be
detected by the predicate $\test(\nnull)$.

As an example (taken from \cite{vLe1}) we give a regular
$\Pd$(Count-down) transducer $\G_7$ with $L(\G_7)$ = $\{a^{n^2}\mid n \geq 1\}$.
When the pushdown is preset to $k$
cells (i.e., the encoding chooses $k-1$), $\G_7$ first generates $k$ $a$'s, by
moving to
the ceiling and back again. Then $\G_7$ pushes twice, and generates
$k-2$ $a$'s, by moving to the ceiling, and back to the same square,
marked for that purpose. Then $\G_7$ pushes twice, etc. In this way
$\G_7$ generates $k+(k-2)+\cdots+3+1$ $a$'s (assuming that $k$ is odd), i.e.,
all squares. Formally $\G_7 = (N,e,\{a\},\Ain,R)$, where $N = \{A,B,C\}$,
$\Ain = A$, $e = (\#,\mathrm{en})$, and $R$ consists of the following rules
(pushdown symbols $\#$ and $\$$ are used; $\#$ as marker, and $\$$ as a
dummy symbol).

\begin{center}\parbox{0cm}{
\begin{tabbing}
    \=$A $\=$\ar $\=$\ruleif\text{ \bool{not}}\test(\nnull) \rulethen
      aB(\push(\$,\dec)) \ruleelse a$\\[1mm]
    \>$B $\=$\ar $\=$\ruleif\text{ \bool{not}}\test(\nnull) \rulethen
      aB(\push(\$,\dec)) \ruleelse aC(\pop)$\\[1mm]
    \>$C $\=$\ar $\=$\ruleif\text{ \bool{not}}\ttop{=}\,\# \rulethen
      C(\pop) \ruleelse A(\push(\#,\dec);\push(\#,\dec))$
\end{tabbing}}
\end{center}

\noindent This concludes the example.

By Theorem \ref{theo5.1}(1), $\lambdaREG(\Pd$(Count-down)) =
$\lambdaCFext$(Count-down). In other words, the preset pushdown
automaton is equivalent to what might be called the extended EOL
system. It is easy to see (cf. the proof of Lemma~\ref{cfext_pf}) that 
this extended EOL system can be viewed as an EOL
system in which the right-hand sides of rules are regular
languages: a so-called $\REG$-iteration grammar, with one
substitution (see \cite{RozSal, Asv, Eng3}; in \cite{vLe1} the
corresponding class of languages is called the hyper-algebraic
extension of $\REG$). This equivalence of preset pushdown automata
to $\REG$-iteration grammars was proved in \cite{vLe1}. To obtain an
automaton equivalent to the EOL system, the preset pushdown
automaton should be restricted to have the bounded excursion
property (see \cite{vLe1, EngVog2}).

\vspace{1.5em}

(5) $S$ = One-way. By a rather easy argument it can be
shown that $\lambdaCFext$(One-way) = $\lambdaCF$(One-way), 
cf. Lemma 3.3.2 of \cite{EngRozSlu}. 
Thus Theorem \ref{theo5.1}(1) implies that the one-way
$\Pd$(One-way) automaton is equivalent to the ETOL system. In fact,
this automaton is the \emp{checking-stack/pushdown automaton} (CS-PD automaton),
introduced in \cite{vLe2} where this equivalence was proved (see also
\cite{EngSchvLe, EngRozSlu, RozSal}). To see that the $\REG(\Pd$(One-way))
\mbox{r-acceptor} is the CS-PD automaton, note that the pushdown of a
$\REG(\Pd$(One-way)) r-acceptor $\G$ is always of the form 
$(\gamma_m,c_m)\cdots (\gamma_2,c_2)(\gamma_1,c_1)$

\begin{tabbing}
 with \; \=$c_1$\; \=$=\;a_1$\= $a_2\cdots$\= $\cdots a_n$\\
 \>$c_2$\>$=$ \>$a_2\cdots$\>$\cdots a_n$\\
 \>\> $\cdots$\\
 \>$c_m$\>$=$\>$\quad \, a_m$\>$\cdots a_n$
\end{tabbing}

\noindent where $a_1 a_2\cdots a_n$ is the string ``guessed'' by the encoding
of $\G$ at the beginning of its computation. Another, obviously
equivalent, way of representing this storage configuration is by
an ordinary pushdown $\gamma_m\cdots \gamma_2\gamma_1$ 
and a checking stack $a_1 a_2\cdots a_n$, with their reading
heads combined into one, pointing to $\gamma_m$ and $a_m$, see Fig.~5. 
The one reading head moves
in a two-way fashion over the checking stack, synchronously with
the movements of the top of the pushdown. A $\push(\gamma,\rread)$
instruction moves the reading head to the right, and a $\pop$
instruction moves it to the left. This is precisely the
behavior of a CS-PD automaton.

It can be shown that $\Pd$(One-way$_\LA$) is ``equivalent'' to
$\Pd$(One-way): it is the monadic case of $\Pd$(Tree$_\LA$) = $\Pd$(Tree), see
\cite{EngVog2}. Hence $\lambdaDREG(\Pd$(One-way)) =
$\lambdaDCF$(One-way$_\LA$), cf. the
discussion after Theorem \ref{theo5.3}. It can also easily be shown that
$\lambdaDCF$(One-way$_\LA$) = $\lambdaDCF$(One-way), cf. Lemma 3.3.4 of
\cite{EngRozSlu}. Hence $\lambdaDREG(\Pd$(One-way)) is the class of EDTOL
languages. Note that the D stands for transducer determinism
(thus EDTOL does not correspond to the deterministic CS-PD
automaton), cf. Section 5 of \cite{EngSchvLe}.

Let us consider $\Pd$(One-way) with just one pushdown
symbol, i.e., $\Pd_{\p}$(One-way). It should be clear from the
discussion above that $\Pd_{\p}$(One-way) could be called \emp{Two-way}: the
storage type corresponding to a two-way read-only tape (with
endmarkers). Thus, the $\REG(\Pd_{\p}$(One-way)) transducer is the
\emp{two-way finite state transducer} or two-way gsm, the
$\CF(\Pd_{\p}$(One-way)) d-acceptor is the \emp{alternating two-way finite
automaton}, and the $\REG(\Pd_{\p}$(One-way)) r-acceptor is the one-way
\emp{checking stack automaton}.

\vspace{0.5cm}

\begin{figure}[h]
\centering
\begin{tikzpicture}

\draw (0,1.8) --(0,-0.4);
\draw (0,0.7) -- (3.7,0.7);
\draw (0,1.4) -- (3.7,1.4);
\draw (0.7,0) -- (0.7,1.4);
\draw (1.4,0) -- (1.4,1.4);
\draw (0,0) -- (3.7,0);
\draw (3,0) -- (3,1.4);
\draw (3.7,0) -- (3.7,1.4);

\draw[-latex] (3.35,-0.6) -- (3.35,0);
\draw[-latex] (3.35,2) -- (3.35,1.4);

\node[align=center] at (-0.8,0.75) {bottom};
\node[align=center] at (4.2,0.75) {top};

\node at (0.35,0.35) {$c_1$};
\node at (0.35,1.05) {$\gamma_1$};
\node at (1.05,0.35) {$c_2$};
\node at (1.05,1.05) {$\gamma_2$};
\node at (2.2,0.35) {$\dots$};
\node at (2.2,1.05) {$\dots$};
\node at (3.35,0.35) {$c_m$};
\node at (3.35,1.05) {$\gamma_m$};

\node at (5.4,2.3) {with};

\draw (8,1.9) --(8,2.6);
\draw (8.7,1.9) --(8.7,2.6);
\draw (9.4,1.9) --(9.4,2.6);
\draw (11,1.9) --(11,2.6);
\draw (11.7,1.9) --(11.7,2.6);
\draw (8,1.9) --(11.7,1.9);
\draw (8,2.6) --(11.7,2.6);

\node at (7.2,2.25) {$c_1 =$};
\node at (8.35,2.25) {$a_1$};
\node at (9.05,2.25) {$a_2$};
\node at (10.2,2.25) {$\dots$};
\node at (11.35,2.25) {$a_n$};

\draw (8.7,0.8) --(8.7,1.5);
\draw (9.4,0.8) --(9.4,1.5);
\draw (11,0.8) --(11,1.5);
\draw (11.7,0.8) --(11.7,1.5);
\draw (8.7,0.8) --(11.7,0.8);
\draw (8.7,1.5) --(11.7,1.5);

\node at (7.2,1.15) {$c_2 =$};
\node at (9.05,1.15) {$a_2$};
\node at (10.2,1.15) {$\dots$};
\node at (11.35,1.15) {$a_n$};

\node at (7.2,0.15) {$\vdots$};

\draw (10.1,-0.5) --(10.1,-1.2);
\draw (9.4,-0.5) --(9.4,-1.2);
\draw (11,-0.5) --(11,-1.2);
\draw (11.7,-0.5) --(11.7,-1.2);
\draw (9.4,-0.5) --(11.7,-0.5);
\draw (9.4,-1.2) --(11.7,-1.2);

\node at (7.2,-0.85) {$c_m =$};
\node at (9.75,-0.85) {$a_m$};
\node at (10.55,-0.85) {$\dots$};
\node at (11.35,-0.85) {$a_n$};

\end{tikzpicture}
 \captionsetup{labelformat=bf-parens, name=Fig. 5(a)}
 \caption{Configuration of P(One-way).}
\end{figure}
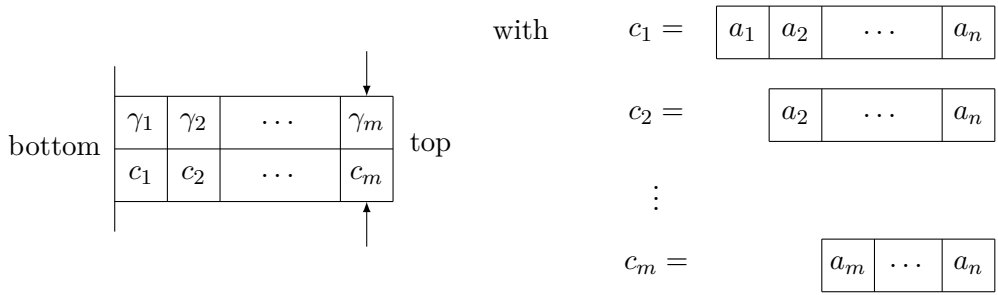

\vspace{0.5cm}

\begin{figure}[h]
\centering
\begin{tikzpicture}

\draw (0,1.75) --(0,0.35);
\draw (0,0.7) -- (3.7,0.7);
\draw (0,1.4) -- (3.7,1.4);

\draw (0.7,0.7) -- (0.7, 1.4);
\draw (1.4,0.7) -- (1.4,1.4);

\draw (0.7,-0.7) -- (0.7, 0);
\draw (1.4,-0.7) -- (1.4,0);

\draw (3,0.7) -- (3,1.4);
\draw (3.7,0.7) -- (3.7,1.4);

\draw (0,0) -- (6,0);
\draw (0,-0.7) -- (6,-0.7);
\draw (3,-0.7) -- (3,0);
\draw (0,-0.7) -- (0,0);
\draw (3.7,-0.7) -- (3.7,0);
\draw (5.3,-0.7) -- (5.3,0);
\draw (6,-0.7) -- (6,0);

\draw[-latex] (3.35,0) -- (3.35,0.7);
\draw[-latex] (3.35,0.7) -- (3.35,0);

\node[align=center] at (-0.8,1.05) {bottom};
\node[align=center] at (4.2,1.0) {top};

\node at (0.35,-0.35) {$a_1$};
\node at (0.35,1.05) {$\gamma_1$};
\node at (1.05,-0.35) {$a_2$};
\node at (1.05,1.05) {$\gamma_2$};
\node at (2.2,-0.35) {$\dots$};
\node at (4.5,-0.35) {$\dots$};
\node at (2.2,1.05) {$\dots$};
\node at (3.35,-0.35) {$a_m$};
\node at (5.65,-0.35) {$a_n$};
\node at (3.35,1.05) {$\gamma_m$};

\end{tikzpicture}
 \captionsetup{labelformat=bf-parens, name=Fig. 5(b)}
 \caption{Configuration of CS-PD.}
\end{figure}
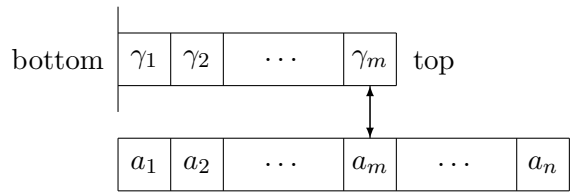
\addtocounter{figure}{-1}

\vspace{1.5em}

(6) $S$ = Tree. Theorem \ref{theo5.1}(1) shows that
$\tauCFext$(Tree) = $\tauREG(\Pd$(Tree)); see Theorem 4.5 of \cite{EngRozSlu}.
The $\CFext$(Tree) transducer is the \emp{regularly extended top-down
tree transducer}. Moreover, the $\REG(\Pd$(Tree)) transducer is the
\emp{checking-tree/pushdown transducer} (CT-PD transducer). To see
this, note that, as for the previous case of $S$ = One-way, the
pushdown of a $\REG(\Pd$(Tree)) transducer is always of the form
$(\gamma_m,t_m)\cdots(\gamma_2,t_2)(\gamma_1,t_1)$ where $t_1$ is the input tree
(as given by the encoding), and $t_{i+1}$ is a direct subtree of $t_i$, for all
$i$.
Therefore, this storage configuration can also be represented by
an ordinary pushdown (viz. $\gamma_m\cdots\gamma_2\gamma_1$) and a tree (viz.
$t_1$) with a pointer to one of its nodes (viz. the one corresponding to the
subtree $t_m$ of $t_1$). Again, the pushdown pointer and the tree
pointer are combined into one: a $\push(\gamma,\sel_i)$ moves the pointer
down in the tree to the $i$-th son of the current node, and a $\pop$
moves it up to its father. Thus, the length of the pushdown
equals the length of the path from the root to the current node
of the input tree. This is precisely the storage type of the
CT-PD transducer, see Fig.~6. In programming terminology it is
just the familiar fact that tree walking (CT-PD) can be
implemented by a pushdown of pointers to the tree ($\Pd$(Tree)).

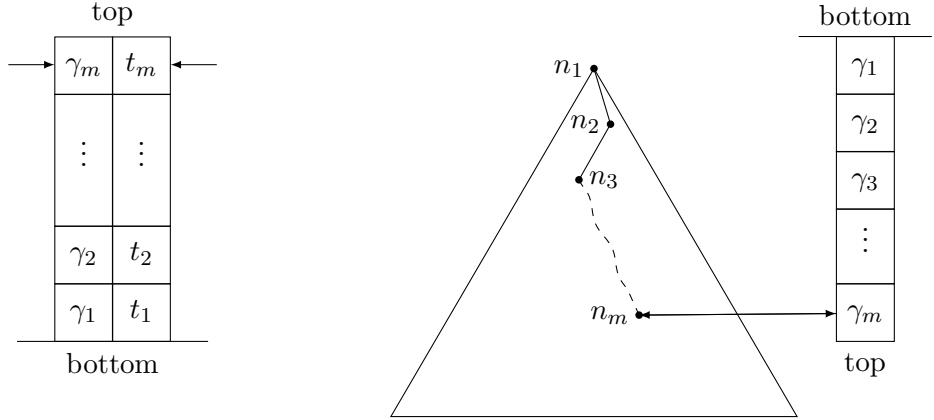
\begin{figure}[h]
\centering
\begin{tikzpicture}

\draw (6.4,3.63) -- (3.7,-1) -- (9.1,-1) -- cycle;

\coordinate[label=left:$n_1$] (A) at (6.4,3.63);

\coordinate[label=left:$n_2$] (B) at (6.62,2.89);

\coordinate[label=right:$n_3$] (C) at(6.2,2.15) ;

\coordinate[label=left:$n_m$] (D) at (7,0.35);

\fill (A) circle (0.05);

\fill (B) circle (0.05);

\fill (C) circle (0.05);

\fill (D) circle (0.05);

\draw (A) -- (B) -- (C);

\path [dashed,-, decoration={snake,segment length=22,amplitude=1},
  line join=round] (C) edge[decorate] (D);

\matrix(M1) at (0em, -\pgflinewidth)
[ inner sep=0
, anchor=south
, matrix of nodes
, row sep=-\pgflinewidth
, column sep=-\pgflinewidth
, every node/.style=
 { draw
 , minimum width=2em
 , text height=1.18em
 , text depth=0. 81em
 }
, row 2/.style={minimum height=4.56em}
]{
$\gamma_m$ &  $t_m$
\\
 \vdots & \vdots
\\
$\gamma_2$ &  $t_2$
\\
 $\gamma_1$ &  $t_1$
\\};

\draw (-3.3em, 0) -- ++(6.6em, 0);

\matrix(M2) at (26em, -\pgflinewidth)
[ inner sep=0
, anchor=south
, matrix of nodes
, row sep=-\pgflinewidth
, column sep=-\pgflinewidth
, every node/.style=
 { draw
 , minimum width=2em
 , text height=1.18em
 , text depth=0. 81em
 }
, row 4/.style={minimum height=2.58em}
]{
$\gamma_1$ 
\\
$\gamma_2$ 
\\
$\gamma_3$ 
\\
 \vdots 
\\
 $\gamma_m$ 
\\};

\draw (23.7em, 10.54em) -- ++(4.6em, 0);

\draw [-latex](D) -- ($(M2.south)+(-1em,1em)$);

\draw [-latex]($(M2.south)+(-1em,1em)$) -- (D);

\node[align=center, below=0em of M1](N1)  {bottom};

\node[align=center, above=0em of M1](N2) {top};

\node[align=center, below=0em of M2](N3) {top};

\node[align=center, above=0em of M2](N4) {bottom};

\draw[-latex] ($(M1.north)+(3.6em,-1em)$)--($(M1.north)+(2em,-1em)$);

\draw[-latex] ($(M1.north)+(-3.6em,-1em)$)--($(M1.north)+(-2em,-1em)$);

\end{tikzpicture}
 \captionsetup{labelformat=bf-parens, name=Fig. 6}
 \caption{To the left a configuration of P(Tree), and to the right a
configuration of CT-PD, where $n_i$ is the root of the subtree $t_i$, and
$n_{i+1}$ is a son of $n_i$.}
\end{figure}

Note that a $\Pd$(Tree) configuration can be represented by
several CT-PD configurations: a sequence of consecutive direct
subtrees does not uniquely determine the path in the tree
(think, e.g., of a full binary tree over $\{\sigma,a\}$, with $\rk(\sigma) =
2$ and $\rk(a) = 0$). However, there is still a one-to-one
correspondence between the instructions and predicates of
$\Pd$(Tree) and those of CT-PD. In fact, the \emp{storage type CT-PD}
can easily be defined, as a generalization of Tree-walk. Its
instructions are: $\down_i(\gamma)$, to move down to the $i$-th son and
push $\gamma$; $\up$, to move up to the father and pop; $\stay(\gamma)$, to stay
at the same node and change the top of the pushdown to $\gamma$; and
$\stay$, the identity. Its predicates are: $\ttop{=}\,\gamma$, to test the top
of the pushdown; $\llabel{=}\,\sigma$, to test the label of the current node;
and $\rroot$, to see whether the current node is the root of the
input tree. The correspondence between the instructions and
predicates of CT-PD and those of $\Pd$(Tree) is as follows.
\begin{tabbing}
\quad \=CT-PD: \qquad \=$\down_i(\gamma)$ \qquad \=$\up$ \qquad \=$\stay(\gamma)$ \quad \=$\stay$ \quad \=$\ttop{=}\,\gamma$ \quad \=$\llabel{=}\,\sigma$ \hspace{1cm} \=$\rroot$ \\[2mm]
\>$\Pd$(Tree): \>$\push(\gamma,\sel_i)$ \>$\pop$ \>$\stay(\gamma)$ \>$\stay$ \>$\ttop{=}\,\gamma$ \>$\test(\rroot{=}\,\sigma)$ \>$\bottom$
\end{tabbing}

\noindent Thus, we may identify CT-PD and $\Pd$(Tree).

As noted in Section \ref{sect5}, the storage types $\Pd$(Tree$_\LA$) and
$\Pd$(Tree) are equivalent, and hence $\tauDCF$(Tree$_\LA$) =
$\tauDREG(\Pd$(Tree)), i.e., the deterministic top-down tree
transducer with regular look-ahead is equivalent to the
deterministic CT-PD transducer, as shown in Theorem~4.7 of \cite{EngRozSlu}.

The storage type $\Pd$(Tree) might give the reader the idea
to consider $\CF(\Pd$(Tree)) transducers: by Theorem \ref{theo5.1}(2) they are
equivalent to $\REG(\Pd^2$(Tree)) transducers. Since $\Pd$(Tree) is
``closed under look-ahead'', this also holds for the deterministic
transducers. In fact, both the $\RT(\Pd$(Tree)) transducer and the
$\REG(\Pd^2$(Tree)) transducer are studied in \cite{EngVog2}, where they
are called \emp{indexed tree transducer} (because of the similarity to
the $\CF(\Pd)$ grammar, i.e., the indexed grammar) and \emp{pushdown$^2$
tree-to-string transducer}, respectively. Let us point out
(again) a correspondence between alternation and tree
transducers. In (5) of this section we noted that the
alternating two-way finite automaton is the $\CF(\Pd_{\p}$(One-way))
d-acceptor. Since One-way is the monadic case of Tree, $\Pd_{\p}$ is a
special case of $\Pd$, and $\CF$ is the yield of $\RT$, it follows that
every $\CF(\Pd_{\p}$(One-way)) d-acceptor may be viewed as an $\RT(\Pd$(Tree))
d-acceptor: $\alphaCF(\Pd_{\p}$(One-way)) $\subseteq \alphaRT(\Pd$(Tree)). 
Thus, the alternating two-way finite automaton languages are domains of
indexed tree transducers. In Lemma 6.11 of \cite{EngVog2} it is shown
that $\alphaRT(\Pd$(Tree)) = $\alphaRT$(Tree), which equals $\RT$. 
From this follows
the fact, shown in \cite{LadLipSto}, that $\alphaCF(\Pd_{\p}$(One-way)) =
$\REG$. Thus, regularity of the domains of indexed tree transducers implies,
``immediately'', regularity of the languages accepted by
alternating two-way finite automata.

Since $\alphaRT(\Pd$(Tree)) = $\alphaRT$(Tree) = $\RT$, it follows that
also $\alphaREG(\Pd$(Tree)) = $\RT$ (because $\alphaCF$(Tree) $\subseteq
\alphaREG(\Pd$(Tree)) $\subseteq \alphaCF(\Pd$(Tree)) ). 
Hence, the $\REG(\Pd$(Tree)) \mbox{d-acceptor} accepts the
regular tree languages; it is a tree walking automaton with a
synchronized pushdown. 
Thus, by New Observation~\ref{bojcol},
$\alphaREG$(Tree-walk) is properly included in $\alphaREG(\Pd$(Tree)). 
Note that the
class $\alphaCF(\Pd$(Tree)) of domains of indexed tree transducers may also
be viewed as the class of tree languages accepted by the
$\CF(\Pd$(Tree)) d-acceptor, i.e., the alternating $\Pd$(Tree) automaton;
from this point of view $\alphaCF(\Pd$(Tree)) = $\RT$ was shown in
\cite{Slu}.

\vspace{1.5em}

(7) $S$ = Tree-walk. From the previous discussion it
should now be clear that, apart from the $\son=i$ predicates,
Tree-walk is the same as $\Pd_{\p}$(Tree). In fact, using the natural
numbers as pushdown symbols, Tree-walk can be simulated by
$\Pd$(Tree) as follows: $\down_i$ corresponds to $\push(i,\sel_i)$, up to
$\pop$, $\stay$ to $\stay$, $\llabel{=}\,\sigma$ to $\test(\rroot{=}\,\sigma)$, $\son{=}\,i$
to $\ttop{=}\,i$, and $\rroot$ to $\bottom$. Thus, $\tauDRT$(Tree-walk)
$\subseteq \tauDRT(\Pd$(Tree)), i.e.,
the attribute grammar is a special case of the deterministic
indexed tree transducer. Similarly $\tauREG$(Tree-walk) $\subseteq$
\mbox{$\tauREG(\Pd$(Tree))}, i.e., the tree walking automaton of \cite{AhoUll} is
a special case of the \mbox{CT-PD} transducer. Note that Two-way =
$\Pd_{\p}$(One-way) is the monadic case of Tree-walk (in the monadic
case the $\son{=}\,i$ predicates are superfluous).

Theorem \ref{theo5.1}(2) shows that $\tauCF$(Tree-walk) =
$\tauREG(\Pd$(Tree-walk)). Although the deterministic case of this
equation also involves look-ahead, it should be clear that there
is a close relationship between the attribute grammar (with
strings as values) and the deterministic $\REG(\Pd$(Tree-walk))
transducer. This relationship was pointed out in \cite{Kam}, where
the $\REG(\Pd$(Tree-walk)) transducer is called the \emp{tree-walking
(synchronized) pushdown tree-to-string transducer}. Intuitively,
it is a tree-walking automaton that can back-track on the path
it has walked (cf. also the description of $\Pd(\Pd$(Tree)) in
\cite{EngVog2}). Since $\Pd$(Tree-walk) can be simulated by $\Pd^2$(Tree), it
is a special case of the pushdown$^2$ tree-to-string transducer.
The alternating version of the corresponding d-acceptor was
shown to accept the regular tree languages in \cite{Slu}. This fits
with the fact that domains of $\RT(\Pd^2$(Tree)) transducers are
regular (cf.\cite{EngVog2,EngVog3}).

Since trees are strings, $\tauDRT$(Tree-walk) $\subseteq
\tauCF$(Tree-walk) = $\tauREG(\Pd$(Tree-walk)). In particular when we
would write trees in postfix rather than prefix notation (i.e.,
$t_1\cdots t_k\sigma$ rather than $\sigma t_1\cdots t_k$) it would be quite
natural to output a tree symbol by symbol on the output tape of a
$\REG(\Pd$(Tree-walk)) transducer: these symbols may be viewed as
instructions to operate on a stack of (attribute) values, in the
usual way. Note however, that this is still the same, very
inefficient, way of evaluating attributes.

It may be worthwile to see what happens if one extends
attribute grammars (in their usual notation) to ``attribute
grammars with pushdown'', which are equivalent to DRT($\Pd$(Tree))
transducers (i.e., deterministic indexed tree transducers), in
the same way as attribute grammars are equivalent to
DRT(Tree-walk) transducers.

\vspace{1.5em}

(8) Since Tree-pushdown is closely related to Pushdown,
$\Pd$(Tree-pushdown) is closely related to $\Pd^2$. We note here that
Tree-pushdown has been generalized to an operator Tree-pushdown
of $S$, abbreviated $\TP(S)$, in \cite{EngVog2}. The $\RT(\TP$(Tree))
transducer with one state only is the \emp{macro tree transducer} of
\cite{CouFra,Eng6,EngVog1}, and it is closely related to the
$\RT(\Pd$(Tree)) transducer, i.e., the indexed tree transducer (see
\cite{EngVog2}).

\vspace{1.5em}

(9) The last storage type we consider is $\SPACE(f)$, as
discussed in Section \ref{sect4}. By Theorem~\ref{theo5.1}(2), 
$\alphaCF(\SPACE(f)) = \alphaREG(\Pd(\SPACE(f)))$. We now note that it
is quite easy to see that the storage type $\Pd(\SPACE(f))$
can be simulated by $\SPACE(f) \times \Pd$ (see Definition~\ref{prod_type}
for the product of storage types). In fact, a pushdown
$(\gamma_m,c_m)(\gamma_{m-1},c_{m-1})\cdots(\gamma_1,c_1)$
where $c_m,c_{m-1},\ldots,c_1$ are $\SPACE(f)$ configurations, can be
represented by the $\SPACE(f)$ configuration $c_m$ and the ordinary
pushdown $\gamma_{m-1}c_{m-1}\cdots\gamma_1 c_1$ (where it is understood that
the $c_i$ are put on the pushdown, coded as strings in the usual way);
$\gamma_m$ can be kept in the control of the involved transducer. Thus we obtain that
\mbox{$\alphaCF(\SPACE(f))$} 
$\subseteq \alphaREG(\SPACE(f) \times \Pd))$. 
In other words, the alternating $\SPACE(f)$ Turing machine can be simulated by
the nondeterministic \emp{$\SPACE(f)$ auxiliary pushdown automaton} \cite{Coo}. 
In fact, the above inclusion is an equality: the well-known relationship
between alternating Turing machines and auxiliary pushdown
automata (see \cite{Ruz}). The inclusion in the other direction requires 
a different construction.

More generally, 
$\alphaCF(\SPACE(f)\times S) = \alphaREG(\SPACE(f)\times \Pd(S))$
for every storage type $S$ that has an identity, 
see Theorem 2 of \cite{Eng9}. Note
that, in this equation, $\SPACE(f)$ cannot be replaced by One-way$_{\id}$. In
fact, \mbox{$\alphaREG$(One-way$_{\id}\times\Pd(S)) =$} $\lambdaREG(\Pd(S)) =
\lambdaCF(S)$, and we have
seen at the end of Section~\ref{sect4} (first observation) that, in
general, $\alphaCF$(One-way$_{\id}\times S$) and $\lambdaCF(S)$ are not equal.

\begin{newob}\rm
Another stupid mistake in the original version of this report, was the statement that
it is quite easy to see that $\Pd(\SPACE(f))$ and $\SPACE(f) \times \Pd$ are
equivalent storage types. This was a typical case of wishful thinking: 
it would have been so nice if Theorem~\ref{theo5.1} would have ``explained'' the 
relationship between alternating Turing machines and auxiliary pushdown
automata. The mistake was discovered by Roland Bol in 1987, when he attended 
my lectures on the subject. 
\end{newob}


\section{\sect{Deterministic r-acceptors}}\label{sect7}

In the next section we will show that, for a restricted
kind of storage type $S$, the languages generated by $\CF(S)$
grammars can be obtained from the languages accepted by
deterministic one-way $S$ automata, by a certain class of
operations on languages. Therefore we give in this section the
definition of determinism for $\REG(S)$ r-acceptors (cf.
Section \ref{sect2}). We call this \mbox{\emp{r-acceptor determinism}} 
to distinguish it from the (transducer) determinism discussed upto now. It has
to be admitted at this point that the deterministic r-acceptors
do not fit so nicely in our framework of $\CF(S)$ grammars, due to
the different ways in which they end their computations.

What is a deterministic $\REG(S)$ r-acceptor? First of all
we have to require that the encoding cannot be used to ``guess''
an initial configuration; we formulate this as a property of the
storage type (and since we consider one-way $S$ automata, we also
require an identity; see Section \ref{sect2}).

\begin{df}
 A storage type $S = (C,P,F,I,E,m)$ is \emp{r-acceptor
 deterministic} if $S$ has an identity, and $I$ is a singleton. \QEDB
\end{df}

Thus, for such a storage type, with, say, $I = \{u_0\}$, an
encoding $e$ just determines a fixed initial configuration
$m(e)(u_0)$. Most of the usual storage types for one-way $S$ automata
are r-acceptor deterministic. The particular ones discussed in
this paper are Counter and the iterated pushdown storage types
P$^n$, for $n \geq 0$. The results of the next section will be applied
to the one-way iterated pushdown automata.

We now define determinism of r-acceptors in the obvious
way. Recall from Section \ref{sect2} the intuitive interpretation of a
rule $A \rightarrow \ruleif b \rulethen wB(f)$ of a $\REG(S)$ r-acceptor $\G$.
Note that such a rule is applicable if $\G$ is in state $A$, $b$ holds for its
current storage configuration, and $w$ is a prefix of the rest of
the input. Determinism should ensure the applicability of at
most one rule in every situation. Recall also the notion of
normal form from Section \ref{sect2}: every rule is of the form 
$A \rightarrow \ruleif b \rulethen \xi$, where 
$\xi=aB(f)$ or $\xi=B(f)$ or $\xi=a$ or $\xi=\lambda$ (with $a\in\Delta$).

\begin{df}
 Let $S$ be an r-acceptor deterministic storage
type. A $\REG(S)$ transducer $\G = (N,e,\Delta,\Ain,R)$ is \emp{r-acceptor
deterministic} if the following two conditions hold.
\begin{itemize}
\setlength{\itemsep}{0pt}
 \item[(1)] $\G$ is in normal form.
 \item[(2)] If $A \rightarrow \ruleif b_1 \rulethen a_1 Q_1$ and 
      $A \rightarrow \ruleif b_2 \rulethen a_2 Q_2$ are two
      different rules in $R$ (with $a_i \in \Delta \cup \{\lambda\}$ and 
      $Q_i \in N(F) \cup \{\lambda\})$
      such that $a_1 = a_2$ or $a_1 = \lambda$ or $a_2 = \lambda$, 
      then $m(b_1\text{ \bool{and} }b_2)(c) = \false$ for every $c \in C$. \QEDB
\end{itemize}

\end{df}

Rather than ``r-acceptor deterministic $\REG(S)$
transducer'' we will also say ``deterministic $\REG(S)$ r-acceptor''.

An example of an r-acceptor deterministic $\REG(\Pd)$
transducer is $\G_2$, discussed in Sections~\ref{sect1.1} and~\ref{sect3}(2), in
which the third rule should be replaced by the rule $A \rightarrow bB(\pop)$.
This transducer is not deterministic. Thus, every deterministic
$\REG(S)$ transducer is r-acceptor deterministic, but not vice
versa. 

However, the deterministic $\REG(S)$ r-acceptors do not
yet correspond to the usual deterministic one-way $S$ automata.
The reason is that they decide deterministically when the input
string should end; hence they only accept prefix-free languages:
a language $L \subseteq \Delta^*$ is \emp{prefix-free} if $w \in L$ implies
$wv \notin L$ for all $v \in \Delta^+$. Thus, they accept, so to say, by empty
storage, and what we need is acceptance by final state.

\begin{df}
 Let $\G = (N,e,\Delta,\Ain,R)$ be a deterministic $\REG(S)$ r-acceptor. 
Then $L(\G)$ is called the \emp{language} of $\G$ \emp{accepted by empty store}. 
For a set $N_H \subseteq N$ of final states, the \emp{language} of $\G$
and $N_H$ \emp{accepted by final state}, denoted $L(\G,N_H)$, is defined by
$L(\G,N_H) = \{w \in \Delta^* \mid \Ain(m(e)(u_0)) \Rightarrow^*_G wA(c)$ for
some $A \in N_H$ and $c \in C\}$, where $I = \{u_0\}$. \QEDB
\end{df}

We denote by $\lambdaDeREG(S)$ the class $\{L(\G) \mid \G$ is a
deterministic $\REG(S)$ r-acceptor$\}$, and by $\lambdaDfREG(S)$ the class
$\{L(\G,N_H) \mid \G$ is a deterministic $\REG(S)$ r-acceptor and $N_H$ a
subset of its set of states$\}$. It should be clear to the reader
that
\begin{center}
``$\lambdaDfREG(S)$ = deterministic one-way $S$ automaton languages''.
\end{center}

\begin{lm}\label{DeREG-DfREG}
$\lambdaDeREG(S) \subseteq \lambdaDfREG(S) \subseteq \lambdaREG(S)$.
\end{lm}

\begin{proof}
 First inclusion: replace every rule $A \rightarrow \ruleif b \rulethen w$
(with $w \in \Delta^*$) by the rule $A \rightarrow \ruleif b \rulethen wQ(\id)$,
where $Q$ is a (new) final state, and $\id$ is the identity of $S$. Second
inclusion: remove all rules $A \rightarrow \ruleif b \rulethen w$ with $A \in N_H$
(which are superfluous when accepting by final state) and 
add rules $A \rightarrow \lambda$ for all $A \in N_H$.
\end{proof}

In fact, $\lambdaDeREG(S)$ is the class of prefix-free
languages in $\lambdaDfREG(S)$, a situation that is familiar from the
deterministic pushdown automata; proof: let $L(\G,N_H)$ be
prefix-free; for every $A \in N_H$, replace all rules with left-hand
side $A$ by the one rule $A \rightarrow \lambda$.

In the next section we will also show that the tree
languages accepted by deterministic top-down $S$ tree automata can
be obtained from the languages accepted by deterministic one-way
$S$ automata, by certain string-to-tree operations. Therefore we
also define determinism for $\RT(S)$ r-acceptors, in the obvious
way. Recall the notion of normal form from Section \ref{sect2}: 
every rule is of the form $A \rightarrow \ruleif b \rulethen \xi$, 
where $\xi=B(f)$ or $\xi=\sigma B_1(f_1)\cdots B_k(f_k)$ (with $\sigma\in\Delta_k$).

\begin{df}
 Let $S$ be an r-acceptor deterministic storage
type. An $\RT(S)$ transducer $\G = (N,e,\Delta,\Ain,R)$ is \emp{r-acceptor
deterministic} if the following two conditions hold.
\begin{itemize}
\setlength{\itemsep}{0pt}
 \item[(1)] $\G$ is in normal form.
 \item[(2)] If $A \rightarrow \ruleif b_1 \rulethen \sigma_1 Q_1$ and $A
\rightarrow \ruleif b_2 \rulethen \sigma_2 Q_2$ are two
different rules in $R$ (with $\sigma_i \in \Delta \cup \{\lambda\}$ and $Q_i \in
N(F)^*$) such that $\sigma_1 = \sigma_2$ or $\sigma_1 = \lambda$ or $\sigma_2 = \lambda$, 
then $m(b_1\text{ \bool{and} }b_2)(c) = \false$ for every $c \in C$.\QEDB
\end{itemize}
\end{df}

We denote by $\lambdaDeRT(S)$ the class $\{L(\G) \mid \G$ is a
deterministic $\RT(S)$ r-acceptor$\}$. For $S = S_0$, this is the class
of tree languages accepted by deterministic top-down finite tree
automata (rules $A \rightarrow B(\id)$ can easily be removed). For $S = \Pd$, it
is the class of tree languages accepted by deterministic
pushdown tree automata \cite{Gue2}.

Although tree languages are prefix-free ``by nature''
(viewing trees as strings), there is still a difference between
$\lambdaDeREG(S)$ and $\lambdaDeRT(S)$, due to the fact that there is
an empty string but no empty tree. In fact, a deterministic $\REG(S)$
r-acceptor can read the last symbol of the input string, and
then check some property of the storage configuration (reading
the empty string). On the other hand, a deterministic $\RT(S)$
r-acceptor has to decide whether to accept or reject as soon as
it sees a symbol of rank 0. We now define $\RT(S)$ r-acceptors that
can ``read beyond the leaves'' by way of a trick, as follows (see \cite{Gue1}).

\begin{df}
For a ranked alphabet $\Delta$, and $\# \notin \Delta$, $\Delta\#$ is the
ranked alphabet with $(\Delta\#)_0 = \{\#\}$, $(\Delta\#)_1 = \Delta_0 \cup \Delta_1$, and
$(\Delta\#)_k = \Delta_k$ for $k \geq 2$. 
The mapping mark: $T_\Delta \rightarrow T_{\Delta\#}$ is defined as
follows: for $t \in T_\Delta$, mark($t$) is the result of replacing, in $t$,
every $\sigma$ by $\sigma\#$, for all $\sigma \in \Delta_0$. For a tree language
$L$, mark($L$) = $\{$mark($t$) $\mid t \in L\}$. \QEDB
\end{df}

Thus every leaf $\sigma$ of a tree is made to be of rank $1$,
and a new leaf $\#$ is attached to it, i.e., it is replaced by the
tree $\sigma(\#)$.

We denote by $\lambdaDfRT(S)$ the class
$\{L \mid$ mark($L$) $\in \lambdaDeRT(S)\}$.

\begin{lm}\label{DeRT-DfRT}
 $\lambdaDeRT(S) \subseteq \lambdaDfRT(S) \subseteq \lambdaRT(S)$.
\end{lm}

\begin{proof}
First inclusion: replace every rule \mbox{$A \rightarrow \ruleif b \rulethen\sigma$}
with $\sigma \in \Delta_0$ by the rule \mbox{$A \rightarrow \ruleif b \rulethen
\sigma Q(\id)$}, where $Q$ is a new
nonterminal (and $\id$ is the identity of $S$), and add the rule
\mbox{$Q \rightarrow \#$}. Second inclusion: introduce new nonterminals
$B_\sigma$ for every nonterminal $B$ and every $\sigma \in \Delta_0$; replace
every rule
\mbox{$A \rightarrow \ruleif b \rulethen \sigma B(f)$}, with $\sigma\in\Delta_0$
(and so $\sigma\in (\Delta\#)_1)$, by the rule
\mbox{$A \rightarrow \ruleif b \rulethen B_\sigma(f)$}; 
if \mbox{$A \rightarrow \ruleif b \rulethen B(f)$} is a rule, then add
the rules \mbox{$A_\sigma \rightarrow \ruleif b \rulethen B_\sigma(f)$}, for
every $\sigma\in \Delta_0$; replace every
rule \mbox{$A \rightarrow \ruleif b \rulethen \#$} by the rules 
\mbox{$A_\sigma \rightarrow \ruleif b \rulethen \sigma$}, for every $\sigma \in \Delta_0$.
\end{proof}

In case $S$ is ``closed under look-ahead'' (e.g., $S = \Pd$ or
$S = S_0$), these two ways of acceptance are the same, i.e.,
$\lambdaDeRT(S) = \lambdaDfRT(S)$. To see this, we prove the following fact.

\begin{lm}
$\lambdaDfRT(S) \subseteq \lambdaDeRT(S_\LA)$.
\end{lm}

\begin{proof}
 Let $L \in \lambdaDfRT(S)$, and let $\G = (N,e,\Delta\#,\Ain,R)$ be a
deterministic $\RT(S)$ \mbox{r-acceptor} such that $L(\G)$ = mark($L$). Delete
every rule in which $\#$ occurs. Replace every rule
\mbox{$A \rightarrow \ruleif b \rulethen \sigma B(f)$}, with $\sigma \in
\Delta_0$, by the rule \mbox{$A \rightarrow \ruleif b\text{ \bool{and} }\acc(\G(B,f))
\rulethen \sigma$}. In this rule, $\G(B,f)$ is the $\RT(S)$ transducer $(N \cup
\{\underline{A}\},e,\Delta\#,\underline{A},\underline{R})$, where $\underline{A}$ is a new
initial nonterminal, and $\underline{R}$ is $R$ plus the rule $\underline{A} \rightarrow B(f)$.
\end{proof}

This concludes our discussion of determinism for
r-acceptors.


\newpage
\section{\sect{A new operation on languages}}\label{sect8}

One of the questions in formal language theory
(particularly in AFL/AFA theory) is whether one class of
languages can be obtained from another class by the application
of certain operations on languages. As an example, the
context-free languages can be obtained from the regular
languages by the application of nested iterated substitutions,
and the recursively enumerable languages can be obtained from
the (deterministic) context-free languages by the application of
intersections and homomorphisms (in both cases the resulting
class of languages is also closed under these operations). One
of the concrete questions in this area is whether it is possible
to obtain the indexed languages from the context-free languages
by a class of (natural) operations (see \cite{Gre}). The adjective
``natural'' is important here: one can, e.g., always view an
indexed grammar as an operation on languages. We give a partial
answer to this question: we define a class $\delta$ of unary operations
on languages, such that $\delta(\REG) = \CF$, and $\delta(\DCF)$ = Indexed,
where $\DCF$ denotes the class of deterministic context-free languages,
i.e., in our notation, $\lambdaDfREG(\Pd)$. (But $\CF$ and Indexed are not
closed under the $\delta$ operations; in fact, $\delta(\CF)$ = RE.) These are
two particular cases of our general result that $\lambdaCF(S) =
\delta(\lambdaDfREG(S))$, for r-acceptor deterministic $S$, shown in
Theorem \ref{theo8.1}(3). Thus the languages generated by $\CF(S)$ grammars
can be obtained by the $\delta$ operations from the languages accepted
by deterministic one-way $S$ automata.

The $\delta$ operations are defined by viewing strings as
paths through a tree, and taking the yield of that tree; thus
they incorporate the essence of the general philosophy in tree
language theory (for obtaining higher level devices), as
discussed in \cite{Eng1,Eng6, Dam}. We start by defining paths through
trees.

For a ranked alphabet $\Delta$, the \emp{path alphabet} $\pi(\Delta)$ is the
(nonranked) alphabet $\Delta_0 \cup \{(\sigma,i) \mid \sigma \in \Delta_k$ and
$1 \leq i \leq k$, for some $k \geq 1\}$. We will also write $\sigma_i$ rather
than $(\sigma,i)$. Intuitively, $\sigma_i$ denotes the fact that the path goes
to the $i$-th son of a node labeled $\sigma$. For every $t \in T_\Delta$, the
\emp{set of paths through $t$}, denoted $\pi(t)$, is the finite subset of
$\pi(\Delta)^*$ defined
as follows: (1) for $\sigma\in\Delta_0$, $\pi(\sigma) = \{\sigma\}$, and 
(2) for $\sigma\in\Delta_k$ $(k \geq 1)$ and $t_1,\ldots,t_k \in T_\Delta$,
$\pi(\sigma t_1\cdots t_k) = \cup\{(\sigma,i)\cdot \pi(t_i) \mid$
$1 \leq i \leq k\}$. Thus, $\pi(t)$ contains all paths from the root of $t$ to
its leaves (coded as strings). For a tree language $L$, $\pi(L) =
\{\pi(t) \mid t \in L\}$ is its path language.

As an example, let $\Delta_3 = \{a\}$, $\Delta_2 = \{b\}$, and
$\Delta_0 = \{c,d,\eps\}$, and let $t = abcdb\eps c\eps =
a(b(c,d),b(\eps,c),\eps)$. Then
$\pi(t) = \{a_1 b_1 c, \,a_1 b_2 d, \,a_2 b_1 \eps, \,a_2 b_2 c, \,a_3\eps\}$.
Note that $\yield(t) = cdc$ (because $\yield(\eps) = \lambda$, by
convention).

\begin{df}
 Let $\Delta$ be a ranked alphabet. For a (string)
language $L$, $\treeD(L)$ is the tree language $\treeD(L) = \{t \in
T_\Delta \mid \pi(t) \subseteq L\}$, and $\delta_\Delta(L)$ is the (string)
language $\delta_\Delta(L) = \yield(\treeD(L)) = \{w \in \Delta^*_0 \mid w =
\yield(t)$ for some $t \in T_\Delta$ such that $\pi(t) \subseteq L\}$. We call
$\delta_\Delta$ a \emp{delta operation}.\QEDB
\end{df}

For a class $K$ of (string) languages tree($K$) denotes
$\{\treeD(L) \mid L \in$ $K$, $\Delta$ is a ranked alphabet$\}$, and
$\delta(K)$ denotes $\yield$(tree($K$)), i.e., $\{\delta_\Delta(L) \mid L \in$
$K$, $\Delta$ is a ranked alphabet$\}$.

Thus, $\treeD(L)$ is obtained by taking paths from the
building kit $L$ and glueing them together to trees over $\Delta$. One
argument for the ``naturalness'' of the delta operations is that
they are continuous in the following sense: for every language $L$, 
$\delta_\Delta(L) = \cup\{\delta_\Delta(F) \mid F$ is a finite subset of $L\}$. 
All the full AFL operations are also continuous in this sense (think
for instance of Kleene $*$).

We use the symbol $\delta$ because its capital version $\Delta$ looks
like a (mathematical) tree, see Figs.~2, 6, and 8.
Let us consider some examples of $\delta_\Delta(L)$.

\begin{ex}
(1) Let $\Delta_3 = \{c\}$ and $\Delta_0 = \{a,b,\eps\}$. Consider the
regular language $L = c_2^*(c_1 a \cup c_2\eps \cup c_3 b)$ over
$\pi(\Delta)$. A simple argument shows that the trees that can be put together
from the paths in $L$, are those that have a ``spine'' of $c$'s, ending in an
$\eps$, with one $a$ and one $b$ sticking out of each $c$, see Fig.~7(1).
Thus $\treeD(L)$ is the tree language generated by the regular tree
grammar with rules $A \rightarrow caAb$, and $A \rightarrow ca\eps b$. Hence
$\delta_\Delta(L)= \yield(\treeD(L)) = \{a^n b^n \mid n \geq 1\}$. Thus
$\delta_\Delta$ transforms a regular
language into a (nonregular) context-free language.

(2) Let $\Delta_2 = \{b\}$, $\Delta_1 = \{c\}$, and $\Delta_0 = \{a\}$. Consider the
(deterministic) context-free language 
$L = \{c^n_1 wa \mid w \in \{b_1,b_2\}^n, n \geq 0\}$. 
The trees in $\treeD(L)$ consist of a monadic
``handle'' of, say, $n$ $c$'s, connected with a full binary tree of $b$'s of
depth $n$ (with $a$'s at the leaves), see Fig.~7(2). This tree language
can be generated by the $\RT(\Pd)$ grammar with rules 
$A \rightarrow cA(\push(c))$, $A \rightarrow B(\stay)$, 
$B \rightarrow \ruleif \ttop=c \rulethen bB(\pop)B(\pop) \ruleelse a$. 
Hence $\delta_\Delta(L)$ is the indexed language $\{a^{2^n} \mid n \geq 0\}$. 

(3) This example is similar to the one in (2). Let
$\Delta_2 = \{f,d,b\}$, $\Delta_1 = \{c\}$, and $\Delta_0 = \{p,q,r,s,a\}$.
Consider the context-free language $L = \{f_2 s\} \cup L_1\cup L_2$ with 
\begin{quote}
$L_1=\{f_1c_1^n d_1 wp \mid w \in \ \{b_1,b_2\}^n, \,n \geq 0\}$ $\quad$ and 
\end{quote} 
\begin{quote}
$L_2=\{f_1 c_1^n d_2 wx \mid w \in \{b_1,b_2\}^n, \,x \in \{q,r\}, \,n \geq 0\}$.
\end{quote} 
For a tree in $\treeD(L)$, see Fig.~7(3). It should now be clear that
$\delta_\Delta(L)$ is
the indexed language $\{p^{2^n}ws \mid w \in \{q,r\}^{2^n}, n \geq 0\}$.
Now suppose that $p,q,r,s$ really denote the symbols $c_1,b_1,b_2,a$,
respectively. Then, $\delta_\Delta(\delta_\Delta(L)) = \{a^{d(n)} \mid n \geq 0\}$, 
where $d(n) = 2^{2^n}$. This language can be
generated by a $\CF(\Pd^2)$ grammar. \QEDB
\end{ex}

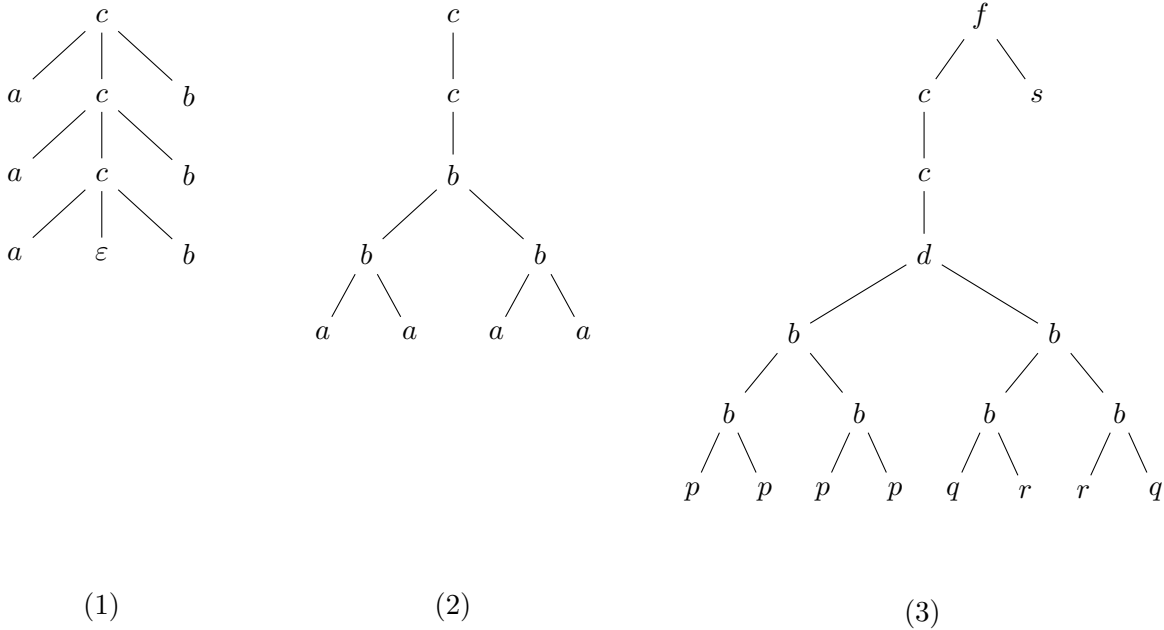
\begin{figure}[h]
\centering
\begin{tikzpicture}
\tikzset{level distance=30pt}
\node(N1) [sibling distance = 3em] {$c$}
	child [sibling distance = 3em] {node {$a$}}
	child [sibling distance = 3em] {node {$c$}
		child [sibling distance = 3em] {node {$a$}}
		child [sibling distance = 3em] {node {$c$}
			child {node {$a$}}
			child {node {$\varepsilon$}}
			child {node {$b$}}
		}
		child {node {$b$}}
	}
	child [sibling distance = 3em] {node {$b$}};

\node(N2)[right=11em of N1] {$c$}
	child {node {$c$}
		child {node {$b$}
			child [sibling distance = 6em] {node {$b$}
				child [sibling distance = 3em] {node {$a$}}
				child [sibling distance = 3em] {node {$a$}}
			}
			child [sibling distance = 6em]{node {$b$}
				child [sibling distance = 3em] {node {$a$}}
				child [sibling distance = 3em] {node {$a$}}
			}
		}
	};

\node(N3)[right=17em of N2] {$f$}
	child {node {$c$}
		child {node {$c$}
			child {node {$d$}
				child [sibling distance = 9em]{node {$b$}
					child [sibling distance = 4.5em]{node {$b$}
						child [sibling distance = 2.5em]{node {$p$}}
						child [sibling distance = 2.5em]{node {$p$}}
					}
					child [sibling distance = 4.5em]{node {$b$}
						child [sibling distance = 2.5em]{node {$p$}}
						child [sibling distance = 2.5em]{node {$p$}}
					}
				}
				child [sibling distance = 9em]{node {$b$}
					child [sibling distance = 4.5em]{node {$b$}
						child [sibling distance = 2.5em]{node {$q$}}
						child [sibling distance = 2.5em]{node {$r$}}
					}
					child [sibling distance = 4.5em]{node {$b$}
						child [sibling distance = 2.5em]{node {$r$}}
						child [sibling distance = 2.5em]{node {$q$}}
					}
				}
			}
		}
	}
	child {node {$s$}};

\node[align=center, below=19em of N1](M1)  {(1)};
\node[align=center, below=19em of N2](M1)  {(2)};
\node[align=center, below=19em of N3,xshift=-2em](M1)  {(3)};

\end{tikzpicture}
 \caption{Trees that illustrate tree$_\Delta$ and $\delta_\Delta$.}
\end{figure}

We now show the announced result. The basic idea
involved is that a tree language that can be recognized by a
deterministic top-down $S$ tree automaton, is completely
determined by its path language.

\begin{theo}\label{theo8.1}
 Let $S$ be an r-acceptor deterministic storage type.
\begin{itemize}
\setlength{\itemsep}{0pt}
 \item[\rm(1)] $\lambdaCF(S) = \yield(\lambdaDfRT(S))= \yield(\lambdaRT(S))$
 \item[\rm(2)] $\lambdaDfRT(S) = \tree(\lambdaDfREG(S))$
 \item[\rm(3)] $\lambdaCF(S) = \delta(\lambdaDfREG(S))$.
\end{itemize}
\end{theo}

\begin{proof}
 (3) follows immediately from (1) and (2).

 (1) $\hspace{1mm}$ First, $\yield(\lambdaRT(S)) \subseteq \lambdaCF(S)$: replace
every rule $A \rightarrow \ruleif b \rulethen t$ by the rule $A \rightarrow
\ruleif b \rulethen \yield(t)$. Second, we show that $\lambdaCF(S) \subseteq
\yield(\lambdaDeRT(S))$, see Lemma~\ref{DeRT-DfRT}. Let $\G = (N,e,\Delta,\Ain,R)$
be a $\CF(S)$ grammar. We construct an $\RT(S)$ r-acceptor $\G'$ that
recognizes (deterministically) the derivation trees of $\G$: to
make this possible, we assume that each internal node of a
derivation tree is labeled with the rule applied at that node;
leaves are labeled with terminals as usual (or with $\eps$, to denote
$\lambda$). Formally, $\G' = (N',e,\Sigma,\Ain,R')$, where $N' = N \cup
\{A_\sigma \mid \sigma\in \Delta \cup \{\eps\}\}$, and $\Sigma$ and $R'$ are
determined as follows.
The symbols of rank~0 in $\Sigma$ are $\Sigma_0 = \Delta \cup \{\eps\}$, 
where $\eps$ is the symbol with $\yield(\eps) = \lambda$. 
For each rule $r: A \rightarrow \ruleif b \rulethen \xi$ of $R$,
with $\xi\neq \lambda$, we put a symbol $\hat{r}$ of
rank $|\xi|$ in $\Sigma$, and we put the rule $A \rightarrow \ruleif b \rulethen
\hat{r}\xi'$ in $R'$, where $\xi'$ is the result of replacing, in $\xi$, every
$\sigma\in\Delta$ by $A_\sigma(\id)$. Similarly, for each rule $r: A \rightarrow
\ruleif b \rulethen \lambda$ of $R$, we put a symbol $\hat{r}$ of rank $1$ in
$\Sigma$, and put the rule $A \rightarrow \ruleif b \rulethen
\hat{r}A_\eps(\id)$ in $R'$. Finally, $R'$ contains all rules $A_\sigma
\rightarrow \sigma$ with $\sigma\in\Delta\cup \{\eps\}$. Clearly, $\G'$ is
r-acceptor deterministic, and $L(\G')$ is the set of (rule labeled) derivation
trees of $\G$, i.e., $\yield(L(\G')) = L(\G)$.

(2) $\hspace{1mm}$ First we show that $\lambdaDfRT(S) \subseteq
\tree(\lambdaDeREG(S))$, see Lemma~\ref{DeREG-DfREG}. 
Let $L \subseteq T_\Delta$ be a tree language in
$\lambdaDfRT(S)$. Thus there is a deterministic $\RT(S)$ r-acceptor $\G =
(N,e,\Delta\#,\Ain,R)$ such that $L(\G)$ = mark($L$). We construct a
deterministic $\REG(S)$ r-acceptor $G'$ such that $\treeD(L(\G')) = L$.
The acceptor $\G'$ just imitates the behavior of $\G$, on the paths of the trees. Thus, $\G'
= (N \cup \overline{N},e,\pi(\Delta),\Ain,R')$, where $\overline{N} = \{\overline{A} \mid A \in
N\}$, and the rules of $R'$ are determined as follows.

\begin{itemize}
 \item If the rule $A \rightarrow \ruleif b \rulethen \sigma B_1(f_1)\cdots B_k(f_k)$ is
    in $R$, with $k \geq 2$, or $k = 1$ and $\sigma \notin \Delta_0$, \\[1mm]
    then the rules $A \rightarrow \ruleif b \rulethen \sigma_i B_i(f_i)$ are in $R'$, 
    for every $i$, $1 \leq i \leq k$.
 \item If the rule $A \rightarrow \ruleif b \rulethen \sigma B(f)$, with $\sigma \in
    \Delta_0$, is in $R$, \\[1mm]
    then the rule $A \rightarrow \ruleif b \rulethen \sigma\overline{B}(f)$ is in $R'$.
 \item If the rule $A \rightarrow \ruleif b \rulethen \#$ is in $R$, \\[1mm]
    then the rule $\overline{A} \rightarrow \ruleif b \rulethen \lambda$ is in $R'$. 
 \item If the rule $A \rightarrow \ruleif b \rulethen B(f)$ is in $R$, \\[1mm]
    then this rule and
    the rule $\overline{A} \rightarrow \ruleif b \rulethen \overline{B}(f)$ are in $R'$.
\end{itemize}
This concludes the construction of $\G'$. Due to the use of bars,
$\G'$ is r-acceptor deterministic. It is obvious that if $t \in L$,
then $\pi(t) \subseteq L(\G')$, i.e., all paths of $t$ are accepted by $\G'$.
On the other hand, if $t$ is a tree over $\Delta$ such that all its paths
are accepted by $\G'$, then mark($t$) is accepted by $\G$, due to the
determinism of $\G$, and the fact that all computations of $\G'$ (on
all paths) start in the same initial configuration $m(e)(u_0)$,
where $I = \{u_0\}$. (We note that $G'$ does not accept the language
$\pi(L)$; $L(\G')$ may contain paths that are not in $\pi(L)$. Acceptance
of $\pi(L)$ can be realized by look-ahead, i.e., by a deterministic
$\REG(S_\LA)$ r-acceptor, see Lemma~5.2 of \cite{Vog3}.)

Next we show that $\tree(\lambdaDfREG(S)) \subseteq \lambdaDfRT(S)$.
Let $L \subseteq \Sigma^*$ be a language in \mbox{$\lambdaDfREG(S)$}, and let
$\Delta$ be a ranked alphabet. Consider $L' = L \cap
(\Sigma-\Delta_0)^*\Delta_0$. Since $(\Sigma-\Delta_0)^*\Delta_0$ is a
regular language, $L'$ is still in $\lambdaDfREG(S)$, by the usual
product construction. Moreover $\treeD(L') = \treeD(L)$. Now note
that $L'$ is prefix-free. Hence $L'$ is in $\lambdaDeREG(S)$, cf. the
remark after Lemma~\ref{DeREG-DfREG}.

Thus, it suffices to show that $\tree(\lambdaDeREG(S)) \subseteq
\lambdaDfRT(S)$. Let $\G = (N,e,\Sigma,\Ain,R)$ be a deterministic
$\REG(S)$ r-acceptor, and let $\Delta$ be a ranked alphabet. As shown above, we
may assume that $L(\G) \subseteq (\Sigma-\Delta_0)^*\Delta_0$. We want to
construct a deterministic $\RT(S)$ r-acceptor $\G'$ such that $L(\G')$ =
mark($\treeD(L(\G)))$. The acceptor $\G'$ just simulates $\G$ on all paths of the input
tree. Thus, $\G' = (N\cup\{Q\},e,\Delta\#,\Ain,R')$, where $Q$ is a new
nonterminal, and $R'$ is defined as follows.

\begin{itemize}
 \item If, for $\sigma\in\Delta_k$ with $k \geq 1$, the rules $A \rightarrow
    \ruleif b_i \rulethen \sigma_i B_i(f_i$) are in $R$ for all $i$, $1 \leq i \leq k$, \\[1mm]
    then the rule $A \rightarrow \ruleif b_1\text{ \bool{and} }\cdots
    \text{ \bool{and} }b_k \rulethen \sigma B_1(f_1) \cdots B_k(f_k)$ is in $R'$.
 \item If, for $\sigma\in\Delta_0$, the rule $A \rightarrow \ruleif b \rulethen
    \sigma$ is in $R$, \\[1mm]
    then the rules $A \rightarrow \ruleif b \rulethen
    \sigma Q(\id)$ and $Q \rightarrow \#$ are in $R'$.
 \item If, for $\sigma \in \Delta_0$, the rule $A \rightarrow \ruleif b
    \rulethen \sigma B(f)$ is in $R$, \\[1mm]
    then it is also in $R'$.
 \item If the rule $A \rightarrow \ruleif b \rulethen \lambda$ is in $R$, \\[1mm]
    then the rule $A \rightarrow \ruleif b \rulethen \#$ is in $R'$.
 \item If the rule $A \rightarrow \ruleif b \rulethen B(f)$ is in $R$, \\[1mm]
    then it is also in $R'$.
\end{itemize}
This concludes the construction of $\G'$. It is left to the reader
to show that $\G'$ is indeed r-acceptor deterministic, and that
$L(\G')$ = mark($\treeD(L(\G)))$.
\end{proof}

The idea for this theorem came from the fact, proved by
Magidor and Moran, that \mbox{$\lambdaDeRT(S_0)$}, the class of tree languages
accepted by deterministic top-down finite tree automata, is
included in $\tree(\REG)$, see \cite{Tha}. See also \cite{Cou} and Section
II.11 of \cite{GecSte}; a tree language $L$ over $\Delta$ is said to be
\emp{closed} if $L = \treeD(\pi(L))$; it is easy to see that, for every
string language $L$, $\treeD(L)$ is closed; hence Theorem \ref{theo8.1}(2)
shows that all tree languages in $\lambdaDfRT(S)$ are closed.

The main application of Theorem \ref{theo8.1} is to the iterated
pushdown languages: the class of languages accepted by the
nondeterministic one-way $(n+1)$-iterated pushdown automata can be
obtained by the delta operations from the class of languages
accepted by the deterministic $n$-iterated pushdown automata.

\begin{theo}\label{theo8.2}
For every $n\geq 0$, $\lambdaREG(\Pd^{n+1}) = \delta(\lambdaDfREG(\Pd^n))$.\\
In particular, $\CF = \delta(\REG)$, and $\mathrm{Indexed} = \delta(\DCF)$.
\end{theo}

\begin{proof}
 By Theorem \ref{theo5.1}(2), $\lambdaREG(\Pd^{n+1}) = \lambdaCF(\Pd^n)$, and by
Theorem \ref{theo8.1}(3) $\lambdaCF(\Pd^n) = \delta(\lambdaDfREG(\Pd^n))$. 
Note that finite automata can be made deterministic; hence
$\lambdaDfREG(S_0) = \REG$.
\end{proof}

Of course, we also get from Theorem \ref{theo8.1}(2) that
$\lambdaDfRT(\Pd^n) = \tree(\lambdaDfREG(\Pd^n))$. Using (effective)
closure of $\Pd^n$ under look-ahead, one could rather easily prove that the
equivalence problems for $\lambdaDfRT(\Pd^n)$ and for
$\lambdaDfREG(\Pd^n)$ are equivalent (see \cite{Cou}); note that the
decidability of these problems is open.\footnote{{\bf New Observation.} 
For $n=1$ this problem (the dpda equivalence problem) was solved 
by S\'enizergues in \cite{*Sen1} (see also \cite{*Sen2}).
}

Another consequence of Theorem \ref{theo8.1}(2) is that the
languages accepted by alternating one-way $S$ automata can be
expressed in terms of the languages accepted by deterministic
one-way $S$ automata. In fact, as observed at the end of Section \ref{sect4},
the class 
$\alphaCF$(One-way$_{\id}\times S)$ equals
$\tauRT$(One-way$_{\id})^{-1}(\lambdaRT(S))$. 
Now it is easy
to see from the proof of this statement in \cite{Eng6, DamGue}, 
that it even equals 
$\tauRT$(One-way$_{\id})^{-1}(\lambdaDfRT(S))$. 
Hence, by Theorem~\ref{theo8.1}(2),
\begin{quote}
$\alphaCF$(One-way$_{\id} \times S) =
\tauRT$(One-way$_{\id})^{-1}(\tree(\lambdaDfREG(S)))$.
\end{quote}
Thus, e.g., the class of alternating one-way $\Pd^n$ languages can be
expressed in terms of \mbox{$\lambdaDfREG(\Pd^n)$}; note that this is the class
DTIME($\exp_k(cn))$ of $k$-iterated exponential time languages \cite{Eng9}.

To illustrate that determinism plays an essential role
in Theorems \ref{theo8.1} and \ref{theo8.2}, we show that $\delta$(LIN) = RE,
where LIN is the class of linear context-free languages.

\begin{theo}\label{theo8.3}
  $\delta\mathrm{(LIN) = RE}$.
\end{theo}

\begin{proof}
 Obviously $\delta$(LIN) $\subseteq$ RE. We now prove that RE $\subseteq
\delta$(LIN).
Every recursively enumerable language is of the form $h(L \cap M)$,
where $L$ and $M$ are linear context-free languages, over some
alphabet $\Sigma$, and $h$ is a homomorphism from $\Sigma$ to some alphabet
$\Omega$ (see, e.g., \cite{EngRoz}). We now construct a linear context-free
language $K$ and a ranked alphabet $\Delta$, such that
$\delta_\Delta(K) = h(L \cap M)$.
In fact $\treeD(K)$ consists of all trees of the form suggested in
Fig.~8, with $a_1 a_2\cdots a_n \in L \cap M$ (and $a_i \in\Sigma$).
Let $\Delta$ be the ranked alphabet with $\Delta_0 = \Omega \cup \{\eps\}$,
$\Delta_1 = \{\#_1\}$, $\Delta_2 = \Sigma \cup \{\$,\#_2\}$, and
$\Delta_k = \{\#_k\}$ for $3 \leq k \leq m$, where $m = \max\{|h(a)| \mid a
\in\Sigma\}$. We define $K = L_2(\$,1)\eps \cup
M_2(\$,2)\eps \cup R$, where $L_2 =
\{(a_1,2)(a_2,2)\cdots(a_n,2) \mid a_1 a_2\cdots a_n \in L\}$, and similarly
for $M_2$. Moreover, $R$
is the regular language $\cup\{\Sigma^*_2 F_a \mid a \in\Sigma\}$, where
$\Sigma_2 = \{(a,2) \mid a\in\Sigma\}$, and
\begin{itemize}
\setlength{\itemsep}{0pt}
 \item if $h(a) = b_1\cdots b_k$ with $k \geq 1$ and $b_i \in \Omega$, then
    $F_a= \{(a,1)(\#_k,1)b_1,\ldots,(a,1)(\#_k,k)b_k\}$,
 \item if $h(a) = \lambda$, then $F_a = \{(a,1)(\#_1,1)\eps\}$.
\end{itemize}

\noindent Clearly $K$ is a linear context-free language, and
$\delta_\Delta(K) = h(L \cap M)$.
\end{proof}

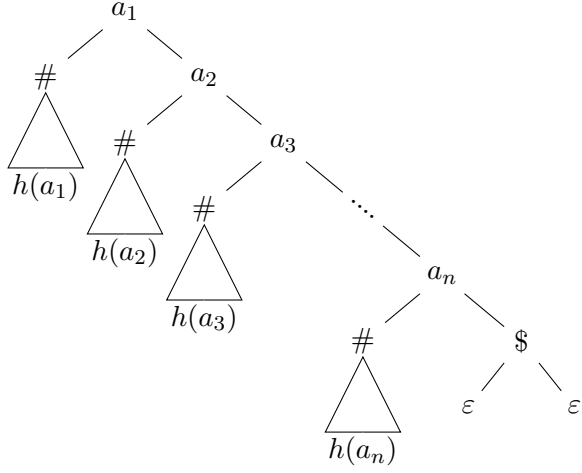
\begin{figure}[h]
\centering
\begin{tikzpicture}[
leaf/.style={isosceles triangle,draw,shape border rotate=90,isosceles triangle stretches=true, minimum height=10mm,minimum width=10mm,inner sep=0,yshift={-9mm},font=\tiny},
te/.style={inner sep=0,yshift={3.2mm}}]

\node {$a_1$}[sibling distance = 60pt, level distance=25pt]
child {node {$\#$}{node[leaf]{} child [draw=white] {node[te] {$h(a_1)$}}}}
child {node {$a_2$}
child {node {$\#$ } {node[leaf]{}child [draw=white] {node[te] {$h(a_2)$}}}}
child {node {$a_3$}
child {node {$\#$}{node[leaf]{}child [draw=white] {node[te] {$h(a_3)$}}}}
child {node[rotate=-39.5] {$....$}
child [draw=white] {node {}}
child {node {$a_n$}
child {node {$\#$}{node[leaf]{}child [draw=white] {node[te] {$h(a_n)$}}}}
child {node {$\$$}[sibling distance = 40pt]
child {node {$\varepsilon$}}
child {node {$\varepsilon$}}
}
}
}
}
};
\end{tikzpicture}
 \caption{Illustration of $\delta$(LIN) = RE.}
\end{figure}

This implies that the linear context-free languages are
not contained in $\lambdaDfREG(\Pd^n)$, as shown, in a different way, in
\cite{EngVog4}. In fact, if LIN $\subseteq \lambdaDfREG(\Pd^n)$, then RE
= $\delta$(LIN) $\subseteq
\delta(\lambdaDfREG(\Pd^n)) = \lambdaREG(\Pd^{n+1})$. However, the
languages in $\lambdaREG(\Pd^n)$ are known to be recursive \cite{Dam,Eng9}.

In \cite{Gre} an example is given of a (r-acceptor
deterministic) storage type $\widetilde{S}$ such that \mbox{$\lambdaREG(\widetilde{S}) = \CF$} and 
$\lambdaREG(\Pd(\widetilde{S}))$ = RE. In our notation, $\widetilde{S} =
(C,P,F,I,E,m)$ with $C = \Omega^* \cup \{\bot\}$, where $\Omega$ is a fixed
infinite set of symbols and $\bot \notin \Omega$, $I = \{u_0\}$ and $E =
\{e_0\}$ with $m(e_0)(u_0) = \lambda$, $P = \emptyset$, $F= \{$write$(v) \mid v
\in \Omega^*\} \cup \{\test(L) \mid L \in \CF\}$, and
for $w \in \Omega^*$, $m($write$(v))(w) = wv$,
$m($write$(v))(\bot)= \bot$, $m(\test(L))(w) = \bot$ if $w \in L$ and
undefined otherwise. Thus the
storage of $\widetilde{S}$ is that of an ordinary write-only output tape, and
the content of that tape may be tested once for membership in a
context-free language. Hence, for $L \in \CF$, the $\REG(\widetilde{S})$ grammar
with rules $A \rightarrow aA($write$(a))$ for all $a$, $A \rightarrow
B(\test(L))$, $B \rightarrow \lambda$,
generates $L$, and so $\CF \subseteq \lambdaREG(\widetilde{S}$). By standard
AFA/AFL techniques it can be shown that $\lambdaREG(\widetilde{S}) \subseteq
\CF$. The intersection $L_1 \cap L_2$ of two context-free languages $L_1$
and $L_2$ can be generated by the $\CF(\widetilde{S})$ grammar with rules $A
\rightarrow aA($write$(a))$ for all $a$, $A \rightarrow
B(\test(L_1))B(\test(L_2))$, $B\rightarrow\lambda$. Hence, since every
recursively enumerable language is the homomorphic image of such
an intersection, and $\lambdaCF(\widetilde{S})$ is closed under homomorphisms
(it is a full trio), RE $\subseteq \lambdaCF(\widetilde{S}) =
\lambdaREG(\Pd(\widetilde{S}))$.

Thus, as noted in \cite{Gre}, the following desirable
property of $\Pd(S)$, and thus of $\lambdaCF(S)$, is \emp{not true}:
\begin{quote}
($*$) $\quad$ if $\lambdaREG(S_1) = \lambdaREG(S_2)$ then
$\lambdaREG(\Pd(S_1)) = \lambdaREG(\Pd(S_2))$.
\end{quote}
Indeed, for $S_1 = \widetilde{S}$ and $S_2 = \Pd$, we get
$\lambdaREG(\Pd(S_1))$ = RE and
$\lambdaREG(\Pd(S_2)) = \lambdaCF(\Pd)$ = Indexed $\subseteq$ the class
of recursive languages. 
Property ($*$) holds in fact for most other operations
$O(S)$ on storage types, such as the well-nested AFA, because
there exists a class $F$ of operations on languages such that
$\lambdaREG(O(S)) = F(\lambdaREG(S)$), see \cite{Gin, Gre} (for well-nested AFA,
$F$ is the class of nested iterated substitutions). From
Theorems \ref{theo5.1}(2) and \ref{theo8.1}(3) it follows, for an r-acceptor
deterministic storage type $S$, that the analogous equality
$\lambdaREG(\Pd(S))= $ \mbox{$\delta(\lambdaDfREG(S))$} holds. Hence the
following property of $\Pd(S)$, for r-acceptor deterministic storage types, is
\emp{true}:
\begin{quote}
($**$) $\quad$ if $\lambdaDfREG(S_1) = \lambdaDfREG(S_2)$ then
$\lambdaREG(\Pd(S_1)) = \lambdaREG(\Pd(S_2))$.
\end{quote}
It can probably even be shown that: if $\lambdaDfREG(S_1) =
\lambdaDfREG(S_2)$ then $\lambdaDfREG(\Pd(S_1)) =$
$\lambdaDfREG(\Pd(S_2))$. This would suggest to define two
storage types $S_1$ and $S_2$ to be \emp{equivalent} if 
$\lambdaDfREG(S_1) = \lambdaDfREG(S_2)$,
and use this notion of equivalence, rather than the, more
structural, one used in \cite{Eng9, EngVog2,EngVog3}, as the basic notion
of indistinguishability of storage types. In this respect it
would be nice to have a class $\delta'$ of operations such that
$\lambdaDfREG(\Pd(S)) = \delta'(\lambdaDfREG(S))$.

Some remaining questions in this section are the following.
\begin{itemize}
\setlength{\itemsep}{0pt}
 \item Is there a (natural) class $F$ of operations such that
    $F(\lambdaREG(\Pd^n)) = \lambdaREG(\Pd^{n+1})$?
 \item Is there a formal relationship between $\delta$ and YIELD? Note that
    repeated application of YIELD to RT gives the IO-hierarchy,
    whereas, apart from the restriction to determinism, repeated
    application of $\delta$ to RT gives the OI-hierarchy
    $\{\lambdaREG(\Pd^n)\}_n$ (see \cite{Dam, Eng6}).
 \item Is there a ``delta theorem'' (i.e., a theorem analogous to
    Theorem \ref{theo8.1}) for arbitrary storage types?
\end{itemize}


\newpage
\noindent
\textbf{Acknowledgments.}
I thank Peter Asveld, Hendrik Jan Hoogeboom,
and Jetty Kleijn for their comments on a previous version of
this paper. I am grateful to Werner Damm, Gilberto File, and
Heiko Vogler for many stimulating conversations. Finally, I
thank Heiko Vogler again, for our long co-operation in finding
all the appropriate definitions.

\vspace{1cm}
\noindent
In the following list, new references (i.e., references added in 2014) are starred.


\bibliographystyle{alpha}

\end{document}